\documentclass[12pt]{article}
\usepackage{subcaption}
\usepackage{changepage}
\usepackage[dvipsnames]{xcolor}
\usepackage{xr}
\usepackage[english]{babel}
\usepackage{amsfonts}
\usepackage{amssymb}
\usepackage{bm}
\usepackage{amsmath}
\usepackage{adjustbox}
\usepackage{threeparttable}
\usepackage{scrextend}

\usepackage{relsize}
\usepackage{booktabs}	
\usepackage[normalem]{ulem}
\usepackage{enumerate}
\usepackage{bbm}
\usepackage{mathtools}
\usepackage{setspace}
\usepackage{graphicx}
\usepackage{fullpage}

\usepackage{soul}
\usepackage{xspace}
\usepackage{csquotes}
\usepackage[left=1.8cm,right=1.8cm,top=2cm,bottom=2cm]{geometry}
\usepackage{placeins}
\usepackage[footfont=normalsize,font=normalsize]{floatrow}
\usepackage{lscape}
\usepackage{longtable}
\usepackage{mathabx}
\usepackage{accents}
\usepackage{float}
\usepackage{fullwidth}
\usepackage{verbatim}
\usepackage{autobreak}
\usepackage[bottom]{footmisc} 
\usepackage[titletoc,title, toc,page]{appendix}
\usepackage{tikz}
\usetikzlibrary{calc,intersections,shapes,decorations,fit,positioning,arrows, arrows.meta}
\usetikzlibrary{patterns}
\usetikzlibrary{decorations.pathreplacing}

\usepackage{microtype}
\usepackage{colortbl}

	\usepackage[longnamesfirst]{natbib}
	\usepackage{bibunits}
	\defaultbibliography{Bibliography.bib}  
	\defaultbibliographystyle{aer}

\usepackage{hyperref}
\definecolor{cite_blue}{rgb}{0.0, 0.1, 0.3}
	
	\hypersetup{
    colorlinks=true,
    linkcolor=cite_blue,
    filecolor=magenta,      
    urlcolor=cite_blue,
    citecolor = cite_blue
}
	\usepackage{cleveref}

    \usepackage{rotating}
	\usepackage[framemethod=tikz]{mdframed}

    \usepackage{tcolorbox}
    \newtcbox{\feedback}{nobeforeafter,colframe=black,colback=white,boxrule=0.5pt,arc=2pt,
      boxsep=0pt,left=2pt,right=2pt,top=2pt,bottom=2pt,tcbox raise base}
      
    \usepackage{amsthm}
    
    \newtheorem{theorem}{Theorem}[section]
    \newtheorem{assumption}{Assumption}

    \newtheorem{lemma}{Lemma}[section]
    \newtheorem{corollary}{Corollary}[section]

    \theoremstyle{definition}
    \newtheorem{remark}{Remark}[section]

\usepackage{ragged2e}
\newlength\ubwidth

\newtheorem{innercustomass}{Assumption}
\newenvironment{namedassumption}[1] {
    \renewcommand\theinnercustomass{#1}
    \innercustomass 
} {\endinnercustomass}

\allowdisplaybreaks

\definecolor{gold}{RGB}{162,132 ,72}
\definecolor{green}{RGB}{63, 157, 0}
\definecolor{pink}{RGB}{255,0,222}
\definecolor{purple}{RGB}{145,0,255}
\definecolor{steelblue}{RGB}{70,130,180}
\definecolor{navy}{RGB}{0,0,128}
\definecolor{red}{rgb}{1, 0, 0}
\definecolor{orange}{RGB}{240,124,8}
\definecolor{darkgreen}{rgb}{0.0, 0.2, 0.13}
\definecolor{darkblue}{rgb}{0.0, 0.0, 0.55}
\definecolor{cspurple}{RGB}{63, 44, 110}
\definecolor{cspink}{RGB}{194, 15, 90}
\definecolor{csblue}{HTML}{0170b9}
\definecolor{amazonorange}{RGB}{255, 153, 0}
\definecolor{amazonblack}{RGB}{9, 26, 34}
\definecolor{amazongrey}{RGB}{35, 47, 62}

\definecolor{emoryblue}{RGB}{1, 33, 105}
\definecolor{emoryyellow}{RGB}{242, 169, 0}
\definecolor{emorygold}{RGB}{181, 133, 0}
\definecolor{emorymgold}{RGB}{132, 117, 78}
\definecolor{emorylblue}{RGB}{0, 125, 186}
\definecolor{emorymblue}{RGB}{0, 51, 160}

\definecolor{yaleblue}{HTML}{00356b}
\definecolor{yalelightblue}{HTML}{63aaff}
 \definecolor{yalemblue}{HTML}{286dc0}
\definecolor{cslightpink1}{HTML}{f0c3d6}
\definecolor{lblue}{HTML}{f0f8ff}

\tikzset{
    -Latex,auto,node distance =1 cm and 1 cm,semithick,
    state/.style ={ellipse, draw, minimum width = 0.7 cm},
    point/.style = {circle, draw, inner sep=0.04cm,fill,node contents={}},
    bidirected/.style={Latex-Latex,dashed},
    el/.style = {inner sep=2pt, align=left, sloped}
}
\renewcommand{\P}{\mathop{}\!\mathbb{P}}

\newcommand{\E}{\mathop{}\!\mathbb{E}}

\newcommand{\Cov}{\mathop{}\!\textnormal{Cov}}
\newcommand{\aCov}{\mathop{}\!\textnormal{aCov}}

\newcommand{\Var}{\mathop{}\!\textnormal{Var}}

\DeclareMathOperator*{\argmin}{arg\,min}

\renewcommand{\j}{\mathbf{j}}

\makeatletter
\newcommand*{\indep}{%
  \mathbin{%
    \mathpalette{\@indep}{}%
  }%
}
\newcommand*{\nindep}{%
  \mathbin{
    \mathpalette{\@indep}{/}%
  }%
}
\newcommand*{\@indep}[2]{%
  \sbox0{$#1\perp\m@th$}
  \sbox2{$#1=$}
  \sbox4{$#1\vcenter{}$}
  \rlap{\copy0}
  \dimen@=\dimexpr\ht2-\ht4-.2pt\relax
  \kern\dimen@
  \ifx\\#2\\%
  \else
    \hbox to \wd2{\hss$#1#2\m@th$\hss}%
    \kern-\wd2 %
  \fi
  \kern\dimen@
  \copy0 
}
\makeatother

\makeatletter
\newcommand*{\addFileDependency}[1]{
  \typeout{(#1)}
  \@addtofilelist{#1}
  \IfFileExists{#1}{}{\typeout{No file #1.}}
}
\makeatother

\newcommand{\tprime}{t_{\text{pre}}}
\newcommand{\tdoubleprime}{t_{\text{pre}}'}

\numberwithin{equation}{section}

\usepackage{subcaption}
\begin{document}

\setlength{\abovedisplayskip}{4pt}
\setlength{\belowdisplayskip}{3pt}
\setlength{\abovedisplayshortskip}{4pt}
\setlength{\belowdisplayshortskip}{3pt}

\title{Efficient Difference-in-Differences and Event Study Estimators\footnote{Authors are listed in alphabetical order. We are grateful to Dmitry Arkhangelsky, Gregorio Caetano, Kevin Chen, Andrew Goodman-Bacon, Yingying Li, Jann Spiess, Liyang Sun, Sarah Vicol, Davide Viviano, Stefan Wager, Jeff Wooldridge, and the audience at several conferences (ESAM2024, CES2025, SETA2025,  STATA Texas Empirical Micro Conference) and workshops (NYFed, Stanford, UPenn, BU, Emory, GSM/Peking Univ, City Univ of Hong Kong, Rice, U Carlos III de Madrid) for their helpful comments.}} 
\author{Xiaohong Chen\thanks{Yale University. E-mail
		xiaohong.chen@yale.edu} \and
        Pedro H. C. Sant'Anna\thanks{Emory University.
		E-mail: pedro.santanna@emory.edu} \and
  Haitian Xie\thanks{Peking University. E-mail: xht@gsm.pku.edu.cn}
  }
\date{\today}

\maketitle
\thispagestyle{empty}
\begin{abstract}

This paper investigates efficient Difference-in-Differences (DiD) and Event Study (ES) estimation using short panel data sets within the heterogeneous treatment effect framework, free from parametric functional form assumptions and allowing for variation in treatment timing. We provide an equivalent characterization of the DiD potential outcome model using sequential conditional moment restrictions on observables, which shows that the DiD identification assumptions typically imply nonparametric overidentification restrictions. We derive the semiparametric efficient influence function (EIF) in closed form for DiD and ES causal parameters under commonly imposed parallel trends assumptions. The EIF is automatically Neyman orthogonal and yields the smallest variance among all asymptotically normal, regular estimators of the DiD and ES parameters. Leveraging the EIF, we propose simple-to-compute efficient estimators. Our results highlight how to optimally explore different pre-treatment periods and comparison groups to obtain the tightest (asymptotic) confidence intervals, offering practical tools for improving inference in modern DiD and ES applications even in small samples. Calibrated simulations and an empirical application demonstrate substantial precision gains of our efficient estimators in finite samples. 

\end{abstract}

\newpage

\setcounter{page}{1}

\section{Introduction}

Difference-in-Differences (DiD) and Event Study (ES) designs are among the most widely used empirical strategies in economics and related fields. For instance, recent data indicates that over 30\% of 2024 NBER applied microeconomics working papers mention DiD or ES---more than any other causal inference method,\footnote{The NBER data used in \citet{Goldsmith-Pinkham2024} is up to May 2024.} and their popularity is also rapidly expanding in empirical macroeconomics and finance \citep{Goldsmith-Pinkham2024}. Although recent methodological advances have improved the robustness of DiD estimators to treatment effect heterogeneity,\footnote{See \citet{Roth2023_JoE_Review}, \citet{deChaisemartin2022b_survey} and \citet{Baker_etal_DiD_practioner_guide_2025} for recent reviews of DiD advances and a DiD practitioner's guide. See also \citet{Baker2022} for an overview with application in finance.} several empirically relevant econometric questions remain underexplored.

First, modern DiD and ES estimators are designed to accommodate rich treatment effect 
heterogeneity using short panel data.
Unfortunately, most of them can have wide confidence intervals in applications with limited sample size.\footnote{See, e.g., \citet{Chiu2023_replicationDiD}, \citet{Weiss2024_DiD_power}, and \citet{Lal2025_TWFE} for a discussion.} Whether the flexibility of modern DiD necessarily entails a meaningful loss in power compared to more restrictive estimators, such as two-way fixed effects (TWFE), remains an open question. Second, most heterogeneity-robust DiD implementations in short panels implicitly weight all pre-treatment periods equally or discard them entirely. These assumptions are generally unsupported by theory, data or subject-specific knowledge, and often yield imprecise estimators. Third, although it has become a common empirical practice to report several DiD and ES estimators,\footnote{See, e.g., \citet{Braghieri2022_facebook}, \citet{Hansen_Wingender_2023_AERI}, \citet{Arold_2024_QJE}, and \citet{Mast2024_Restat}.} there is limited formal guidance on how to compare them, especially so when applied researchers are unwilling to impose strong, arbitrary restrictions on temporal dependence or treatment effect heterogeneity. Sometimes, even basic questions such as whether different DiD estimators target the same causal parameters or if they rely on similar identification assumptions can be difficult to answer.

In this paper, we address these challenges by developing a unified framework for semiparametrically efficient estimation of DiD and ES causal parameters under point-identification assumptions---namely, parallel trends and no anticipation. We allow for several commonly-used DiD designs, including designs with a single treatment date and staggered treatment adoption, when covariates may or may not be important for identification. 
In particular, we (a) characterize the DiD potential outcome model in terms of equivalent restrictions on the joint distribution of observables, (b) derive the semiparametric efficient influence functions (EIF) for DiD and ES parameters and the efficient (i.e., smallest) variance bounds, (c) provide closed-form root-n asymptotically normal and efficient estimators as well as consistent estimators of the variances, and (d) show that semiparametric efficiency requires non-uniform weighting of pre-treatment periods and untreated cohorts. In addition, we highlight that the gains in efficiency/power can be empirically relevant in finite samples, even in designs without variation in treatment timing and covariates. 

We start by providing an observationally equivalent characterization of the DiD potential outcome identification assumptions in terms of restrictions on the joint distribution of observables. This characterization links modern DiD designs to econometric models of sequential conditional moment restrictions with unknown functions of observables \citep{Ai_Chen_2012}.
It clarifies the informational content embedded in modern DiD identification assumptions, and enables our subsequent derivation of the semiparametric efficiency results for DiD models.
Importantly, we show that DiD models are typically nonparametrically overidentified in the sense of \citet{Chen_Santos_2018_ECMA}. To the best of our knowledge, this is the first result to establish such nonparametric overidentification in a causal inference setting without imposing parametric or semiparametric functional-form restrictions on nuisance components.


Next, we derive the semiparametric efficient influence functions for DiD and ES parameters in closed-forms, under various parallel trends assumptions commonly employed in the DiD literature, including settings with and without variation in treatment timing and cases where parallel trends hold only after conditioning on observed covariates. By definition, an EIF for a causal parameter has mean zero and its second moment is the semiparametric efficient variance bound, which is the smallest asymptotic variance across all possible root-$n$ consistent and asymptotically normal estimators under the DiD design identification conditions. Moreover, we provide closed-form estimators based on the EIFs to achieve the efficient variance bounds. To our knowledge, these are the first semiparametric efficiency results in DiD setups with multiple periods and without parametric functional-form assumptions. Importantly, these efficiency results only explore the DiD identification conditions and do not involve additional hard-to-justify restrictions on treatment effect heterogeneity (e.g., homoskedasticity) or the serial correlation of the outcomes.

Our efficient estimators for the DiD and ES are computed using sample EIFs, which are automatically Neyman orthogonal moments.\footnote{For nonparametrically overidentified models \citep{Chen_Santos_2018_ECMA}, not all Neyman orthogonal moment-based doubly robust causal estimators are semiparametric efficient. Only the ones corresponding to the EIFs are semiparametric efficient.} The efficient estimators aggregate over multiple comparison groups and pre-treatment periods using analytically derived optimal weights that are proportional to the (conditional) covariances of outcome changes.

According to our semiparametric efficient variance bounds, DiD and ES estimators that treat all pre-treatment periods and comparison groups as equally informative or only use a particular comparison group and the last pre-treatment period as baseline, are generally inefficient, even in setups with a single treatment date. To achieve semiparametric efficiency, it is important to take into account that different untreated groups and pre-treatment periods are not equally informative, and that the best way to aggregate these information to gain precision involves constructing weights that depend on (conditional) covariance terms related to the outcome changes from different pre-treatment periods $\tprime$ to a given post-treatment period $t_{\text{post}}$ and different comparison groups. We introduce visualization tools that clarify how different baseline periods and comparison groups contribute to the estimator's efficiency, illustrating the structure implied by the semiparametric theory.  Importantly, we stress that the geometry of these weights arises from our semiparametric efficiency results.\footnote{For large $n$ and large $T$ panel data with single treatment date design, \citet{Arkhangelsky2021_SDiD} synthetic DiD exploits information from pre-treatment periods via a researcher-specified weighting criterion. Their weight is not for semiparametric efficiency consideration, however. See Remark \ref{rem:sdid} for details.}

Our new semiparametric efficient variance bounds offer a rigorous benchmark for evaluating a broad class of DiD and ES estimators, including TWFE and those developed by \citet{deChaisemartin2020_AER}, \citet{Callaway_Santanna_2021}, \citet{Sun2021}, \citet{Borusyak2023}, \citet{Wooldridge2021a}, \citet{Gardner2021}, among others. Our proposed EIF-based estimators attain the efficient bounds, thereby dominating the existing estimators in terms of asymptotic efficiency. Simulation studies confirm the theoretical rankings.

 We illustrate the practical relevance of our framework through empirically calibrated simulations and a real-data application. Our simulations build on DiD designs from \citet{Arkhangelsky2021_SDiD} (for single-treatment) and \citet{Baker2022} (for staggered-adoption), calibrated to CPS and Compustat data, respectively. In both settings, our semiparametric efficient estimators deliver markedly lower root mean square error and narrower confidence intervals than alternative popular DiD estimators---often exceeding 40\% gains in precision, with no loss in bias performance. We further revisit \citet{Dobkin_etal_2018_AER}’s real-data analysis of hospitalization and out-of-pocket medical expenses, using the data compiled by \citet{Sun2021}. In this application, alternative estimators would often require sample sizes at least 30\% larger to achieve precision comparable to ours. These findings reinforce the central insight of the paper: aligning estimation strategies with the informational structure of the DiD model identification assumptions can yield substantial improvements in precision across a range of DiD settings.

The rest of the paper is as follows. Section \ref{sec:Setup} formally introduces the general DiD designs and the parameters of interest. Section \ref{sec:bound} derives the semiparametric EIFs and the efficient variance bounds. Section \ref{sec:estimation_inference} proposes our semiparametric efficient estimation. Sections \ref{sec:Simulations} and \ref{sec:application} present the simulation studies and real-data application, respectively. Section \ref{sec:conclusion} concludes with discussions of future extensions.
Appendix \ref{sec:hausman-tests} introduces Hausman-type as well as incremental Sargan tests for the overidentification restrictions. Appendix \ref{sec:Ext1} extends the semiparametric efficiency results to an instrumented DiD setting. Appendix \ref{sec:prof.ident} presents the proofs of the theoretical results.
 
\textbf{Related Literature:} This paper contributes to the popular DiD literature that accommodates treatment effect heterogeneity (see the references mentioned above). We provide the first semiparametric efficiency bounds for DiD and ES estimators in settings with multiple periods and varying forms of the parallel trends assumption, including covariate-conditional and staggered adoption designs. 
There are some recent work on efficient estimation under some extra conditions in DiD literature. Relative to \citet{Santanna2020}, who studies a just-identified two-period, two-group DiD model, we generalize to more empirically relevant overidentified settings with multiple periods and staggered adoption. In contrast to \citet{Borusyak2023}, \citet{Wooldridge2021a}, and \citet{Harmon_2023_Efficiency}, our efficiency results do not rely on auxiliary assumptions such as functional form, homoskedasticity, absence of serial correlation, or other time-series dependence restrictions. These assumptions are convenient but typically hard to justify on theoretical or empirical grounds, particularly in settings with heterogeneous treatment effects and dynamic responses.

Our semiparametric efficiency result builds on \citet{Ai_Chen_2012}, as we characterize the DiD potential outcome model in terms of sequential conditional moment restrictions with unknown functions of observables.\footnote{\citet{Chamberlain_1992_ECMA_efficiency} and \citet{Ai_Chen_2003} contain unknown functions of observables, but without sequential moments.} 
Without conditioning variables, the DiD potential outcome model can be equivalently transformed into Hansen's overidentified unconditional generalized moment restrictions \citep{Hansen1982_GMM}, and the semiparametric efficient variance bound of \citet{Chamberlain_1987_JOE} would also be applicable.
Even in the case without covariates, our closed-form EIF expressions lead to simpler estimation procedure and provide new insights into how an efficient DiD estimator should weight different pre-treatment and comparison groups---to the best of our knowledge, these insights are new to the literature.


\section{Framework, causal parameters, and estimands}\label{sec:Setup}
We start by describing our general setup. We consider a setting with $T$ periods indexed by $t \in \mathcal{T}= \{1,2,\dots,T\}$. At any given time $t>1$, units can start receiving a binary treatment, and different units can start receiving treatment at different points. We focus on setups where treatment is an absorbing state, so once a unit is treated, it remains treated until $t=T$. Let $D_{i,t}$ be an indicator for whether unit $i$ is treated by period $t$ and let $G_i = \min \{ t: D_{i,t} =1\}$ be a ``group'' variable (or cohort) that indicates the first period at which unit $i$ has received treatment. If unit $i$ does not receive treatment by $t=T$, we set $G_i = \infty$ and refer to these units as ``never treated''. Since treatment is an absorbing state, $D_{i,t} = 1$ for all $t \geq G_i$. Let $\mathcal{G}$ denote the support of $G$ and let $\mathcal{G}_{\text{trt}} = \mathcal{G}\setminus \{\infty\}$ be the support of $G$ among ``eventually-treated'' units. Without loss of generality, we assume that a ``never-treated'' group always exists. If all units are eventually treated, we drop all the data from when the last cohort is treated, so the last-treated cohort becomes the ``never-treated'' cohort, and $T$ here denotes the number of available periods in the subset of the data that we will use in our analysis. Finally, we also assume that a vector of pre-treatment covariates $X_{i}$, whose support is denoted by $\mathcal{X}\subseteq \mathbb{R}^d$ is available. Throughout the paper, we adopt a ``large-$n$, fixed-$T$'' short panel data regime, as is typical in most DiD setups.

In this paper, we adopt the Neyman-Rubin potential outcome framework, indexing each potential outcome by the entire treatment path. Let $\textbf{0}_s$ and $\textbf{1}_s$ denote $s$-dimensional vectors of zeros and ones, respectively. We denote the potential outcome of unit $i$ in period $t$ if they were first treated at time $g$ by $Y_{i,t}(\textbf{0}_{g-1}, \textbf{1}_{T-g+1})$, and denote by $Y_{i,t}(\textbf{0}_{T})$ their ``never-treated'' potential outcome. To simplify notation, we can explore that treatment is an absorbing state and index potential outcomes by treatment starting time: $Y_{i,t}(g) = Y_{i,t}(\textbf{0}_{g-1}, \textbf{1}_{T-g+1})$ and $Y_{i,t}(\infty) = Y_{i,t}(\textbf{0}_{T})$. In practice, though, we observe $Y_{i,t} = \sum_{g \in \mathcal{G}} 1\{G_i = g\} Y_{i,t}(g)$, where $1\{A\}$ is the indicator function that takes value one if $A$ is true, and zero otherwise. Henceforth, we write $G_g = 1\{G = g\}$.

We assume that we observe a random sample of $(Y_{t=1},\dots,Y_{t=T}, X', G)'$.
\begin{namedassumption}{S}[Random Sampling] \label{asm:sampling_panel} $\{(Y_{i,t=1},\dots,Y_{i,t=T}, X_i', G_i)'\}_{i=1}^n$ is a random sample from $(Y_{t=1},\dots,Y_{t=T}, X', G)'$.
\end{namedassumption}

We also impose overlap assumptions to avoid ``extrapolation'' and ensure that our identification arguments are nonparametrically valid.
\begin{namedassumption}{O}[Overlap] \label{asm:overlap} For each $g\in \mathcal{G}$, 
$\E[G_g|X] \in (0,1)$ almost surely (a.s.).
\end{namedassumption}

We maintain the following no-anticipation condition throughout the paper. 
\begin{namedassumption}{NA}[No-anticipation] \label{asm:no anticipation} For every $g \in \mathcal{G}_{\text{trt}}$, and every pre-treatment periods $t<g$, $\E[Y_{i,t}(g)|G=g,X] = \E[Y_{i,t}(\infty)|G=g,X]$ almost surely.
\end{namedassumption}
Assumption \ref{asm:no anticipation} states that, on average, units do not change their behavior before treatment starts, which essentially requires that the eventually treated units do not anticipate the treatment taking place. If treatment is announced in advance and units are expected to change their behavior due to it, this assumption suggests that researchers define the treatment date as the time of treatment announcements and not the time of treatment take-up. See \citet{Malani2015} for a discussion.

\subsection{Causal parameters of interest}

Many counterfactuals may be of interest in our context with potential variation in treatment timing. One particular family of causal parameters that have been popular in empirical work is the $ATT(g,t)$ defined by
\begin{equation}\label{eqn:ATT(g,t)}
ATT\left( g,t\right) :=\E\left[ Y_{t}\left( g\right) -Y_{t}\left(\infty\right) |G=g\right] ,~\text{for}~t\geq g.
\end{equation}
Note that $ATT\left( g,t\right)$ measures the Average Treatment Effect at time $t$ of starting treatment at time $g$ versus not starting it, among the units that indeed started treatment at time $g$. Let $CATT(g,t,X) := \E\left[ Y_t(g) - Y_t(\infty) | G=g, X\right]$ denote the conditional $ATT(g,t)$ given covariates $X$.

In setups with a unique treatment date, there is a single treated group, $G=g$ with $g \neq \infty$, so tracking $ATT\left( g,t\right)$ over time for this treated group provides a measure of how treatment effects evolve with elapsed treatment time. This is usually referred to as event studies. With multiple groups defined by their treatment timing, $ATT\left( g,t\right)$ still provides information about treatment effect dynamics for each treated cohort $g\in \mathcal{G}_{\text{trt}}$. However, in such cases, researchers often want to summarize the average causal effects using weighted averages of $ATT\left( g,t\right)$ among units that started treatment $e$ periods ago, with $e \ge 0$. A natural aggregation uses the size of each treated cohort as weights, leading to the event study parameter that we denote by $ES(e)$,
\begin{align}
ES(e) := \E\big[ATT\left( G, G+e\right) \big| G+e \in [2,T]\big] = \sum_{g\in\mathcal{G}_{\text{trt}}} \P(G=g|G+e \in [2,T]) ATT(g,g+e).\label{eqn:ES}
\end{align}

One may also want to further aggregate the event study coefficients to recover a scalar summary measure. Let $\mathcal{E}$ denote the support of post-treatment event time $E=t-G$, $t\geq G$, and let $N_E$ denote its cardinality. Then, 
\begin{align}
ES_{\text{avg}} &:= \dfrac{1}{N_E}\sum_{e\in \mathcal{E}} ES(e)\label{eqn:overall_ATT},
\end{align}
provides a simple average of all post-treatment event study coefficients. 

Under Assumptions \ref{asm:sampling_panel}, \ref{asm:overlap} and \ref{asm:no anticipation} only, one cannot identify $ATT(g,t)$ for the post-treatment periods $t\ge g$, as $\E[Y_t(\infty)|G=g]$ is not observed or identified. DiD procedure tackles this problem by restricting the evolution of average untreated potential outcomes across the treated and comparison groups, a parallel trends (PT) assumption. Depending on the treatment timing and covariates, different versions of PT assumptions have been imposed in the literature. We discuss these below.

\subsection{Parallel trends and DiD estimands}\label{sec:X}
 
Two natural types of (conditional) PT have been adopted in the literature: one that effectively restricts trends of $Y_t(\infty)$ only in post-treatment periods and uses never-treated units as the relevant comparison group, and one that restricts trends of $Y_t(\infty)$ in both pre and post-treatment periods, and allow researchers to use any untreated group as a comparison group. We formally state these two PT assumptions below by allowing them to hold only after conditioning on covariates $X$, i.e., by allowing for covariate-specific trends. If covariates are unavailable or do not play an important role, one can take $X=1~ a.s.$ as a special case to resort to an unconditional DiD analysis.

\begin{namedassumption}{PT-Post}[Parallel Trends in the post-treatment periods for all treated groups]\label{asm:pt-post}
For each $t\in\left\{2,\dots,T\right\} $ and $g \in \mathcal{G}_{\text{trt}}$ such that $t \geq g$,
\begin{equation*}
\mathbb{E}[Y_{t}(\infty)-Y_{t-1}(\infty)|G=g, X]=\mathbb{E}[Y_{t}(\infty)-Y_{t-1}(\infty)|G=\infty, X]~a.s.
\end{equation*}
\end{namedassumption}

\begin{namedassumption}{PT-All}[Parallel Trends for all groups and periods]\label{asm:pt-n-all}
For each $t\in\left\{2,\dots,T\right\}$ and $(g,g') \in \mathcal{G}_{\text{trt}}\times \mathcal{G}$,
\begin{equation*}
\mathbb{E}[Y_{t}(\infty)-Y_{t-1}(\infty)|G=g,X]=\mathbb{E}[Y_{t}(\infty)-Y_{t-1}(\infty)|G=g',X]~a.s.
\end{equation*}
\end{namedassumption}

Assumption \ref{asm:pt-post} only imposes PT in the post-treatment periods, {$t \geq g$}, and imposes that, conditional on covariates, the average untreated potential outcome evolution of any treated group $g\in \mathcal{G}_{\text{trt}}$ and the never-treated group is the same. This assumption has been effectively used by \citet{Callaway_Santanna_2021} and \citet{Sun2021} when constructing DiD estimators using never-treated units as the comparison group. In setups with a single treatment date without covariates, this is the assumption implicitly used in TWFE event study regressions \citep{Jacobson_etal_1993_AER_ES}
\begin{equation}
Y_{i,t} = \alpha_t + \eta_i + \sum_{e\neq -1} 1[E_{i,t} = e] \beta_{e} + \varepsilon_{i,t}, \label{eqn:dynamic-twfe}
\end{equation}
where $ Y_{i,t}$ is the outcome for unit $i$ period $t$, $\alpha_t$ and $\eta_i$ are time and unit fixed effects often described as capturing time-invariant or common time-varying shocks, $E_{i,t} = t - G_i $ is the time relative to treatment (event-time), $\varepsilon_{i,t}$ is an idiosyncratic error term. The summation runs over all possible values of $E_{i,t}$ among eventually-treated units except for event-time $-1$.

Under Assumption \ref{asm:pt-post} (and Assumptions \ref{asm:sampling_panel}, \ref{asm:overlap}, \ref{asm:no anticipation}), the $ATT(g,t)$'s are identified as
\begin{align} 
    ATT(g,t) &  = \mathbb{E}[Y_t - Y_{g-1} | G=g] - \mathbb{E}[ \mathbb{E}[Y_t - Y_{g-1} | G=\infty, X] | G=g ].\label{eqn:ATT-just-identification}
\end{align}
In this case, $t_{\text{pre}} = g-1$ is the only reliable baseline period because the PT condition does not necessarily hold for the pre-treatment periods. In addition, Assumption \ref{asm:pt-post} prevents one from using untreated groups other than the never-treated as the comparison group. In such cases, even if more data on pre-treatment outcomes are available, or the relative size of the never-treated units is small compared to other untreated groups, this information cannot be used to improve the precision of DiD and ES estimators.  But this raises the questions: When is PT plausible in post-treatment periods but not in pre-treatment periods? When does economic theory suggest we can only use a specific untreated group as a valid comparison group?  The answers to these questions are probably application-dependent. But if we want to leverage more data to potentially get more informative inferences, we must strengthen Assumption \ref{asm:pt-post}, leading us to Assumption \ref{asm:pt-n-all}.

Assumption \ref{asm:pt-n-all} states that, conditional on covariates, the average untreated potential outcome follows the same path in all treatment groups and periods. It is perhaps the most popular identification assumption used in the literature, as invoked by  \cite{deChaisemartin2020_AER}, \cite{Sun2021}, \citet{Gardner2021}, \cite{Marcus2021}, \citet{Borusyak2023}, and \citet{Harmon_2023_Efficiency} in the unconditional form, and \citet{Callaway_Santanna_2021}, \citet{Wooldridge2021a}, and \citet{Lee_Wooldridge_2023} in the conditional form. In setups with a single treatment date and no covariates, this assumption is arguably the one that justifies using a ``static'' TWFE regression specification 
\begin{equation}
Y_{i,t} = \alpha_t + \eta_i + \beta D_{i,t} + \varepsilon_{i,t}, \label{eqn:static-twfe}
\end{equation}
where $\beta = ES_{\text{avg}}$ is the parameter of interest. Unlike Assumption \ref{asm:pt-post}, Assumption \ref{asm:pt-n-all} imposes a parallel pre-trends condition across all time periods. 

As Assumption \ref{asm:pt-n-all} allows one to use different pre-treatments as baseline periods and leverage different sets of untreated units as a comparison group, the setup becomes much richer. The following lemma highlights how to combine comparison groups and pre-treatment periods to identify the $ATT(g,t)$'s, hinting that the DiD model becomes overidentified. Let $\mathcal{H}^{g,t} = \{(g',\tprime,\tdoubleprime) \in \mathcal{G} \times \mathcal{T} \times \mathcal{T}: g>\tdoubleprime, g'>\max\{\tprime, \tdoubleprime\} \}$.

\begin{lemma} \label{lem:id-attgt-over}Let Assumptions \ref{asm:sampling_panel}, \ref{asm:overlap}, \ref{asm:no anticipation}, and \ref{asm:pt-n-all} hold. Then, for every group $(g,g')\in \mathcal{G}_{\text{trt}} \times \mathcal{G}_{\text{trt}}$ and time periods $(t,\tprime,\tdoubleprime)\in \mathcal{T} \times \mathcal{T} \times \mathcal{T}$ such that $t\ge g$, $g>\tdoubleprime$, and $g'>\max\{\tprime,\tdoubleprime\}$,  with probability one,
\begin{small}
\vspace{-.3cm}
\begin{align}
    CATT(g,t,X) = \underbrace{\mathbb{E}[Y_t - Y_{\tdoubleprime} |G=g,X]}_{:= m_{g,t,\tdoubleprime}(X)} - \big(\underbrace{\mathbb{E}[Y_t - Y_{\tprime} |G=\infty,X]}_{:= m_{\infty,t,\tprime}(X)} + \underbrace{\mathbb{E}[Y_{\tprime} - Y_{\tdoubleprime} |G=g',X]}_{:= m_{g',\tprime,\tdoubleprime}(X)}\big),\label{eqn:CATT-id}
\end{align}
\end{small}
\vspace{-0.7cm}

\noindent and, as a consequence, 
\begin{align}
    ATT(g,t) =&~ 
    \E\big[Y_t - Y_{\tdoubleprime}\big|G=g \big] -\E\big[ \left(m_{\infty,t,\tprime}(X) + m_{g',\tprime,\tdoubleprime}(X) \right) \big|G=g\big].\label{eqn:ATT-overidentification-more}
\end{align}
More generally, for any covariate-specific weights $w^{\text{att(g,t)}}_{g',\tprime,\tdoubleprime}(X)$ summing to one,  we have that
\begin{align*}
    ATT(g,t) = \E\bigg[ \sum_{(g',\tprime, \tdoubleprime) \in \mathcal{H}^{g,t}} w^{\text{att(g,t)}}_{g',\tprime,\tdoubleprime}(X)\left[ m_{g,t,\tdoubleprime}(X) - \left(m_{\infty,t,\tprime}(X) + m_{g',\tprime,\tdoubleprime}(X) \right)\right] \bigg|G=g\bigg].
\end{align*}
\end{lemma}

Lemma \ref{lem:id-attgt-over} leverages several restrictions on the indexes worth explaining. First, the restriction $t\ge g$ means we identify $ATT(g,t)$'s in group $g$ post-treatment periods. The restriction $g>\tdoubleprime$ suggests using any of the group's $g$ pre-treatment period as a baseline, while the restriction $g'>\max\{\tprime,\tdoubleprime\}$ means that we are using the eventually treated cohort $g'$ that is treated after periods $\tprime, \tdoubleprime$ as part of our effective comparison group. Interestingly, these restrictions allow $\tprime$ to be a post-treatment period for group $g$, as long as it is a pre-treatment period for cohort $g'$. Figure \ref{fig:illustration} illustrates why this is possible in our setup using a stylized example with four groups and ten time periods---essentially, when parallel trends hold in all pre-treatment periods, one can ``subtract back'' some ``excessive'' corrections using the never-treated units. 

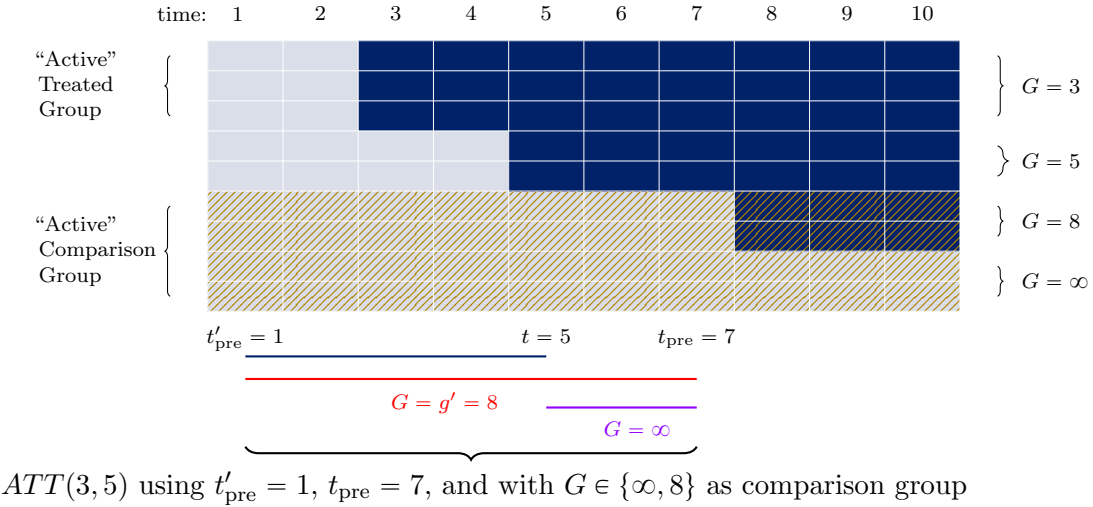
\begin{figure}[!htbp]
\caption{Illustration of how we can use $\tprime$ in post-treatment period for group $g$} 
\label{fig:illustration}
\centering
        \begin{tikzpicture}[xscale=1, yscale=0.4]
    \def\entities{9}
    \def\periods{10}
    
    \def\colorGray{gray!60}
    \def\colorBlack{black}
    
    \foreach \i in {1,...,\entities} {
        \foreach \j in {1,...,\periods} {
            \fill[emoryblue!15!white] (\j-1,\i) rectangle ++(1,-1);
        }
    }
    
    \foreach \i/\start in {
    1/3,
    2/3, 
    3/3,  
    4/5,
    5/5,
    6/8,
    7/8
} {
    \foreach \j in {\start,...,\periods} {  
        \fill[emoryblue] (\j-1, \entities-\i+1) rectangle ++(1,-1);
    }
}
    \begin{scope}
    \clip (0,0) rectangle (\periods, 4);  
    \fill[emoryblue!15!white, pattern=north east lines, pattern color=emorygold] (0, 0) rectangle (\periods, 4);  
    \end{scope}

    \foreach \j in {0,...,\periods} {
        \draw[very thin, -, white] (\j, 0) -- (\j, \entities);
    }
    
    \foreach \i in {0,...,\entities} {
        \draw[very thin, -, white] (0, \i) -- (\periods, \i);
    }
    
    \foreach \j in {2,...,\periods} {
        \node[anchor=north] at (\j-0.5, 10.5) {\scriptsize \j};
    }
    \foreach \j in {1} {
        \node[anchor=north] at (\j-1.1, 10.5) {\scriptsize time:~~ \j};
    }
    

    \node[anchor=north, text height=1.5ex] at (.5, 0) {\scriptsize $\tdoubleprime=1$};  
\node[anchor=north, text height=1.5ex] at (4.5, 0) {\scriptsize$t=5$};  
\node[anchor=north, text height=1.5ex] at (6.5, 0) {\scriptsize $\tprime=7$};

    \draw [decorate,decoration={brace,amplitude=0pt,mirror},-,emoryblue, thick] (.5, -1.5) -- (4.5, -1.5);

    \draw [decorate,decoration={brace,amplitude=0pt},-, red, thick] (6.5, -2.25) -- (.5, -2.25) node[midway,xshift=-10pt,below=0.8pt] {\scriptsize $G=g'=8$};

    \draw [decorate,decoration={brace,amplitude=0pt},-, purple, thick] (6.5, -3.2) -- (4.5, -3.2) node[midway,xshift=6pt,below=1pt] {\scriptsize $G=\infty$};
    
    \draw [decorate,decoration={brace,amplitude=5pt},-, thick] (6.5, -4.5) -- (.5, -4.5) node[midway,xshift=5pt,below=5pt] {\small
    $ATT(3,5)$ using $\tdoubleprime=1$, $\tprime=7$, and with $G\in \{\infty, 8\}$ as comparison group};

    \draw [decorate,decoration={brace,amplitude=2pt},-] (10.5, 8.5) -- (10.5, 6.5) node[midway,xshift=5pt,right=0pt] {\scriptsize$G=3$};

    \draw [decorate,decoration={brace,amplitude=4pt},-] (10.5, 5.5) -- (10.5, 4.5) node[midway,xshift=5pt,right=0pt] {\scriptsize$G=5$};

    \draw [decorate,decoration={brace,amplitude=2pt},-] (10.5, 3.5) -- (10.5, 2.5) node[midway,xshift=5pt,right=0pt] {\scriptsize$G=8$};

    \draw [decorate,decoration={brace,amplitude=2pt},-] (10.5, 1.5) -- (10.5, 0.5) node[midway,xshift=5pt,right=0pt] {\scriptsize$G=\infty$};

    \draw [decorate,decoration={brace,amplitude=2pt,mirror},-] (-.5, 3.5) -- (-.5, 0.5) node[midway,left=0pt] {\parbox{1.6cm}{\scriptsize ``Active'' \\ Comparison \\ Group}};

        \draw [decorate,decoration={brace,amplitude=2pt,mirror},-] (-.5, 8.5) -- (-.5, 6.5) node[midway,left=0pt] {\parbox{1.6cm}{\scriptsize ``Active'' \\ Treated \\ Group}};
    \end{tikzpicture}

\justifying
\noindent\scriptsize{\textit{Notes:} Illustrative example of a setup with 10 time periods and four different treatment groups, where one is interested in learning about $ATT(3,5)$. Dark blue means units are treated (post-treatment periods), while light blue denotes pre-treatment periods. The hashed area denotes the ``active'' comparison group.
The illustration highlights that one can use a combination of two not-yet-treated cohorts, $G=8$ and $G=\infty$, as a comparison group, as well as leverage $\tdoubleprime=1$ and $\tprime=7$ as ``baseline periods''. Using a post-treatment $\tprime=7$ as a baseline period is possible because the never-treated comparison group bridges this back to pre-treatment periods, since parallel trends hold in all periods for all groups. }
     
\end{figure}

Overall, Lemma \ref{lem:id-attgt-over} indicates that one can use any group $g$'s pre-treatment period as a baseline, combine never-treated units with any not-yet-treated cohort $g'>g$ to form comparison groups, and use pre-treatment data for eventually-treated cohort $g'$ when identifying an $ATT(g,t)$. In practice, there are infinitely many ways one can combine these to identify the $ATT(g,t)$'s. In practice, however, one very relevant question is how to combine these different estimands to asymptotically obtain the most precise DiD estimator for the $ATT(g,t)$'s, or to obtain the most precise event study parameters $ES(e)$ as defined in \eqref{eqn:ES}? We discuss these points in the next section.

\begin{remark}\label{rem:negative-weights}
    As Lemma \ref{lem:id-attgt-over} indicates, any covariate-specific weights $w^{\text{att(g,t)}}_{g',\tprime,\tdoubleprime}(X)$ summing to one can be used to identify a $ATT(g,t)$. That is, the result in Lemma \ref{lem:id-attgt-over} allows for non-convex weights. Different from the DiD literature that discusses limitations of simple two-way fixed effects specifications (\citealp{Goodmanbacon2021}, \citealp{Borusyak2023}, \citealp{deChaisemartin2020_AER}, \citealp{Sun2021}), non-convex weights are \emph{not} a concern in our context. Under the identifications in Lemma \ref{lem:id-attgt-over}, the DiD model is overidentified, and \eqref{eqn:CATT-id} is \emph{homogeneous} across baseline periods and comparison groups. Thus, issues related to an estimand not being weakly causal \citep{Blandhol_etal_2SLS_2022} due to negative weights are not binding in our context.
\end{remark}

\section{Semiparametric Efficiency Bound for DiD and ES}\label{sec:bound}

In this Section, we derive the semiparametric efficiency bound for the $ATT(g,t)$'s under Assumption \ref{asm:pt-post} and under Assumption \ref{asm:pt-n-all}. We also provide analogous results for the event study parameters $ES(e), e\ge0$. These semiparametric efficiency bounds characterize the lowest possible asymptotic variance attainable by any regular, asymptotically linear estimator given the identification assumptions. Based on the efficient influence function, we derive closed-form expressions for estimands that can be used to construct efficient estimators under minimal smoothness assumptions. These results provide a benchmark for evaluating the asymptotic efficiency of existing DiD and ES estimators in the literature.

To build intuition, we first focus on the canonical case of a single treatment date. This simpler setting provides a transparent view of the key ideas and allows us to isolate the informational content of each pre-treatment period. We then generalize the results to the more complex case of staggered treatment timing. When covariates are uninformative for identification, the results further simplify to the unconditional DiD setup by setting $X=1$. 

\subsection{DiD setups with a single treatment date}\label{sec:single-treatment-time}

Consider the setting with a single treatment date, resulting in two groups: treated ($G=g$) and never-treated ($G=\infty$). In this case, the event study parameters $ES(e)$ are equivalent to $ATT(g,t)$'s with $t=g+e$. We proceed under either Assumption \ref{asm:pt-post} or Assumption \ref{asm:pt-n-all}, noting that the former is a special case of the latter. With two groups, we have  $G_\infty = 1-G_g$ and $p_\infty(X) = 1-p_g(X)$, with $p_g(X) = \E[G_g|X]$ being the propensity score, i.e., the probability of belonging to the treated group given covariates $X$.

Under Assumptions \ref{asm:sampling_panel}, \ref{asm:overlap}, \ref{asm:no anticipation}, and \ref{asm:pt-n-all}, Lemma \ref{lem:id-attgt-over} implies that for any post-treatment period $t\ge g$, any pre-treatment period $\tprime \in \{1,\dots, g-1\}$, and any weight $w^{\text{att(g,t)}}_{\tprime}(\cdot)$ summing to one,
\begin{align*}
    ATT(g,t) &=\E\big[Y_t - Y_{\tprime}\big|G=g \big] -\E[ \E\big[Y_t - Y_{\tprime}\big|G=\infty, X] \big|G=g\big] \\
    & =   \E\bigg[ \sum_{\tprime = 1}^{g-1} w^{\text{att(g,t)}}_{\tprime}(X)\left[ m_{g,t,\tprime}(X) - m_{\infty,t,\tprime}(X) \right] \bigg|G=g\bigg].
\end{align*}
This expression shows how information from multiple pre-treatment periods can be used to identify our target parameters, but it does not yet indicate \emph{how} to weigh each period to gain precision. To this end, the next step is to characterize all the information content in our identification assumptions. In the next lemma, we establish an equivalent representation of our identification assumptions that impose restrictions on potential outcomes using sequential conditional moment restrictions on observable variables. 

\begin{lemma}[Moment-restrictions for overidentified DiD with single treatment time] \label{lm:1} \phantom{abcgsfafafa}\linebreak
Suppose there are only two treatment groups, such that $G\in \mathcal{G}=\{g,\infty\}$, with $g\ge 2$. The family of probability distributions of $(Y_{t=1}\cdots,Y_{t=T},X',G)$ satisfying Assumptions \ref{asm:sampling_panel}, \ref{asm:overlap}, \ref{asm:no anticipation}, and \ref{asm:pt-n-all} are observationally equivalent to the family of probability distributions of $(Y_{t=1}\cdots,Y_{t=T},X',G)$ satisfying Assumptions \ref{asm:sampling_panel}, \ref{asm:overlap}, and the following set of moment restrictions: for all $t \in \{g,\dots, T\}$, with probability one, 
    \begin{align*}
    \mathbb{E}[G_g(ATT(g,t) - CATT(g,t,X))] & =0, \\
    \mathbb{E} \left[ CATT(g,t,X) - \frac{G_g(Y_{t} - Y_{g-1})}{p_g(X)} + \frac{G_\infty(Y_{t} - Y_{g-1})}{p_\infty(X)} \Bigg|X \right] & = 0, \\
    \mathbb{E} \left[ \frac{G_g(Y_{\tprime} - Y_{1})}{p_g(X)} - \frac{G_\infty(Y_{\tprime} - Y_{1})}{p_\infty(X)} \Bigg|X \right] & = 0, \text{ for all } 2 \leq \tprime \leq g-1, \\
    \mathbb{E}[G_g - p_g(X) | X] & = 0.
\end{align*}
\end{lemma}
Lemma \ref{lm:1} lays the groundwork for deriving the efficient influence function and the semiparametric efficiency bound. Beyond that, it succinctly presents all the identifying power of our identification assumptions and highlights that one can construct many different DiD estimators based on it: it is just a matter of choosing an estimation method suitable for moment restrictions. Examples include (two-step) generalized methods of moments (GMM; see, e.g., \citealp{Hansen1982_GMM}, \citealp{Ackerberg2014_two-step-gmm-Restud}), generalized empirical likelihood  \citep{Newey_Smith_2004_GEL}, and minimum distance estimators \citep{Ai_Chen_2003, Ai_Chen_2007, Ai_Chen_2012}. 

A natural question is whether the choice of the estimation procedure matters in terms of asymptotic efficiency. As discussed in the previous section, the DiD model characterized in Lemma \ref{lm:1} is overidentified in the sense of \citet{Chen_Santos_2018_ECMA}. As such, not every nonparametric estimation procedure is asymptotically semiparametrically efficient, implying that it is relevant to characterize the minimum asymptotic variance for any regular, asymptotic linear (RAL) DiD estimator, and to consider estimators that achieve this bound. 

Although different efficiency tools are available in the literature, many of them are not suitable for our DiD model in Lemma \ref{lm:1}. For instance, the DiD model in Lemma \ref{lm:1} is based on sequential conditional moment restrictions with unknown nonparametric functions, making the efficiency results in \citet{Chamberlain_1992_ECMA_efficiency} and \citet{Ai_Chen_2003} not directly applicable. In addition, note that the conditional average treatment effects among the treated units in period $t$ in our DiD model, $CATT(g,t,X)$, is given by 
\begin{align*}
    CATT(g,t,X) = \E\left[Y_{t} - Y_{\tprime}|G=g, X\right] - \E\left[Y_{t} - Y_{\tprime}|G=\infty, X \right]~~\text{for any $1 \leq \tprime \leq g-1$}.
\end{align*}
This suggests that our first-step nuisance function is overidentified, making the efficiency results of \citet{Ackerberg2014_two-step-gmm-Restud} based on two-step GMM unsuitable for our context. On the other hand, we can build on \citet{Ai_Chen_2012}, as their models are well-suited to our setup. Nonetheless, applying their results still requires substantial additional work: we need to first solve a calculus variations problem and then unravel several nested matrix inversions before arriving at the concise expression provided in Theorems \ref{thm:efficiency-fixed-g} and \ref{thm:efficiency}.

In what follows, we present our semiparametric efficiency result under Assumption \ref{asm:pt-n-all}, and then present as a special case the semiparametric efficiency under Assumption \ref{asm:pt-post}. Let $\pi_g =\mathbb{E}[G_g]$ denote the population proportion of treated units, and denote the group-specific conditional expectations of the evolution of outcomes from period $\tprime$ to period $t$ by
\begin{align} 
    m_{g,t,\tprime}(X)  :=\mathbb{E}[Y_t - Y_{\tprime} | G=g,X].\label{eqn:m_functions}
\end{align}
Let $\mathbf{1}$ denote a column vector of ones with the appropriate length varying according to the context. Let $W = (Y_{t=1}, \dots, Y_{t=T}, X', G_g)$ denote the available random variables, and   
\begin{align}
    \widetilde{Y}_{g,t,\tprime} &= \dfrac{1}{\pi_g} \left(G_g - \dfrac{p_g(X)}{p_\infty(X)}G_\infty \right)(Y_t - Y_{\tprime} - m_{\infty,t,\tprime}(X)).\label{eqn:Y_tilde_single}
\end{align}
For $1 \leq \tprime \leq g-1$, let
\begin{align} \label{eqn:IF-g}
     \mathbb{IF}_{\tprime}^{att(g,t)} & = \widetilde{Y}_{g,t,\tprime} - \frac{G_g}{\pi_g}ATT(g,t), 
\end{align}
and denote the $g-1$ column vector that stacks of these $g-1$ influence functions by
    $\mathbb{IF}^{att(g,t)} =(\mathbb{IF}_{1}^{att(g,t)}, \mathbb{IF}_{2}^{att(g,t)},\cdots,\mathbb{IF}_{g-1}^{att(g,t)})'$.
Finally, denote the $(g-1)\times (g-1)$ conditional covariance matrix of $\mathbb{IF}^{att(g,t)}$ as $V_{gt}(X) =\Cov(\mathbb{IF}^{att(g,t)}|X)$, and let $V_{gt}^*(X)$ be a matrix of the same dimension as $V_{gt}(X)$ with the $(j,k)$-th element being
\begin{align} 
    \frac{1}{p_g(X)} \Cov(Y_{t} - Y_{j},Y_{t} - Y_{k} |G=g,X) + \frac{1}{1 - p_g(X)} \Cov (Y_{t} - Y_{j} ,Y_{t} - Y_{k}  |G=\infty,X ).\label{eqn:V-star}
\end{align}

\begin{theorem} \label{thm:efficiency-fixed-g}
   Suppose there are only two treatment groups, such that $G\in \mathcal{G}=\{g,\infty\}$, with $g\ge 2$. Under Assumptions \ref{asm:sampling_panel}, \ref{asm:overlap}, \ref{asm:no anticipation} and \ref{asm:pt-n-all}, the efficient influence function of a $ATT(g,t)$, $t\ge g$ is given by
    \begin{align*}
        \mathbb{EIF}^{att(g,t)} & = \frac{\mathbf{1}'V_{gt}(X)^{-1}}{\mathbf{1}'V_{gt}(X)^{-1} \mathbf{1} } \mathbb{IF}^{att(g,t)} = \frac{\mathbf{1}'V_{gt}^*(X)^{-1}}{\mathbf{1}'V_{gt}^*(X)^{-1} \mathbf{1} } \mathbb{IF}^{att(g,t)}.
    \end{align*}
    Assuming that the second moment of the efficient influence function is finite, the semiparametric efficiency bound for any regular, asymptotically linear estimator for the $ATT(g,t)$ is given by
    \begin{align*}
       V_{\text{eff}}= & \frac{1}{\pi_g^2 } \left(\mathbb{E}[G_g(CATT(g,t,X) - ATT(g,t))^2] 
        + \mathbb{E}\left[\frac{p_g(X)^2 |V_{gt}^*(X)|}{\sum_{j,j'=1}^{g-1} (-1)^{j+j'}|V_{gt,jj'}^*(X)|} \right] \right),
    \end{align*}
    where $V^*_{gt,jj'}$ denotes the minor matrix obtained by removing the $j$th row and $j'$th column of $V^*_{gt}$.\footnote{In the formula of the efficiency bound, we follow the convention in matrix theory to define the determinant of an empty matrix to be one.}
\end{theorem}

In the proof of Theorem \ref{thm:efficiency-fixed-g}, we show that including one post-treatment at a time and focusing on one $ATT(g,t)$ at a time leads to the same efficiency bound as including all post-treatments and estimating all $ATT(g,t)$'s. The main advantage of working with one $ATT(g,t)$ at a time is that it avoids dealing with a much larger variance-covariance matrix, which would preclude us from getting more intuition about efficiently aggregating each pre-treatment period. When it comes to estimation, this simpler characterization will allow us to avoid dealing with a much larger number of conditional moment restrictions that would lead to more computationally challenging procedures. Moreover, as a side result, this characterization allows us to present the semiparametric efficiency bound for the just-identified DiD setup when one imposes Assumption \ref{asm:pt-post} instead of Assumption \ref{asm:pt-n-all} in Corollary \ref{cor:2-period} as a corollary of Theorem \ref{thm:efficiency-fixed-g}.

\begin{corollary} \label{cor:2-period}
Suppose there are only two treatment groups, such that $G\in \mathcal{G}=\{g,\infty\}$, with $g\ge 2$. Under Assumptions \ref{asm:sampling_panel}, \ref{asm:overlap}, \ref{asm:no anticipation} and \ref{asm:pt-post}, the efficient influence function of a $ATT(g,t)$, $t\ge g$, is equal to $\mathbb{IF}_{g-1}^{att(g,t)}$, and the semiparametric efficiency bound reduces to $V_{\text{eff,j-id}} =  \mathbb{E}\left[\left( \widetilde{Y}_{g,t,g-1} - \frac{G_g}{\pi_g}ATT(g,t)\right)^2\right]$.
\end{corollary}

Theorem \ref{thm:efficiency-fixed-g} and Corollary \ref{cor:2-period} have several interesting implications.  First, suppose we only use the pre-treatment period $t_{\text{pre}}=g-1$ as a baseline, so PT holds only for post-treatment periods (Assumption \ref{asm:pt-post}). Corollary \ref{cor:2-period} shows that the efficient influence function is $\mathbb{IF}_{g,g-1}^{g,t}$ under Assumption \ref{asm:pt-post}. This is effectively the same as dropping data from all periods but time $t$ (post-treatment) and time $g-1$ (pre-treatment), transforming the multi-period DiD model into a two-period DiD setup. Thus, under Assumption \ref{asm:pt-post}, Corollary \ref{cor:2-period} can be understood as generalizing the semiparametric efficiency bound result of \cite{Santanna2020} from the much simpler 2-group-2-period DiD setup to the multi-period DiD setup with a single treatment date. Importantly, Corollary \ref{cor:2-period} highlights that pre-treatment data beyond the last pre-treatment period cannot be used to improve asymptotic efficiency, which can be hard to motivate empirically.

When one can leverage multiple pre-treatment periods (Assumption \ref{asm:pt-n-all}), Theorem \ref{thm:efficiency-fixed-g} highlights that the efficient influence function, in this case, is an optimally weighted average of all influence functions based on different baseline periods, where the optimal set of weights is computed by minimizing the conditional variance given $X$. Notably, the optimal way to aggregate pre-treatment periods for estimating $ATT(g,t)$ will generally differ from when estimating $ATT(g,\tprime)$, with $\tprime\ne t$, with weights generally being non-uniform across pre-treatment periods and covariate strata. Of course, there are special cases where weighing all pre-treatment information equally is optimal. This is the case when outcome changes are (conditionally) uncorrelated across periods for both the treated and comparison group and have constant (conditional) variances across pre-treatment periods for both treatment groups (\citealp{Wooldridge2021a}, \citealp{Borusyak2023}). Outside this scenario, it is hard to see a setup where equal weights for all pre-treatment periods are asymptotically optimal. Our semiparametric efficiency results do not impose any of these strong, hard-to-empirically motivate assumptions. As highlighted in Lemma \ref{lm:1}, we do not impose additional restrictions beyond the identification assumptions, therefore being able to capture much richer notions of heterogeneity.

The efficient weights in Theorem \ref{thm:efficiency-fixed-g} also vary with covariates. This is intuitive, as the optimal way to combine pre-treatment periods may vary across covariate strata. For instance, it can be that for a particular partition of the covariate space, more recent pre-treatment periods are ``more informative'', while for another partition, it can be that all pre-treatment periods are equally informative. The flexibility of the weights varying with $X$ is meant to capture exactly this. It is also worth mentioning that even when the covariance terms in \eqref{eqn:V-star} are trivial functions of $X$, a type of homoscedasticity assumption within each treatment group, $V_{gt}^*(X)$ and, therefore, the optimal weights would still generally depend on $X$ via the propensity score $p_g(X)$.\footnote{It may be the case that the ratio of conditional covariance and the propensity score in \eqref{eqn:V-star} is constant, even when all these terms are not trivial functions of the covariates. Nonetheless, this seems like a knife-edge case requiring a very specific covariance structure, which is unrealistic for most applications.} Thus, Theorem \ref{thm:efficiency-fixed-g} stresses the importance of choosing covariate-specific weights to achieve semiparametric efficiency. 

The efficient influence function in Theorem \ref{thm:efficiency-fixed-g} also informs how one can construct DiD estimators that achieve the semiparametric efficiency bound by effectively weighting different pre-treatment periods. More specifically, we can leverage the fact that the efficient influence function has mean zero to get the following DiD estimand:
\begin{align}\label{eqn:ATT_gt_opt_estimand}
ATT(g,t) = \E\left[\frac{\mathbf{1}'V_{gt}^*(X)^{-1}}{\mathbf{1}'V_{gt}^*(X)^{-1} \mathbf{1} }\widetilde{Y}_{g,t} \right],
\end{align}
where $\widetilde{Y}_{g,t} = (\widetilde{Y}_{g,t,1}, \dots, \widetilde{Y}_{g,t,g-1})'$ is a $(g-1) \times 1$ column vector of transformed outcomes.
Thus, the DiD estimand $ATT(g,t)$ uses a combination of pre-treatment periods as ``effective'' baselines, where the weights for each of these pre-treatment periods are given by $\mathbf{1}'V_{gt}^*(X)^{-1}\big/\mathbf{1}V_{gt}^*(X)^{-1} \mathbf{1}$. This expression highlights that, indeed, ``more informative'' baseline periods get weighted more than ``less informative'' periods, where the relevant notion of ``informativeness'' is given by the inverse of $V_{gt}^*(X)$ with elements defined in \eqref{eqn:V-star}. 

\begin{remark}\label{rem:sdid}
    To our knowledge, the synthetic DiD procedure of \citet{Arkhangelsky2021_SDiD} is the only other panel-data method that weights pre-treatment periods in a non-uniform manner. However, many important differences between our procedure and theirs are worth stressing. First, and in contrast to \citet{Arkhangelsky2021_SDiD}, our non-uniform weights are solved in closed-form for the semiparametric efficiency purpose. Second, our weights can vary with covariates, while covariates play no major role in their synthetic DiD procedures. On the other hand, our results depend on the plausibility of parallel trends, while synthetic DiD procedures can accommodate some deviations of parallel trends. Finally, while we focus on setups with a fixed number of periods, \citet{Arkhangelsky2021_SDiD} inferential procedures are tailored for applications with a diverging number of pre-treatment periods. 
\end{remark}


\begin{remark}\label{rem:summary}
    Researchers often want to summarize the average treatment effects across all periods into a scalar parameter,
    \begin{align*}
ES_{\text{avg}} &= \dfrac{1}{N_E}\sum_{e\in \mathcal{E}} ES(e).
\end{align*}
It follows from Theorem \ref{thm:efficiency-fixed-g} and the Delta method that one can form efficient estimators for $ES_{\text{avg}}$ in this two-group DiD setup using
\begin{align}
    ES_{\text{avg}} = \dfrac{1}{T-g+1}\sum_{t=g}^T \E\left[\frac{\mathbf{1}'V_{gt}(X)^{-1}}{\mathbf{1}'V_{gt}(X)^{-1} \mathbf{1} } \widetilde{Y}_{g,t}\right]\label{eqn:EI_ES},
\end{align}
with $ATT(g,t)$'s as defined in \eqref{eqn:ATT_gt_opt_estimand}. As evident from \eqref{eqn:EI_ES}, the efficient way to combine pre-treatment information also varies with covariates, and it is unlikely to use a uniform weighting scheme in many applications with any serial correlation. 
\end{remark}

\begin{remark}\label{rem:sargan}
    From Lemma \ref{lm:1}, Theorem \ref{thm:efficiency-fixed-g}, and Corollary \ref{cor:2-period}, it is clear that our DiD model is over-identified. A natural consequence of these results is that we can directly assess the validity of the identification assumption via a Hausman-type test. One can also sequentially assess whether using data from periods further away from the treatment date is warranted via incremental Sargan tests. We discuss these testing procedures in Appendix \ref{sec:hausman-tests}.
    
    \end{remark}

\subsection{DiD setups with staggered treatment designs}\label{sec:staggered}
We now extend our semiparametric efficiency bound results for ATT and ES estimators from the previous section to the staggered setups and provide guidelines on the most precise RAL DiD and ES estimators. We start studying the overidentified staggered DiD model for $ATT(g,t)$, and later discuss the just-identified one as a special case. As we discuss below, the efficiency of $ES(e)$, $e\ge0$, follows from the efficiency of $ATT(g,t)$'s and the Delta method. 

As staggered treatment times imply the existence of multiple treatment groups, we now have that $G_\infty = 1 - \sum_{g \in \mathcal{G}_{\text{trt}}} G_g $ and $p_\infty(X) = 1-\sum_{g \in \mathcal{G}_{\text{trt}}} p_g(X)$, with $p_g(X) = \E[G_g|X]$ being a generalized propensity score, i.e., the probability of belonging to the group that is treated for the first time in period $g$ given covariates $X$. 

The following lemma characterizes the information content of our staggered DiD identification assumptions based on potential outcomes using sequential conditional moment restrictions on observable variables.

\begin{lemma}[Moment-restrictions for overidentified staggered DiD]\label{lm:2}
\phantom{abcgsfafafaagagagagahahah}\linebreak
The family of probability distributions of $(Y_{t=1}\cdots,Y_{t=T},X',G)$ satisfying Assumptions \ref{asm:sampling_panel}, \ref{asm:overlap}, \ref{asm:no anticipation}, and \ref{asm:pt-n-all} are observational equivalent to the family of probability distributions of $(Y_{t=1}\cdots,Y_{t=T},X',G)$ satisfying Assumptions \ref{asm:sampling_panel}, \ref{asm:overlap} and the following set of moment restrictions: for all $g,g' \in \mathcal{G}_{\text{trt}}\times \mathcal{G}_{\text{trt}}$ and post-treatment periods $t\in \{g,\dots, T\}$, with probability one,  
\begin{align}
        \mathbb{E}[\pi_g - G_g] & = 0, \nonumber \\
        \mathbb{E}[G_g (ATT(g,t) - CATT(g,t,X))] & = 0, \nonumber\\
        \mathbb{E}\left[ CATT(g,t,X) - \frac{G_g(Y_{t} - Y_{g-1})}{p_g(X)} + \frac{G_{\infty}(Y_{t} - Y_{g-1})}{p_{\infty}(X)} \Big| X \right] & = 0, \nonumber \\
        \mathbb{E}\left[\frac{G_{g'}(Y_{\tprime} - Y_1)}{p_{g'}(X)} - \frac{G_\infty(Y_{\tprime} - Y_1)}{p_\infty(X)}\Big| X \right] & = 0, \text{for all } 2 \leq \tprime \leq g'-1, \label{eqn:over-id-moments}\\
        \mathbb{E}[G_g - p_g(X)|X]&=0  \nonumber.
\end{align}
\end{lemma}

The first two (unconditional) sets of moment conditions in Lemma  \ref{lm:2} define each cohort's relative size, $\pi_g$, and the $ATT(g,t)$ as a functional of the conditional $ATT(g,t)$. The third and fourth sets of (conditional) moment restrictions define the conditional $ATT(g,t)$ given covariates, $CATT(g,t,X)$, and the set of overidentification restrictions implied by our Assumptions \ref{asm:sampling_panel}, \ref{asm:overlap}, \ref{asm:no anticipation}, and \ref{asm:pt-n-all}. By combining the third and fourth conditional moment restriction, it is clear that any pre-treatment period $\tprime<g$ and multiple comparison groups can be used to identify the $CATT(g,t,X)$, and therefore, the $ATT(g,t)$. This is what essentially rationalizes Lemma \ref{lem:id-attgt-over}.
The last set of conditional moment restrictions defines the generalized propensity scores, $p_g(X)$. Finally, also note that, from Lemma \ref{lm:2}, we can construct the event study parameters $ES(e), e\ge0$, as
\begin{align}
    ES(e) = \sum_{g \in \mathcal{G}_{\text{trt},e}} \frac{\pi_g}{\sum_{g \in \mathcal{G}_{\text{trt},e}} \pi_g} ATT(g,g+e), \label{eqn:ES_staggered}
\end{align}
where, for $e\ge 0$, $\mathcal{G}_{\text{trt},e} =\{g \in \mathcal{G}_{\text{trt}}: g + e \leq T\}$ denotes the set of eventually-treated groups that have data for $e$-periods after treatment started. 

Similarly to Lemma \ref{lm:1}, Lemma \ref{lm:2} presents all the identifying power of our identification assumptions in the staggered DiD setup and highlights that one can construct many different DiD estimators based on it. The relevant question is how to fully explore the empirical content of the moment conditions in Lemma \ref{lm:2} to form  DiD and ES estimators with appealing semiparametric efficiency guarantees. Here, in contrast to the setup in Lemma \ref{lm:1}, we can use different pre-treatment periods and different comparison groups when constructing DiD and ES estimators. Hence, the setup is more interesting, though it also requires additional notation to characterize how the ``most precise'' DiD and ES estimators should look. 

Let the ``generated outcome'' for a given $ATT(g,t)$ be defined as
\begin{align} \widetilde{Y}^{\text{att(g,t)}}_{g',\tprime}  =& ~\dfrac{G_g}{\pi_g} ( Y_t - Y_1 -  m_{\infty,t,\tprime}(X) -  m_{g',\tprime,1}(X)) \nonumber\\
    & - \frac{p_g(X)}{p_\infty(X)}~\dfrac{G_\infty}{\pi_g}(Y_t - Y_{\tprime} - m_{\infty,t,\tprime}(X)) - \frac{p_g(X)}{p_{g'}(X)}\dfrac{G_{g'}}{\pi_g}(Y_{\tprime} - Y_{1} - m_{g',\tprime,1}(X)) ,\label{eqn:EIF_estimand_for_each}
\end{align} 
and, for $g' \in \mathcal{G}_{\text{trt}}$ and $1 \leq \tprime \leq g'-1$, let 
\begin{align} \label{eqn:IF-general}
    \mathbb{IF}_{g',\tprime}^{att(g,t)}  = ~ \widetilde{Y}^{\text{att(g,t)}}_{g',\tprime}  + \frac{G_g}{\pi_g} ATT(g,t).
\end{align}
Note that $\widetilde{Y}^{\text{att(g,t)}}_{g',\tprime}$ uses data from the earliest pre-treatment and from $\tprime$. It also leverages observations from the ``never-treated'' group $G_\infty$ and from $g'$ as the comparison group. When $g'=g$, though, it only uses data from $G_\infty$ as a comparison group, and, therefore, \eqref{eqn:EIF_estimand_for_each} and \eqref{eqn:IF-general} reduce to \eqref{eqn:Y_tilde_single} and  \eqref{eqn:IF-g}, respectively. 

Next, we collect all noncollinear $\mathbb{IF}_{g',\tprime}^{att(g,t)}$ to characterize the ``relevant'' set of influence functions to be used in the construction of the efficient influence function. For $g=g'$, let $\mathbb{IF}_{g}^{att(g,t)} =(\mathbb{IF}_{g,1}^{att(g,t)}, \mathbb{IF}_{g,2}^{att(g,t)},\cdots,\mathbb{IF}_{g,g-1}^{att(g,t)})'$, while for each $g' \ne g$, we do not include the very first time period and let $\mathbb{IF}_{g'}^{att(g,t)} =(\mathbb{IF}_{g',2}^{att(g,t)}, \mathbb{IF}_{g',3}^{att(g,t)},\cdots,\mathbb{IF}_{g',g'-1}^{att(g,t)})'$. 
Denote the vector stacking all these vectors together by $\mathbb{IF}^{att(g,t)}_{\text{stg}}$ 
\begin{align} \label{eqn:IF-vector}
    \mathbb{IF}^{att(g,t)}_{\text{stg}} =(\mathbb{IF}_{g'}^{att(g,t),'},g' \in \mathcal{G}_{\text{trt}})',
\end{align}
and let its conditional covariance matrix be defined as $\Omega_{gt}(X) =\Cov (\mathbb{IF}^{att(g,t)}_{\text{stg}}|X)$. 
Let $\Omega_{gt}^{*}(X)$ be a matrix of the same dimension as $\Omega_{gt}$ with $(j,k)$-th element given by
\begin{align} \label{eqn:Omega-star}
    & \frac{1}{p_g(X)} \Cov (Y_{t} - Y_{1} ,Y_{t} - Y_{1} |G=g, X) + \frac{1}{p_\infty(X)} \Cov (Y_{t} - Y_{t_j'} ,Y_{t} - Y_{t_k'} |G=\infty, X) \nonumber \\
    - & \frac{\mathbf{1}\{g=g_j'\}}{p_g(X)} \Cov  (Y_{t} - Y_{1} ,Y_{t_j'} - Y_{1} |G=g,X)  -  \frac{\mathbf{1}\{g=g_k'\}}{p_g(X)} \Cov  (Y_{t} - Y_{1} ,Y_{t_k'} - Y_{1} |G=g,X) \nonumber \\
    + &  \frac{\mathbf{1}\{g_j=g_k'\}}{p_{g'_j}(X)} \Cov  (Y_{t_j'} - Y_{1}, Y_{t_k'} - Y_{1} |G=g_j',X),
\end{align}
where $(g'_s, t'_s)$ is the value that $g'$ and $\tprime$ takes in the $s$-th entry of $\mathbb{IF}^{att(g,t)}_{\text{stg}}$.
Finally, let $q_{g,e} = \P(G=g|G +e \in [2,T])$.

The next Theorem establishes the semiparametric efficiency bound for $ATT(g,t)$ and $ES(e)$ parameters under \ref{asm:pt-n-all} and staggered treatment adoption. We present the semiparametric efficiency bound for the just-identified DiD setup when one imposes Assumption \ref{asm:pt-post} instead of Assumption \ref{asm:pt-n-all} in Corollary \ref{cor:staggered_just_id}.
\begin{theorem} \label{thm:efficiency}
    Suppose that treatment adoption is staggered over time. Under Assumptions \ref{asm:sampling_panel}, \ref{asm:overlap}, \ref{asm:no anticipation} and \ref{asm:pt-n-all}, the efficient influence function of $\pi_g$ and $ATT(g,t)$, $t\ge g$, are given by
    \begin{align*}
        \mathbb{EIF}^{\pi_g} & = G_g - \pi_g, \\
        \mathbb{EIF}^{att(g,t)}_{\text{stg}} & = \frac{\mathbf{1}'\Omega_{gt}(X)^{-1}}{\mathbf{1}'\Omega_{gt}(X)^{-1} \mathbf{1} } \mathbb{IF}^{att(g,t)}_{\text{stg}} = \frac{\mathbf{1}'\Omega_{gt}^{*}(X)^{-1}}{\mathbf{1}'\Omega_{gt}^{*}(X)^{-1} \mathbf{1} } \mathbb{IF}^{att(g,t)}_{\text{stg}}.
    \end{align*}
    Consequently, the efficient influence function of a $ES(e)$, $e\ge0$, is given by
    \begin{align*}
     \mathbb{EIF}^{es(e)}_{\text{stg}} = \sum_{g \in \mathcal{G}_{\text{trt},e}} \left(q_{g,e} \mathbb{EIF}^{att(g,g+e)}_{\text{stg}} + \frac{ATT(g,g+e)}{\sum_{g' \in \mathcal{G}_{\text{trt},e}} \pi_{g'} } \left( (G_{g} - \pi_g ) - q_{g,e} \sum_{s \in \mathcal{G}_{\text{trt},e}}(G_s - \pi_s) \right) \right). 
    \end{align*}
    The semiparametric efficiency bounds are obtained as the second moments of the efficient influence functions, assuming that these second moments are finite.
\end{theorem}

\medskip
\begin{corollary} \label{cor:staggered_just_id}
Suppose that treatment adoption is staggered over time. Under Assumptions \ref{asm:sampling_panel}, \ref{asm:overlap}, \ref{asm:no anticipation} and \ref{asm:pt-post}, the model for $ATT(g,t)$, $t\ge g$, is nonparametrically just-identified, and the efficient influence function of a $ATT(g,t)$, $t\ge g$, is equal to $\mathbb{IF}_{g,g-1}^{att(g,t)}$. In such cases, the efficient influence function of a $ES(e)$, $e\ge0$, is given by
\begin{align*}
         \mathbb{EIF}^{es(e)}_{\text{j-id,stg}} = \sum_{g \in \mathcal{G}_{\text{trt},e}} \left(q_{g,e} \mathbb{IF}_{g,g-1}^{att(g,g+e)} + \frac{ATT(g,g+e)}{\sum_{g' \in \mathcal{G}_{\text{trt},e}} \pi_{g'} } \left( (G_{g} - \pi_g ) - q_{g,e} \sum_{s \in \mathcal{G}_{\text{trt},e}}(G_s - \pi_s) \right) \right). 
    \end{align*}
\end{corollary}

Theorem \ref{thm:efficiency} and Corollary \ref{cor:staggered_just_id} have some practical implications.  First, when one is only willing to assume parallel trends for post-treatment periods, as in Assumption \ref{asm:pt-post}, Corollary \ref{cor:staggered_just_id} shows that this is essentially the same problem as in the single treatment setup and that $ATT(g,t)$ is just-identified: one can only use $g-1$ period as baseline, and the never-treated group as comparison group.

When one can leverage multiple pre-treatment periods as in Assumption \ref{asm:pt-n-all}, Theorem \ref{thm:efficiency} highlights that the efficient influence function is an optimally weighted average of all influence functions based on different baseline periods and different comparison groups. Theorem \ref{thm:efficiency} also shows that the optimal way to pool information across periods and comparison groups is likely to be non-uniform,  highlighting that not every comparison group and pre-treatment period has the same ``relevance'' when it comes to inference. Importantly, these weights are allowed to vary with the data-generating process and are derived as a consequence of our semiparametric efficiency considerations. As no currently available staggered DiD and ES estimator shares these weighting schemes, they do not generally achieve the semiparametric efficiency bound, suggesting that we can potentially conduct more informative inferences.

Theorem \ref{thm:efficiency} also provides a ``blueprint'' on how one can construct DiD and ES estimators that enjoy attractive semiparametric efficiency asymptotic guarantees. Similarly to the single-treatment date, we can explore the efficient influence function to get the following DiD $ATT(g,t)$ estimands:

\begin{align}\label{eqn:ATT_gt_opt_estimand_staggered}
ATT(g,t) = ATT_{\text{stg}}(g,t) := \E\left[\dfrac{\mathbf{1}'\Omega_{gt}^*(X)^{-1}}{\mathbf{1}'\Omega_{gt}^*(X)^{-1} \mathbf{1} }\widetilde{Y}^{\text{att(g,t)}} \right],
\end{align}
where $\widetilde{Y}^{\text{att(g,t)}}$ is a column vector of all (non-collinear) generated outcomes $\widetilde{Y}^{\text{att(g,t)}}_{g',\tprime}$ that leveraged different $(g',\tprime)$ pairs. The efficiency weights are given by  $\mathbf{1}'\Omega_{gt}^*(X)^{-1}\big/\mathbf{1}\Omega_{gt}^*(X)^{-1} \mathbf{1}$. We recommend plotting the expected value of these weights so that one can have a better understanding of how each pre-treatment period and comparison group is leveraged for efficiency considerations. We do this in our simulations and empirical application. 

Based on \eqref{eqn:ATT_gt_opt_estimand_staggered}, one can straightforwardly form estimands for event study parameters $ES(e), e \ge 0$, by plugging \eqref{eqn:ATT_gt_opt_estimand_staggered} into \eqref{eqn:ES_staggered},
\begin{align}\label{eqn:ES_staggered_efficient}
ES(e) = ES_{\text{stg}}(e) := 
\sum_{g \in \mathcal{G}_{\text{trt},e}} \frac{\pi_g}{\sum_{g \in \mathcal{G}_{\text{trt},e}} \pi_g} ATT_{\text{stg}}(g,g+e)
\end{align}
An analogous procedure follows for $ES_{\text{avg}}$.

\begin{remark}[Alternative Parallel Trends]
    In this section, we present semiparametric efficiency results for two types of parallel trends: Assumption \ref{asm:pt-post} that uses never-treated units as the comparison group and does not restrict pre-treatment trends, and Assumption \ref{asm:pt-n-all} that imposes parallel trends in all periods and across all groups. We view these two PT assumptions as the ends of a spectrum, and we recognize that other types of parallel trend assumptions are also possible. It is straightforward to adjust our results in Theorem \ref{thm:efficiency} to cover those intermediate cases, too. For instance, the PT assumption involving not-yet-treated units in \citet{Callaway_Santanna_2021} restricts pre-treatment trends only once the first group is treated at period $g_{\text{min}}$. In that particular case, data from period $t=1$ until $t=g_{\text{min}}-2$ are not informative for any $ATT(g,t)$, so we could drop them entirely from the analysis and apply Theorem \ref{thm:efficiency} using the rest of the retained data. \citet{Gormley_Matsa_2011_stacking}, \citet{Deshpande_Li_2019_AEJ}, \citet{Fadlon2021_AEJApplied}, \citet{Cengiz2019}, and \citet{Baker2022} discuss other parallel trends over a set of pre- and post-treatment periods that allow one to different untreated cohorts as a comparison group. Effectively, all one needs to do to cover these intermediate cases is to adjust the $\mathbb{IF}^{att(g,t)}_{\text{stg}}$ in \eqref{eqn:IF-vector} only to include moment equations that are justified by the given version of the parallel trends assumption. 
\end{remark}

\section{Semiparametric efficient estimation and inference}\label{sec:estimation_inference}
In this section, we discuss how one can estimate and make inferences about the $ATT(g,t)$'s and $ES(e)$'s by leveraging the EIF-based estimands in \eqref{eqn:ATT_gt_opt_estimand_staggered} and \eqref{eqn:ES_staggered_efficient}. The estimator and inference procedures for DiD setups without covariates are straightforward, as they follow from standard plug-in arguments and can be seen as a special case of our proposal. The results from DiD setups with a single treatment time follow as special cases, too.

We first discuss estimation for $ATT(g,t)$'s. The EIF-based estimand $ATT_{\text{stg}}(g,t)$ in \eqref{eqn:ATT_gt_opt_estimand_staggered} naturally suggests a two-step estimation procedure for $ATT(g,t)$'s, where one first estimate the nuisance parameters, $m(X) \coloneqq((m_{\infty,t,\tdoubleprime}(X), \tdoubleprime < t),(m_{g',\tdoubleprime,1}, g' \in \mathcal{G}_{\text{trt}}, \tdoubleprime < g'))$ and $p_{\text{ratio}}(X) \coloneqq(p_g(X)\big/p_{g'}(X), (g,g') \in \mathcal{G}\times \mathcal{G})$, the conditional covariance matrix $\Omega_{gt}^*(X)$ (or $\Omega_{gt}(X)$), and then use their fitted values to form a plug-in estimators fo $ATT(g,t)$ based on \eqref{eqn:ATT_gt_opt_estimand_staggered}. 

When one uses parametric working models for the nuisance parameters, it is easy to show that the resulting estimator is doubly robust in the sense that it remains consistent for the $ATT(g,t)$ as long as the working models for either (but not necessarily both) $p_{\text{ratio}}(X)$ or $m(X)$  are correctly specified, regardless of the weighting scheme adopted; see, e.g., \citet[Theorem 1]{Santanna2020} and \citet[Theorem 2]{Callaway_Santanna_2021}.\footnote{In our case, as our estimator averages across different working models for the propensity score and regression adjustment for untreated units, it is possible to refine the notion of double robustness in the sense that our $ATT(g,t)$ estimator remains consistent if \emph{weighted averages} of functionals of $p_{\text{ratio}}(X)$ or of $m(X)$ are consistently estimated. As we advocate for an efficiency-oriented weighting scheme, we would not require it to hold for all possible weighted averages (which would essentially require that $m(X)$ or $p_{\text{ratio}}(X)$ to be correctly specified).} Inference, in this case, follows from delta-method arguments.

In practice, however, parametric models may be too restrictive, and their misspecification can still lead to asymptotic biases. Thus, in what follows, we describe a flexible nonparametric estimation procedure that can leverage modern and traditional estimators for the nuisance functions.

First, recall that $m(X)$ involves conditional expectations of changes in outcomes over time among some untreated units, given covariates, i.e., it involves terms like $ m_{g',\tdoubleprime,\tprime}(X) = \mathbb{E}[Y_{\tdoubleprime} - Y_{\tprime} |G=g',X]$ with $ g'> \max\{\tprime,\tdoubleprime\}$. This is nothing more than a nonparametric regression problem, and we can use a variety of estimators for it, including sieve-based \citep{Chen2007_Handbook}, kernel-based \citep{Cattaneo_etal_npreg_JASA_2018}, or recent machine learning methods such as random forests, lasso, ridge, deep neural nets, boosted trees, and the ensemble of these methods \citep{Chernozhukov2018_DML}. Denote these estimators as $\widehat{m}(X)$.

Analogously, $p_{\text{ratio}}(X)$ involves terms like ratios of $\mathbb{E}[G_g|X]$, and each of these conditional expectations can also be estimated using these same estimation procedures.\footnote{It is often desirable to impose that $p_g(X)$ is bounded between zero and one and choose estimators that respect this constraint, such as multinomial series logit/probit and local multinomial logit/probit estimators. See, e.g., \citet{Chen2007_Handbook} for a discussion of likelihood-based sieve estimators, and \citet{Staniswalis_kernel_logit_1989} for kernel-based ones.} Note that modeling each propensity score separately and then constructing their ratios can potentially lead to instabilities associated with estimated propensity scores close to zero. As the propensity scores enter into $ATT_{\text{stg}}(g,t)$ \eqref{eqn:ATT_gt_opt_estimand_staggered} in a ratio format, it suffices to directly estimate $p_g(X)/p_{g'}(X)$, which is not bounded between zero and one, can be more flexibly modeled, and can lead to more stable procedures. 

Towards this end, note that $p_g(X)/p_{g'}(X)$ can be found as the unique solution to the minimization problem:
\begin{align} \label{eqn:loss-function-propensity-ratio}
    \frac{p_g(X)}{p_{g'}(X)} = \argmin_{r_{g,g'}} \mathbb{E}\left[ \left(r_{g,g'}(X) - \frac{p_g(X)}{p_{g'}(X)}\right)^2 p_{g'}(X) \right] = \argmin_{r_{g,g'}} \mathbb{E}\left[ r_{g,g'}(X)^2 G_{g'} - 2 r_{g,g'}(X) G_g \right].
\end{align}
This equivalence allows the use of \(\mathbb{E}\left[ r_{g,g'}(X)^2 G_{g'} - 2 r_{g,g'}(X) G_g \right]\) as a (convex) loss function to estimate $r_{g,g'}(X):=p_g(X)/p_{g'}(X)$, and different estimation procedures can be used. An easy-to-implement estimator for $r_{g,g'}(X)$ is the sieve-based estimator $\widehat{r}_{g,g'}(X) = {\left(\psi^K(X)'\widehat{\beta}_K\right)}$, where $\psi^K(x)$ is a $K$-dimensional  vector of flexible transformations of the $X$ such as (tensor products of) cubic B-splines,
\begin{align} \label{eqn:series_ratio_pscore}
    \widehat{\beta}_K = \argmin_{\beta_K} \mathbb{E}_n\left[G_{g'} \left(\psi^K(X)'{\beta}_K\right)^2  - 2G_g \left(\psi^K(X)'{\beta}_K\right) \right],
\end{align}
and we avoid indexing $\widehat{\beta}_K$ and ${\beta}_K$ by $(g,g')$ to simplify the notation burden. 
The ($(g,g')$-specific) sieve index \( K \) can be selected using some information criteria as follows:  
\[
    \widehat{K} = \argmin_{K} 2\mathbb{E}_n\left[G_{g'} \left(\psi^K(X)'\widehat{\beta}_K\right)^2  - 2G_g \left(\psi^K(X)'\widehat{\beta}_K\right) \right] + \frac{C_n K}{n},
\]  
where the information criterion corresponds to the Akaike information criterion (AIC) when \( C_n = 2 \), or the Bayesian information criterion (BIC) when \( C_n = \log(n) \). In the appendix, we show the consistency of the estimator based on $\widehat{K}$ following \citet{Chen_Liao_2014_JoE}. Similarly to \eqref{eqn:loss-function-propensity-ratio}, we can estimate $s_{g'} :=1/p_{g'}(X)$ by solving the empirical analog of the minimization problem
\begin{align*} 
    s_{g'}(X) = \argmin_{a_g} \mathbb{E}\left[ a_g(X)^2 G_{g} - 2 a_g(X) \right].
\end{align*}
For instance, one can use sieve-estimators analogous to \eqref{eqn:series_ratio_pscore} for this task.

To estimate the conditional covariance terms of $\Omega_{gt}^*$ as defined in \eqref{eqn:Omega-star}, we propose using a Nadaraya-Watson-type estimator based on kernel smoothing. Let $Ker$ be a kernel function on the covariates space and $h>0$ a bandwidth. Denote $K_h(\cdot) =Ker(\cdot / h)/h$. For $x = X_i$, we use  $ \widehat{\Cov} (Y_t - Y_{\tprime},Y_t - Y_{\tdoubleprime}|G=g',X=x)$ as an estimator for each the covariance terms $\Cov (Y_t - Y_{\tprime},Y_t - Y_{\tdoubleprime}|G=g',X=x)$, where  $ \widehat{\Cov} (Y_t - Y_{\tprime},Y_t - Y_{\tdoubleprime}|G=g',X=x)$  is define as
\smallskip
\begin{align*}
     \dfrac{\sum_{i' : G_{i'} = g' } K_h(X_{i'} - x) (Y_{i',t} - Y_{i',\tprime} - \widehat{m}_{g',t,\tprime}(X_{i'}))(Y_{i',t} - Y_{i',\tdoubleprime} - \widehat{m}_{g',t,\tdoubleprime}(X_{i'}))}{ \sum_{i' : G_{i'} = g' } K_h(X_{i'} - x)}.
\end{align*}
Based on these estimators for the propensity score ratios and the conditional covariance terms, we estimate $\Omega^*_{gt}$ by $\widehat{\Omega}^*_{gt}$ with each $(j,k)$-th element given by the plug-in estimators of \eqref{eqn:Omega-star}.

Based on these estimators for the nuisance functions and weights, we can estimate $ATT(g,t)$ by\footnote{It is straightforward to consider variants of our estimator using sample-splitting arguments, too. }
\begin{align}
\widehat{ATT}_{\text{stg}}(g,t)= \E_n\left[ \dfrac{\mathbf{1}'\widehat{\Omega}_{gt}^{*}(X)^{-1}}{\mathbf{1}'\widehat{\Omega}_{gt}^{*}(X)^{-1} \mathbf{1} }\widehat{\widetilde{Y}}^{\text{att(g,t)}}_{\text{stg}}\right], \label{eqn:ATT-hat}
\end{align}
where 
$\widehat{\widetilde{Y}}^{\text{att(g,t)}}_{\text{stg}}$ the estimated analog of ${\widetilde{Y}}^{\text{att(g,t)}}_{\text{stg}}$, with each $ \widehat{\widetilde{Y}}^{\text{att(g,t)}}_{g',\tprime}$ given by
\begin{align}
\widehat{\widetilde{Y}}^{\text{att(g,t)}}_{g',\tprime} =&  \dfrac{G_g}{\widehat{\pi}_g} ( Y_t - Y_1 -  \widehat{m}_{\infty,t,\tprime}(X) -  \widehat{m}_{g',\tprime,1}(X))\nonumber\\
    & - \widehat{r}_{g,\infty}(X)\dfrac{G_\infty}{\widehat{\pi}_g}(Y_t - Y_{\tprime} - \widehat{m}_{\infty,t,\tprime}(X)) - \widehat{r}_{g,g'}(X)\dfrac{G_{g'}}{\widehat{\pi}_g}(Y_{\tprime} - Y_{1} - \widehat{m}_{g',\tprime,1}(X)) \Big),\label{eqn:generic_att_gt_stg_estimated}
\end{align}
and $\widehat{\pi}_g = \mathbb{E}_n[G_{g}]$. 

The estimator for the event study parameter is
\begin{align}
    \widehat{ES}(e) =\sum_{g \in \mathcal{G}_{\text{trt}}} \frac{\widehat{\pi}_g}{\sum_{g' \in \mathcal{G}_{\text{trt},e}} \widehat{\pi}_{g'}} \widehat{ATT}_{\text{stg}}(g,g+e).\label{eqn:Event_study_efficient}
\end{align}

The following theorem establishes the large sample properties of our proposed DiD and ES estimators, and highlights that they achieve the semiparametric efficiency bound.

\begin{theorem} \label{thm:asy_properties_estimator}
Let Assumptions \ref{asm:overlap}, \ref{asm:no anticipation} and \ref{asm:pt-n-all} and the regularity conditions listed in Assumption \ref{asm:regularity} in the Appendix hold. Then, our proposed nonparametric estimator for $\widehat{ATT}_{\text{stg}}(g,t)$  is consistent, asymptotically normal, and achieves the semiparametric efficiency bound, i.e., as $n \rightarrow \infty$,
    \begin{align*}
        \sqrt{n}(\widehat{ATT}_{\text{stg}}(g,t) - ATT(g,t)) & = \frac{1}{\sqrt{n}} \sum_{i=1}^n \dfrac{\mathbf{1}'\widehat{\Omega}_{gt}^{*}(X_i)^{-1}}{\mathbf{1}'\widehat{\Omega}_{gt}^{*}(X_i)^{-1} \mathbf{1} } \mathbb{IF}^{att(g,t)}_{\text{stg}}(W_i) + o_p(1) \\
        & \overset{d}{\rightarrow} N(0, V_{\text{eff}}),
    \end{align*}
    with $ V_{\text{eff}}= \Var \left(\frac{\mathbf{1}'\widehat{\Omega}_{gt}^{*}(X)^{-1}}{\mathbf{1}'\widehat{\Omega}_{gt}^{*}(X)^{-1} \mathbf{1} } \mathbb{IF}^{att(g,t)}_{\text{stg}}(W))\right)$ being the semiparametric efficiency bound in Theorem \ref{thm:efficiency}. Consequently, $\widehat{ES}(e)$ is also asymptotically normal and semiparametrically efficient.
    
\end{theorem}
\bigskip

Theorem \ref{thm:asy_properties_estimator} provides the pointwise asymptotic normality results for each $(g,t)$ pair, with $t\ge g$. In practice, researchers are often interested in making inferences about multiple groups and periods to better understand treatment effect dynamics and heterogeneity. In such cases, to avoid multiple-testing problems, it is recommended to construct simultaneous confidence bands. We omit the details of this procedure as it is standard.\footnote{See, e.g., Theorems 2 and 3, Algorithm 1, and Corollaries 1 and 2 in \citet{Callaway_Santanna_2021}.} To compute standard errors, one can leverage a multiplier bootstrap procedure or take the square root of the average of the estimated EIF squared divided by the sample size. 

\begin{remark}
    Although the results in Theorem \ref{thm:asy_properties_estimator} focus on the large sample properties of our efficient DiD estimators, $\widehat{ATT}_{\text{stg}}(g,t)$, in its proof, we also establish consistency and asymptotically normality for related estimators for the $ATT(g,t)$'s that replaces the estimators for the efficient weights ${\mathbf{1}'\widehat{\Omega}_{gt}^{*}(X)^{-1}}\big/{\mathbf{1}'\widehat{\Omega}_{gt}^{*}(X)^{-1} \mathbf{1} } $ with any consistent estimator for any weights $w(x)$ that sum up to one for each covariate value $x$. Interestingly, we do not need to require rate conditions for these weights, allowing one to use many different weights as discussed in Lemma \ref{lem:id-attgt-over}. 
\end{remark}
\begin{remark}
   As our proposed estimators are based on efficient influence functions, they satisfy the Neyman orthogonality condition by construction. This implies that it is relatively straightforward to use modern machine learning estimators for the nuisance parameter without incurring a loss of asymptotic semiparametric efficiency in large samples.\footnote{Many machine learning procedures involve Neyman orthogonal moments with sample-splitting, which can lead to a loss of precision in DiD and ES estimation in small samples. This phenomenon is not unique to DiD; see, e.g., \citet{Chen_Chen_Tamer_2023_JoE} for a discussion.}  
\end{remark}

\subsection{Efficient estimation when covariates are not present}

When Assumptions \ref{asm:no anticipation} and \ref{asm:pt-n-all} hold unconditionally, the efficiency results can be simplified by dropping $X$ from the conditioning set. This leads to an even simpler construction of the efficient estimator. In particular, (\ref{eqn:ATT_gt_opt_estimand_staggered}) reduces to
\begin{align*}
ATT(g,t) = \dfrac{\mathbf{1}'(\Omega_{gt}^*)^{-1}}{\mathbf{1}'(\Omega_{gt}^*)^{-1} \mathbf{1} }\widetilde{Y}^{\text{att(g,t)}},
\end{align*}
where each entry $\widetilde{Y}^{\text{att(g,t)}}_{g',\tprime}$ of $\widetilde{Y}^{\text{att(g,t)}}$ is given by
\begin{align*}
    \widetilde{Y}^{\text{att(g,t)}}_{g',\tprime} = \mathbb{E}[Y_t - Y_1 | G=g] - (\mathbb{E}[Y_t - Y_{\tprime} | G=\infty] + \mathbb{E}[Y_{\tprime} - Y_{1} | G=g']).
\end{align*}
The covariance matrix $\Omega_{gt}^{*}$ has $(j,k)$-th element given by
\begin{align*} 
    & \frac{1}{\pi_g} \Cov (Y_{t} - Y_{1} ,Y_{t} - Y_{1} |G=g) + \frac{1}{\pi_\infty} \Cov (Y_{t} - Y_{t_j'} ,Y_{t} - Y_{t_k'} |G=\infty) \nonumber \\
    - & \frac{\mathbf{1}\{g=g_j'\}}{\pi_g} \Cov  (Y_{t} - Y_{1} ,Y_{t_j'} - Y_{1} |G=g)  -  \frac{\mathbf{1}\{g=g_k'\}}{\pi_g} \Cov  (Y_{t} - Y_{1} ,Y_{t_k'} - Y_{1} |G=g) \nonumber \\
    + &  \frac{\mathbf{1}\{g_j=g_k'\}}{\pi_{g'_j}} \Cov  (Y_{t_j'} - Y_{1}, Y_{t_k'} - Y_{1} |G=g_j'),
\end{align*}
where the indices $(g'_s, t'_s)$ are as defined in (\ref{eqn:Omega-star}). 

In estimation, one can replace the expectations and covariances for each group with the corresponding within-group sample means and sample covariances. This procedure does not involve choosing tuning parameters or estimating conditional expectations.

When covariates are absent, the efficient GMM estimator also attains the semiparametric efficiency bound. Our DiD and ES estimators match that efficiency yet are substantially simpler to compute: they are available in closed form, whereas the optimal GMM approach entails solving a high-dimensional optimization problem that is often not practical. In addition, the construction of our efficient estimator provides new insights about how to weigh different pre-treatment and comparison groups to achieve efficiency---it is only a matter of analyzing the efficient weights ${\mathbf{1}'(\Omega_{gt}^*)^{-1}}\big/{\mathbf{1}'(\Omega_{gt}^*)^{-1} \mathbf{1} }$.

\section{Monte Carlo simulations}\label{sec:Simulations}
So far, we have motivated and established that our efficient DiD and ES estimators have attractive efficiency properties that allow researchers to get more precise inference procedures. In this section, we aim to see how these properties play out in realistic empirical settings. We build on \citet{Arkhangelsky2021_SDiD} and \citet{Baker2022} and consider two sets of simulation studies calibrated to datasets people have used for DiD. The first set of simulations builds on \citet{Arkhangelsky2021_SDiD}. It considers a single treatment date and leverages Current Population Survey (CPS) data for constructing a variety of outcomes and treatment assignments for a panel of US states. The second set of simulations builds on \citet{Baker2022}. It explores Compustat data to calibrate outcomes and treatment assignments in staggered treatment setups with treatment timing varying across states. To simplify exposure and match the simulation designs in \citet{Arkhangelsky2021_SDiD} and \citet{Baker2022}, covariates play no important role in these Monte Carlo setups. 

In our first set of simulations, we consider $ES_\text{avg}$ as in \eqref{eqn:overall_ATT} as the parameter of interest and compare our efficient DiD plug-in estimator based on \eqref{eqn:EI_ES} (EDiD), traditional OLS-based TWFE estimator for $\beta$ based on \eqref{eqn:static-twfe} (TWFE), average of post-treatment event-study TWFE estimators $\beta_e$ based on \eqref{eqn:dynamic-twfe} (DTWFE), and the \citet{Arkhangelsky2021_SDiD} Synthetic DiD estimator (SDiD). We note that the DTWFE estimator would be the efficient estimator if parallel trends were to hold only in post-treatment periods.\footnote{In this non-staggered setup, the OLS estimate of $\beta$ based on \eqref{eqn:static-twfe} coincides with the post-treatment average of the event-studies coefficients of the estimators proposed by \citet{Wooldridge2021a}, \citet{Gardner2021}, and \citet{Borusyak2023}. Similarly, the post-treatment average of the OLS estimates of $\beta_e$ in \eqref{eqn:dynamic-twfe} coincides with \citet{Callaway_Santanna_2021}, \citet{Sun2021}, and \citet{deChaisemartin2020_AER} estimators.} 

In our second set of simulations, we also consider $ES_\text{avg}$ as in \eqref{eqn:overall_ATT} as the parameter of interest and compare our efficient (plug-in) estimator $\widehat{ES}_{\text{avg}}$ based on \eqref{eqn:Event_study_efficient} (EDiD), the average of post-treatment event-study estimates based on (i) \citet{Callaway_Santanna_2021} and \citet{Sun2021} DiD estimators using the ``never-treated'' as a comparison group (CS-SA), (ii) \citet{Callaway_Santanna_2021} and \citet{deChaisemartin2020_AER} DiD estimators using the ``not-yet-treated'' as a comparison group (CS-dCDH), and (iii) \citet{Wooldridge2021a}, \citet{Gardner2021}, and \citet{Borusyak2023} ``imputation'' estimators (BJS-G-W). We do not include staggered synthetic DiD estimators in the comparison as we are not aware of a paper describing and establishing the statistical properties of synthetic DiD estimators for $ES_\text{avg}$. In this exercise, we note that the CS-SA estimator would be the efficient one if parallel trends were to hold only for post-treatment periods and the never-treated group was the only valid comparison group.

We compare all these estimators regarding bias, root-mean-squared-error (RMSE) relative to our efficient DiD estimator, $95\%$ empirical coverage, and $95\%$ confidence interval length relative to our efficient DiD. We consider analytical and bootstrapped standard errors for coverage and length of confidence intervals but use the Gaussian critical values. For each data generating process (DGP), we consider 1,000 Monte Carlo experiments, and for each experiment, we use 300 bootstrap draws.\footnote{For Synthetic DiD, \citet{Arkhangelsky2021_SDiD} have not proposed a plug-in analytical standard error. As such, we do not report plug-in analytical standard errors from it.} 

\subsection{Simulations based on CPS with single treatment date}

Our first set of simulations builds on \citet{Arkhangelsky2021_SDiD} and explores CPS data to construct an empirically motivated data-generating process. We differ from \citet{Arkhangelsky2021_SDiD} in some aspects: (i) we consider heterogeneous treatment effects across units; (ii) consistent with our DiD regime with ``large $n$'' and fixed $T$, we consider short panels with $T=7$---four pre-treatment and three post-treatment periods; (iii) we do not limit the maximum number of treated units in a given simulation; and (iv) all our outcomes are measured in log to avoid violating support restrictions in the simulation. Apart from these differences, the construction of our data-generating process follows from \citet{Arkhangelsky2021_SDiD}, which we describe below for completeness.

As in \citet{Arkhangelsky2021_SDiD}, we start designing our simulations using data on wages for women with positive wages in the March outgoing rotation groups in the CPS from 1979 to 2018. We then take logs and average the observations by state-year cells, leaving us with aggregate data from 50 states in 40 years. This will serve as the baseline dataset for our simulation designs. 

Next, we consider a baseline specification for our untreated potential outcomes, $Y_{i,t}(\infty)$. We consider the same interactive fixed-effects specification as \citet{Arkhangelsky2021_SDiD},
\begin{align}
    Y_{i,t} (\infty)  = \gamma_i \upsilon_t +  \varepsilon_{i,t},  \label{eqn:IFE}
\end{align}
with $\gamma_i$ and $\upsilon_t$ being 4-dimensional vectors of latent unit and time factors, and $\varepsilon_{i,t}$ being a mean-zero Gaussian error term that follows an $AR(2)$ process. This specification does not impose a two-way fixed effects structure on $Y_{i,t} (\infty)$ and allows for correlation over times within each state; accounting for serial correlation is very important when conducting valid inference using DiD and other panel data methods, see, e.g., \citet{Wooldridge2003} and \citet{Bertrand2004}.  

To construct a realistic set of simulations, we use the CPS aggregated data to estimate the $\gamma_i$'s, $\upsilon_t$'s, and the variance-covariance matrix of the error terms. In matrix notation, the interactive fixed-effects model \eqref{eqn:IFE} is given by $ Y(\infty) = L + E$, where $L = \Gamma \Upsilon',$ and we estimate $L$ as 
\begin{align*}
    L := \argmin_{L:rank(L)=4} \sum_{i,t}(Y_{i,t}^* - L_{i,t})^2,
\end{align*}
where $Y_{i,t}^*$ is the log wage in state $i$ in year $t$ in the CPS data. To help with interpretation, we decompose $L$ as unit-and-time fixed effects, $F$, and an interactive term $M$, with
\begin{align*}
    F_{i,t}  &= \alpha_t + \eta_i = \frac{1}{T} \sum_{l=1}^T L_{i,l} + \frac{1}{N} \sum_{j=1}^N L_{j,t} - \frac{1}{NT} \sum_{i,t}^{N\times T} L_{i,t} \\
    M_{i,t} &= L_{i,t} - F_{i,t}.
\end{align*}
We estimate the variance-covariance matrix of the residuals $\varepsilon_{i,t}$ by fitting an AR(2) model to the residuals of $Y_{i,t}^* - L_{i,t}$, and then compute the implied variance-covariance matrix of the error term for state $i$, $\Sigma$. Like \citet{Arkhangelsky2021_SDiD}, we assume that the error terms are independent across units and $\Sigma$ is constant across states.

Next, we describe the treatment assignment process. In this set of simulations, units can start treatment in $t=2009$, so we only have two groups: $G_i=2009$ (treated units) and $G_i = \infty$ (untreated units). To simulate whether a state is treated, we follow \citet{Arkhangelsky2021_SDiD} and consider the treatment status $D_{i,t} = 1\{t\ge 2009\} 1\{G_i = 2009\}$ with  
\begin{align*}
1\{G_i = 2009\} &\sim  \text{Bernoulli } (\pi_i), \\
   \pi_i = P(G_i = 2009 | \eta_i, M_i) &= \dfrac{\exp(\phi_\eta \eta_i + \phi_M M_i)}{1 + \exp(\phi_\eta \eta_i + \phi_M M_i)}.
\end{align*}
We choose $\phi_\eta$ and $\phi_M$ as the coefficient estimates from a logistic regression of an observed binary characteristic of the state $i$ on $\eta_i$ and $M_i$.\footnote{These ``unit factors'' are the first four left singular vectors from single value decomposition of the outcome $Y$, as described in \citet{Arkhangelsky2021_SDiD}.} We consider three different characteristics relating treatment groups to minimum wage laws (part of our baseline specification), abortion rights, and gun control laws. We also consider a completely random treatment assignment. Unlike \citet{Arkhangelsky2021_SDiD}, though, we do not restrict the maximum number of treated units; only the minimum number of treated units to be two.\footnote{If we have less than two treated states in a simulation, we randomly select ten out of the 50 states to be treated completely at random. This is similar to \citet{Arkhangelsky2021_SDiD}, though their minimum number of treated states was one.}

Based on these parameters, we simulate untreated and treated potential outcomes $Y_{i,t}(\infty)$ and $Y_{i,t}(2009)$ as 
\begin{align}
    Y_{i,t}(\infty) &= \alpha_t + \eta_i + M_{i,t} + \varepsilon_{i,t}\\
    Y_{i,t}(2009) &=  Y_{i,t}(\infty) + \tau_i 1\{t\ge 2009\}
\end{align}
where $\varepsilon_{i} = (\varepsilon_{i,t=1},\dots, \varepsilon_{i,t=T})'$ have a multivariate Gaussian distribution with mean zero and variance $\Sigma$, and $\tau_i$ have a Gaussian distribution with mean zero and variance one. Note that we have a heterogeneous treatment effects model as the treatment effect varies across states. This differs from \citet{Arkhangelsky2021_SDiD}, who sets $\tau_i = 0$ for all states. However, we stress that $ATT(2009,t) = 0$ for all periods. The observed outcomes for unit $i$ in time $t$ are $Y_{i,t} = 1\{G_i=2009\} Y_{i,t}(2009) + 1\{G_i=\infty\} Y_{i,t}(\infty)$. To respect our DiD setup with ``large $n$ and fixed $T$'', we restrict $T=7$ and keep data from $t=2005$ until $t=2011$. 

Note that we do not impose that the parallel trends assumption holds exactly across periods. If this assumption is violated, we should see a bias in estimates that rely on this assumption, namely our efficient DiD estimator and the DTWFE in \eqref{eqn:dynamic-twfe}. If parallel trends in post-treatment periods is violated, we should also see a bias in the TWFE estimates based on \eqref{eqn:static-twfe}.

Like \citet{Arkhangelsky2021_SDiD}, we consider different choices of treatment assignment and outcome variables. In addition, we also consider settings where we drop different components of the data generating process for the outcome, such as assuming that the error terms are serially independent (``No Corr''), that there is no interactive component $M$ (``No M''), no additive two-way fixed effects (``No F''), or there is $L$ (``Only noise'').\footnote{We omit results for DGPs in which there is no error term (``No noise''). We do it because the RMSE and the length of 95\% confidence interval for our efficient estimators are close to zero, making it hard to report relative performance measures.} For the alternative outcome variables, we consider $log$ of hours worked and $log$ of unemployment rate; \citet{Arkhangelsky2021_SDiD} consider these variables in levels, but that leads to some negative values during the simulation for these variables, violating their natural support restriction.

Finally, we consider two setups. We simulate data from the 50 states and seven periods for the first one. In this case, we expect that inference procedures based on asymptotic approximations may be challenging given the limited (effective) sample size of 50. In the second setup, we increase the number of cross-sectional units by drawing with replacement 200 states (each with its own error term, treatment assignment, and treatment effect). Here, we expect inference procedures based on asymptotic approximations to be better than in the setup with $n=50$, though we acknowledge that $n=200$ is still reasonably small.
\begin{table}[!h]
\centering
\caption{\label{tab:sim_placebo_synth}Simulation results for CPS data: Relative RMSE and Bias}
\centering
\resizebox{\ifdim\width>\linewidth\linewidth\else\width\fi}{!}{
\begin{threeparttable}
\begin{tabular}[t]{lcrrrrrrrrr}
\toprule
\multicolumn{2}{c}{ } & \multicolumn{4}{c}{Relative RMSE} & \phantom{abc} & \multicolumn{4}{c}{Bias ($\times 10$)} \\
\cmidrule(l{3pt}r{3pt}){3-6} \cmidrule(l{3pt}r{3pt}){8-11}
\multicolumn{1}{c}{} & \multicolumn{1}{c}{Sample size} & \multicolumn{1}{c}{EDiD} & \multicolumn{1}{c}{TWFE} & \multicolumn{1}{c}{DTWFE} & \multicolumn{1}{c}{SDiD} & \multicolumn{1}{c}{} & \multicolumn{1}{c}{EDiD} & \multicolumn{1}{c}{TWFE} & \multicolumn{1}{c}{DTWFE} & \multicolumn{1}{c}{SDiD}\\
\midrule
1. Baseline & 50 & 1 & 3.57 & 12.46 & 1.53 &  & 0.01 & 0.60 & 2.58 & 0.00\\
 & 200 & 1 & 2.32 & 3.37 & 1.95 &  & 0.00 & -0.01 & 0.00 & 0.00\\
\noalign{\vskip 1mm} \it{Outcome Model}\\
2. No corr & 50 & 1 & 3.52 & 12.33 & 1.45 &  & 0.02 & 0.59 & 2.46 & 0.00\\
 & 200 & 1 & 2.27 & 3.33 & 1.95 &  & -0.01 & 0.00 & -0.01 & 0.00\\
\noalign{\vskip -3mm}\\
3. No M & 50 & 1 & 3.67 & 13.04 & 1.49 &  & -0.02 & 0.61 & 2.45 & 0.03\\
 & 200 & 1 & 2.17 & 3.00 & 1.68 &  & -0.01 & 0.02 & 0.00 & 0.01\\
\noalign{\vskip -3mm}\\
4. No F & 50 & 1 & 1.47 & 1.88 & 1.42 &  & 0.00 & -0.02 & -0.04 & -0.02\\
 & 200 & 1 & 1.64 & 2.18 & 1.63 &  & 0.00 & 0.00 & -0.01 & 0.00\\
\noalign{\vskip -3mm}\\
5. Only noise & 50 & 1 & 1.26 & 1.67 & 1.25 &  & 0.02 & -0.02 & -0.01 & -0.01\\
 & 200 & 1 & 1.47 & 1.87 & 1.48 &  & 0.00 & -0.01 & -0.01 & -0.01\\
\noalign{\vskip -3mm}\\
\noalign{\vskip 1mm} \it{Treatment Assignment}\\
6. Gun law & 50 & 1 & 7.27 & 18.23 & 4.67 &  & 0.00 & -0.08 & -0.23 & -0.11\\
 & 200 & 1 & 8.70 & 12.94 & 6.05 &  & 0.00 & 0.03 & 0.02 & 0.02\\
7. Abortion & 50 & 1 & 6.99 & 17.19 & 4.77 &  & -0.01 & 0.55 & 1.75 & 0.33\\
 & 200 & 1 & 8.04 & 12.64 & 5.23 &  & 0.00 & -0.01 & -0.05 & 0.00\\
\noalign{\vskip -3mm}\\
8. Random & 50 & 1 & 4.13 & 15.30 & 1.53 &  & -0.01 & -0.07 & -0.22 & 0.01\\
 & 200 & 1 & 2.56 & 3.74 & 2.17 &  & 0.00 & 0.01 & 0.00 & 0.01\\
\noalign{\vskip 1mm} \it{Outcome Variable}\\
9. Ln Hours & 50 & 1 & 1.01 & 1.92 & 0.95 &  & -0.34 & 0.12 & 1.53 & 0.02\\
 & 200 & 1 & 1.24 & 1.95 & 1.21 &  & 0.01 & 0.05 & 0.02 & 0.07\\
10. Ln U-rate & 50 & 1 & 0.82 & 1.44 & 0.82 &  & 0.73 & -0.24 & -0.36 & -0.26\\
 & 200 & 1 & 1.03 & 1.53 & 1.01 &  & 0.03 & 0.00 & 0.01 & 0.00\\
\bottomrule
\end{tabular}
\begin{tablenotes}[para]
\item \vspace{-5ex} \singlespacing \footnotesize{Notes: Simulation results for CPS data based on 1,000 Monte Carlo experiments.
                               The DGP is summarized in Section \ref{sec:Simulations}. The baseline case uses state minimum wage laws to simulate treatment
                               assignment and generates outcomes using the DGPs described in Section \ref{sec:Simulations}.  In subsequent settings, we omit parts of
                               the DGP (rows 2–5), consider different distributions for the treatment variable (rows 6-8),
                               and different distributions for the outcome variable (rows 9-10). The dataset has 7 times periods and either
                               $n=50$ or $n=200$. Relative RMSE is reported using our efficient DiD estimator as benchmark.
                             All results for bias are multiplied by 10 for readability.}
\end{tablenotes}
\end{threeparttable}}
\end{table}

Tables \ref{tab:sim_placebo_synth} present the bias and relative RMSE of the four different estimators for $ES_\text{avg}$ when $n=50$ and $n=200$. At a high level, we have found that all estimators are (nearly) unbiased when $n=50$ and that the bias becomes even smaller when $n=200$. This suggests that, in these DGPs, the assumption that parallel trends hold in all periods is a good approximation. In terms of RMSE, however, there is a significant variation across estimators, with our efficient DiD estimator performing the best in nearly all DGPs, with Synthetic DiD coming as second, and the DTWFE estimator based on \eqref{eqn:dynamic-twfe} being the least precise. The result for DTWFE is expected, though, as this estimator does not leverage data from any pre-treatment period other than $t=2008$; all the other estimators leverage three additional pre-treatment periods. It is also important to stress the large RMSE gains that our efficient DiD estimator enjoys compared to the synthetic DiD estimator. For instance, when $n=50$, the RMSE of the synthetic DiD estimator is nearly 60\% larger than that of our efficient DiD estimator in the baseline model. When treatment assignment is based on gun laws or abortion rights, their RMSE is approximately 5 times larger than ours when $n=50$, and this difference grows further when $n=200$. On the other hand, we notice that when the outcome of interest is the log unemployment rate, the synthetic DiD and the TWFE estimators have smaller RMSE than our efficient estimator. However, this difference reduces when $n=200$. 

\begin{table}[!h]
\centering
\caption{\label{tab:sim_placebo_synth_coverage}Simulation results for CPS data: Coverage}
\centering
\resizebox{\ifdim\width>\linewidth\linewidth\else\width\fi}{!}{
\begin{threeparttable}
\begin{tabular}[t]{lcrrrrrrrrr}
\toprule
\multicolumn{2}{c}{ } & \multicolumn{4}{c}{Bootstrap} & \phantom{abc} & \multicolumn{4}{c}{Analytical} \\
\cmidrule(l{3pt}r{3pt}){3-6} \cmidrule(l{3pt}r{3pt}){8-11}
\multicolumn{1}{c}{} & \multicolumn{1}{c}{Sample size} & \multicolumn{1}{c}{EDiD} & \multicolumn{1}{c}{TWFE} & \multicolumn{1}{c}{DTWFE} & \multicolumn{1}{c}{SDiD} & \multicolumn{1}{c}{} & \multicolumn{1}{c}{EDiD} & \multicolumn{1}{c}{TWFE} & \multicolumn{1}{c}{DTWFE} & \multicolumn{1}{c}{SDiD}\\
\midrule
1. Baseline & 50 & 0.94 & 0.92 & 0.93 & 0.93 &  & 0.80 & 0.92 & 0.92 & -\\
 & 200 & 0.96 & 0.97 & 0.97 & 0.97 &  & 0.94 & 0.98 & 0.97 & -\\
\noalign{\vskip 1mm} \it{Outcome Model}\\
2. No corr & 50 & 0.92 & 0.94 & 0.94 & 0.93 &  & 0.80 & 0.93 & 0.92 & -\\
 & 200 & 0.93 & 0.97 & 0.97 & 0.96 &  & 0.93 & 0.97 & 0.97 & -\\
\noalign{\vskip -3mm}\\
3. No M & 50 & 0.94 & 0.93 & 0.93 & 0.94 &  & 0.82 & 0.92 & 0.92 & -\\
 & 200 & 0.93 & 0.97 & 0.97 & 0.96 &  & 0.91 & 0.96 & 0.97 & -\\
\noalign{\vskip -3mm}\\
4. No F & 50 & 0.94 & 0.94 & 0.94 & 0.95 &  & 0.82 & 0.94 & 0.92 & -\\
 & 200 & 0.95 & 0.95 & 0.96 & 0.96 &  & 0.94 & 0.95 & 0.96 & -\\
\noalign{\vskip -3mm}\\
5. Only noise & 50 & 0.93 & 0.93 & 0.92 & 0.95 &  & 0.79 & 0.92 & 0.90 & -\\
 & 200 & 0.94 & 0.95 & 0.95 & 0.95 &  & 0.93 & 0.94 & 0.95 & -\\
\noalign{\vskip -3mm}\\
\noalign{\vskip 1mm} \it{Treatment Assignment}\\
6. Gun law & 50 & 0.95 & 0.94 & 0.94 & 0.97 &  & 0.94 & 0.94 & 0.94 & -\\
 & 200 & 0.94 & 0.96 & 0.96 & 0.97 &  & 0.94 & 0.96 & 0.96 & -\\
7. Abortion & 50 & 0.94 & 0.92 & 0.91 & 0.93 &  & 0.92 & 0.93 & 0.91 & -\\
 & 200 & 0.95 & 0.98 & 0.98 & 0.97 &  & 0.95 & 0.97 & 0.98 & -\\
\noalign{\vskip -3mm}\\
8. Random & 50 & 0.94 & 0.94 & 0.94 & 0.98 &  & 0.93 & 0.94 & 0.94 & -\\
 & 200 & 0.95 & 0.93 & 0.93 & 0.94 &  & 0.95 & 0.94 & 0.93 & -\\
\noalign{\vskip 1mm} \it{Outcome Variable}\\
9. Ln Hours & 50 & 0.93 & 0.94 & 0.93 & 0.95 &  & 0.77 & 0.94 & 0.92 & -\\
 & 200 & 0.94 & 0.95 & 0.96 & 0.96 &  & 0.93 & 0.95 & 0.95 & -\\
10. Ln U-rate & 50 & 0.93 & 0.91 & 0.92 & 0.92 &  & 0.79 & 0.90 & 0.91 & -\\
 & 200 & 0.95 & 0.94 & 0.95 & 0.95 &  & 0.92 & 0.94 & 0.94 & -\\
\bottomrule
\end{tabular}
\begin{tablenotes}[para]
\item \vspace{-5ex} \singlespacing \footnotesize{Notes: Simulation results for CPS data based on 1,000 Monte Carlo experiments.
                               The DGP is summarized in Section \ref{sec:Simulations}. The dataset has 7 times periods and either
                               $n=50$ or $n=200$. ``Bootstrap'' results use 300 nonparametric clustered bootstrap replications to compute bootstrapped clustered standard errors for all methods. 
                              ``Analytical'' results use asymptotic approximations for computing standard errors for all methods but synthetic DiD, as analytical standard errors are not available. 
                               All the other details are the same as in Table \ref{tab:sim_placebo_synth}.}
\end{tablenotes}
\end{threeparttable}}
\end{table}

Table \ref{tab:sim_placebo_synth_coverage} present the empirical coverage of $95\%$ confidence intervals for $ES_\text{avg}$ when $n=50$ and when $n=200$. We consider both nonparametric clustered-bootstrapped standard errors and analytical standard errors. Our simulation results show that using bootstrapped standard errors to construct confidence intervals yields very good coverage properties for all estimators across all considered DGPs, even in the challenging setup with $n=50$. On the other hand, our results also reveal that inference based on analytical standard errors with our efficient DiD estimator is anticonservative when $n=50$, though the results improve when $n=200$. Thus, in setups with very limited sample sizes, we recommend using cluster bootstrap standard errors.

Finally, Table \ref{tab:sim_placebo_synth_length} presents the length of $95\%$ confidence intervals for all the estimators, using bootstrap and analytical standard errors. Since it is only reasonable to compare the length of confidence intervals across methods that control size, we focus our attention on the bootstrap results. Here, we note that our efficient DiD estimator tends to have substantially shorter confidence intervals than other available DiD estimators. For instance, in the baseline DGP, the bootstrapped confidence interval for the synthetic DiD estimator is more than $40\%$ larger than our efficient DiD estimator when $n=50$, and approximately two times larger than our efficient DiD estimator when $n=200$; these gains are even larger when compared to TWFE and DTWFE estimators. When treatment assignment is based on gun laws across states, the length of bootstrapped confidence intervals of all other considered estimators is more than four times larger than those based on our efficient estimator. We reach a qualitatively similar conclusion when using abortion rights as treatment. On the other hand, when the outcome of interest is the log unemployment rate, we note that bootstrap inference using synthetic DiD or traditional TWFE estimators is more precise than our efficient DiD estimator. However, this difference shrinks when $n=200$. 

\begin{table}[!h]
\centering
\caption{\label{tab:sim_placebo_synth_length}Simulation results for CPS data: Relative length of confidence interval}
\centering
\resizebox{\ifdim\width>\linewidth\linewidth\else\width\fi}{!}{
\begin{threeparttable}
\begin{tabular}[t]{lcrrrrrrrrr}
\toprule
\multicolumn{2}{c}{ } & \multicolumn{4}{c}{Bootstrap} & \phantom{abc} & \multicolumn{4}{c}{Analytical} \\
\cmidrule(l{3pt}r{3pt}){3-6} \cmidrule(l{3pt}r{3pt}){8-11}
\multicolumn{1}{c}{} & \multicolumn{1}{c}{Sample size} & \multicolumn{1}{c}{EDiD} & \multicolumn{1}{c}{TWFE} & \multicolumn{1}{c}{DTWFE} & \multicolumn{1}{c}{SDiD} & \multicolumn{1}{c}{} & \multicolumn{1}{c}{EDiD} & \multicolumn{1}{c}{TWFE} & \multicolumn{1}{c}{DTWFE} & \multicolumn{1}{c}{SDiD}\\
\midrule
1. Baseline & 50 & 1.00 & 3.63 & 12.58 & 1.43 &  & 1.00 & 5.44 & 18.40 & -\\
 & 200 & 1.00 & 2.46 & 3.60 & 2.08 &  & 1.00 & 2.65 & 3.86 & -\\
\noalign{\vskip 1mm} \it{Outcome Model}\\
2. No corr & 50 & 1.00 & 3.67 & 12.78 & 1.36 &  & 1.00 & 5.46 & 18.57 & -\\
 & 200 & 1.00 & 2.47 & 3.58 & 2.11 &  & 1.00 & 2.67 & 3.85 & -\\
\noalign{\vskip -3mm}\\
3. No M & 50 & 1.00 & 3.56 & 12.90 & 1.32 &  & 1.00 & 5.32 & 18.83 & -\\
 & 200 & 1.00 & 2.33 & 3.39 & 1.84 &  & 1.00 & 2.52 & 3.64 & -\\
\noalign{\vskip -3mm}\\
4. No F & 50 & 1.00 & 1.42 & 1.85 & 1.44 &  & 1.00 & 2.14 & 2.72 & -\\
 & 200 & 1.00 & 1.58 & 2.18 & 1.57 &  & 1.00 & 1.71 & 2.35 & -\\
\noalign{\vskip -3mm}\\
5. Only noise & 50 & 1.00 & 1.15 & 1.55 & 1.23 &  & 1.00 & 1.74 & 2.28 & -\\
 & 200 & 1.00 & 1.31 & 1.74 & 1.34 &  & 1.00 & 1.41 & 1.87 & -\\
\noalign{\vskip -3mm}\\
\noalign{\vskip 1mm} \it{Treatment Assignment}\\
6. Gun law & 50 & 1.00 & 6.41 & 16.94 & 4.33 &  & 1.00 & 7.00 & 18.12 & -\\
 & 200 & 1.00 & 8.86 & 13.18 & 5.92 &  & 1.00 & 8.99 & 13.33 & -\\
7. Abortion & 50 & 1.00 & 6.40 & 16.44 & 3.84 &  & 1.00 & 7.14 & 18.00 & -\\
 & 200 & 1.00 & 8.53 & 13.30 & 5.44 &  & 1.00 & 8.76 & 13.53 & -\\
\noalign{\vskip -3mm}\\
8. Random & 50 & 1.00 & 3.93 & 14.71 & 2.13 &  & 1.00 & 4.31 & 15.71 & -\\
 & 200 & 1.00 & 2.36 & 3.49 & 2.03 &  & 1.00 & 2.40 & 3.53 & -\\
\noalign{\vskip 1mm} \it{Outcome Variable}\\
9. Ln Hours & 50 & 1.00 & 1.05 & 2.01 & 1.03 &  & 1.00 & 1.58 & 2.95 & -\\
 & 200 & 1.00 & 1.37 & 2.07 & 1.33 &  & 1.00 & 1.48 & 2.21 & -\\
10. Ln U-rate & 50 & 1.00 & 0.75 & 1.34 & 0.79 &  & 1.00 & 1.13 & 1.98 & -\\
 & 200 & 1.00 & 0.99 & 1.51 & 0.99 &  & 1.00 & 1.07 & 1.62 & -\\
\bottomrule
\end{tabular}
\begin{tablenotes}[para]
\item \vspace{-5ex} \singlespacing \footnotesize{Notes: Simulation results for CPS data based on 1,000 Monte Carlo experiments.
                               The data generating process is summarized in Section \ref{sec:Simulations}. The dataset has 7 times periods and either
                               $n=50$ or $n=200$. ``Bootstrap'' results use 300 nonparametric clustered bootstrap replications to compute bootstrapped clustered standard errors for all methods. 
                              ``Analytical'' results use asymptotic approximations for computing standard errors for all methods but synthetic DiD, as analytical standard errors are not available. Length of Confidence intervals
                              are measured relative to our efficient DiD estimator.
                               All the other details are the same as in Table \ref{tab:sim_placebo_synth}.}
\end{tablenotes}
\end{threeparttable}}
\end{table}

\begin{figure}[!ht]
\caption{Contribution of pre-treatment periods for the Efficient DiD estimator for $ES_{\text{avg}}$} 
\label{fig:heatmap_single}
\begin{subfigure}[t]{0.48\textwidth}
         \centering
          \includegraphics[width=0.9\textwidth, keepaspectratio]{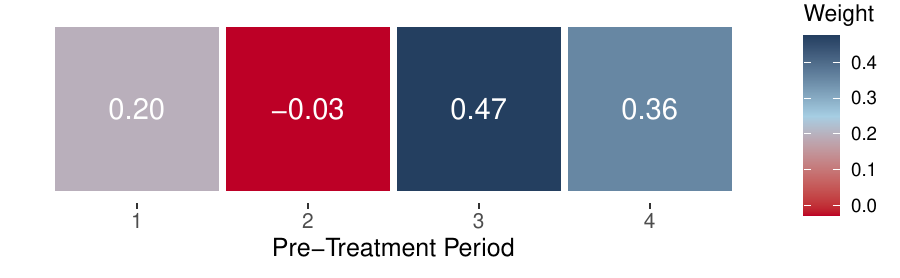}
         \caption{DGP 1: Baseline}
         \label{fig:heatmap1}
     \end{subfigure}
     \hfill
     \begin{subfigure}[t]{0.48\textwidth}
         \centering
         \includegraphics[width=0.9\textwidth, keepaspectratio]{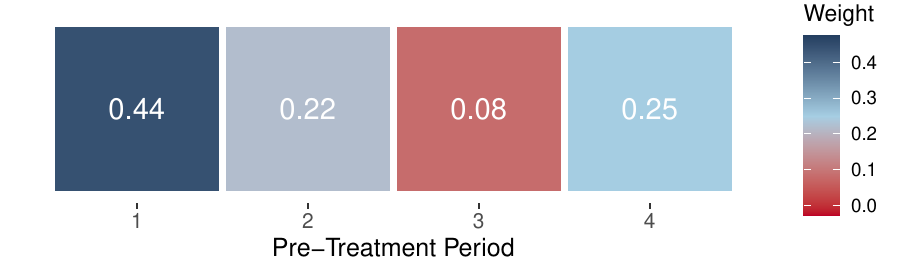}
         \caption{DGP 3: Outcome Model with no M}
         \label{fig:heatmap3}
     \end{subfigure}
\hfill
\begin{subfigure}[t]{0.48\textwidth}
         \centering
         \includegraphics[width=0.9\textwidth, keepaspectratio]{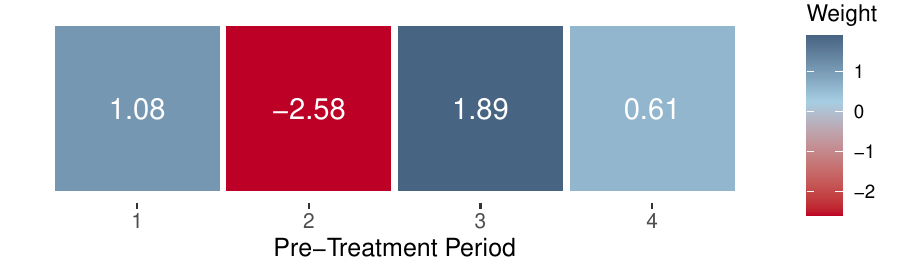}
         \caption{DGP 6: Gun law treatment assignment}
         \label{fig:heatmap6}
     \end{subfigure}
     \hfill
     \begin{subfigure}[t]{0.48\textwidth}
         \centering
         \includegraphics[width=0.9\textwidth, keepaspectratio]{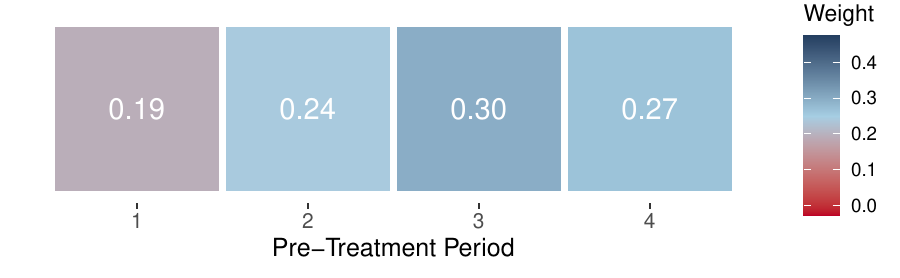}
         \caption{DGP 9: Ln Hours as outcome}
         \label{fig:heatmap9}
     \end{subfigure}

  \justifying
\noindent\scriptsize{\textit{Notes:} The figure displays the contribution of each pre-treatment period when constructing the efficient DiD estimator for $ES_{\text{avg}}$ based on a single draw of the data-generating process. Sample size $n=50$ and $T=7$, with four pre-treatment periods and three post-treatment periods. The weight color scale is the same across all sub-figures except Figure \ref{fig:heatmap6}. }
     
\end{figure} 

In practice, we anticipate that empirical researchers may also be interested in better understanding the contribution of each DiD component that uses a different pre-treatment period as a baseline period. In other words, one may be interested in understanding the ``contribution'' of each pre-treatment period to form efficient estimators. Figure \ref{fig:heatmap_single} plots this in a heatmap-style for four different DGPs, using a single simulation draw. As evident from this, the optimal way to aggregate pre-treatment periods to gain efficiency varies with the DGP, and the further it deviates from uniform weights, the higher the gains are compared to TWFE. This highlights that, in general, we can substantially improve upon TWFE. 

Another point worth stressing is related to negative efficiency weights. In our context, negative efficiency weights do not lead to concerns related to treatment effect sign preservation as discussed in \citet{Goodmanbacon2021}, \citet{Borusyak2023}, and \cite{deChaisemartin2020_AER}, for example. The reason for this is that our overidentification conditions imply that the $ATT(g,t)$'s are \emph{homogeneous} across baseline periods, allowing us to potentially non-convex sums across different baseline periods. As illustrated in Figure \ref{fig:heatmap6} and backed up by the results in Tables \ref{tab:sim_placebo_synth} and \ref{tab:sim_placebo_synth_length}, this should not be interpreted as a concern. Negative efficiency weights arise simply as a consequence of the dependence structure of the changes in outcome in treatment and comparison groups, and de facto leveraging them can lead to substantial gains in precision, as illustrated in DGP 6.

Altogether, these simulation results highlight the excellent finite sample properties of our efficient DiD estimator in realistic DGPs: they tend to have much smaller RMSE and shorter confidence intervals than other available estimators. Importantly, the gains in precision can be very substantial.

\subsection{Simulations based on Compustat with staggered treatment}
Our second set of simulations builds on \citet{Baker2022} and considers DGPs with staggered treatment adoption calibrated to Compustat panel data. We differ from \citet{Baker2022} in some aspects: (i) consistent with our DiD framework with ``large $n$'' and fixed $T$, we consider a setup with $n=400$ firms and follow them for $T=11$ years; (ii) we consider that error terms are serially correlated with varying autocorrelation parameter $\rho$ to allow for richer outcome dynamics and assess its impact of the performance of different estimators; (iii) we follows a one-step treatment assignment in which firms are assigned to different cohorts. Apart from these differences, the construction of our data-generating process follows that of \citet{Baker2022}, which allows for three treatment dates, with all units eventually being treated, and treatment effects evolving dynamically with heterogeneous trends. We describe the DGPs below for completeness. 

As in \citet{Baker2022}, we begin with a sample of all firms in Compustat over the 36-year period from 1980 to 2015 that are U.S. incorporated, non-financial, and contain at least five observations. This leaves us with a sample of 12,020 different firms. Using this unbalanced panel, we compute returns on assets (ROA) and then decompose it into year fixed effects, firm/unit fixed effects, and residuals:
\begin{align*}
    \{\widehat{\eta}_i\}_{i=1}^{12,020},~ \{\widehat{\alpha}_t\}_{t=1}^{36},~\{\widehat{\varepsilon}_{i,t}\}_{i,t=1}^{176,622},~\text{from } ROA_{i,t} = \eta_i + \alpha_t + \varepsilon_{i,t}.
\end{align*}
The empirical (marginal) distribution of these terms serves as the benchmark distribution in our Monte Carlo simulations. We consider three possible treatment dates, $G=5$, $G=8$, and $G=11$, with each firm randomly assigned to each treatment group with probability $1/3$.

Based on these features, we simulate potential outcomes $Y_{i,t}(5)$, $Y_{i,t}(8)$, $Y_{i,t}(11)$, and $Y_{i,t}(\infty)$ as 
\begin{align}
    Y_{i,t}(\infty) &= \widetilde{\alpha}_t + \widetilde{\eta}_i + \widetilde{\varepsilon}_{i,t}\\
    Y_{i,t}(5) &=   Y_{i,t}(\infty) + 0.5\times \sigma_{ROA} 1\{t\ge 5\} \left( t - 5 + 1 \right)\\
    Y_{i,t}(8) &=   Y_{i,t}(\infty) +  0.3\times \sigma_{ROA} 1\{t\ge 8\} \left( t - 8 + 1 \right)\\
    Y_{i,t}(11) &=   Y_{i,t}(\infty) + 0.1\times \sigma_{ROA} 1\{t\ge 11\} \left( t - 11 + 1 \right),
\end{align}
where $\widetilde{\eta}_i$ and $\widetilde{\alpha}_t$ are drawn from the empirical distribution of $\{\widehat{\eta}_i\}_{i=1}^{12,020}$ and $\{\widehat{\alpha}_t\}_{t=1}^{36}$, respectively,   $\widetilde{\varepsilon}_{i,t} = \rho \widetilde{\varepsilon}_{i,t-1} + \widetilde{u}_{i,t}$, with $\widetilde{u}_{i,t}$ being drawn from the empirical distribution of $\{\widehat{\varepsilon}_{i,t}\}_{it=1}^{176,622}$, and $\sigma_{ROA}$ refers to the sample standard deviation of ROA (0.309). To assess the impact of serial correlation on our results, we consider $\rho \in \{-1.1, -1, -0.5, 0, 0.5, 1, 1.1\}$.

The observed outcomes are given by $Y_{i,t} = \sum_{g \in \{5,8,11\}}Y_{i,t}(g)\times1\{G_i=g\}$. As we consider a setup where all units are eventually treated, we are unable to identify average treatment effects in period $ t=11$. Therefore, we effectively drop the data from the last period ($t=11$), and cohort $G=11$ acts as "never-treated" units. In our setup, we have that  $ATT(5,t) = 0.154 (t-4)$ for $t=5\dots, 10$, and $ATT(8,t) = 0.093 (t-7)$ for $t=8,9,10$. As a consequence, we have that for $e=0,1,\dots, 5$, $ES(e)$ is equal to $0.123, 0.247, 0.370, 0.617, 0.772$, and $0.926$, respectively, and $ES_{\text{avg}} = 0.509$.

\begin{table}[!h]
\centering
\caption{\label{tab:sim_Baker_etal}Simulation results for Compustat data: Relative RMSE and Bias}
\centering
\resizebox{\ifdim\width>\linewidth\linewidth\else\width\fi}{!}{
\begin{threeparttable}
\begin{tabular}[t]{lrrrrrrrrr}
\toprule
\phantom{abc} & \multicolumn{4}{c}{Relative RMSE} & \phantom{abc} & \multicolumn{4}{c}{Bias ($\times 10$)} \\
\cmidrule(l{3pt}r{3pt}){2-5} \cmidrule(l{3pt}r{3pt}){7-10}
\multicolumn{1}{c}{} & \multicolumn{1}{c}{EDiD} & \multicolumn{1}{c}{BJS-G-W} & \multicolumn{1}{c}{CS-SA} & \multicolumn{1}{c}{CS-dCDH} & \phantom{abc} & \multicolumn{1}{c}{EDiD } & \multicolumn{1}{c}{BJS-G-W } & \multicolumn{1}{c}{CS-SA } & \multicolumn{1}{c}{CS-dCDH }\\
\midrule
$\rho = 0$ & 1 & 1.01 & 1.61 & 1.55 &  & 0.00 & 0.00 & 0.00 & 0.00\\
$\rho = 0.5$ & 1 & 1.04 & 1.26 & 1.22 &  & -0.01 & -0.01 & 0.00 & 0.00\\
$\rho = 1$ & 1 & 1.18 & 1.08 & 1.04 &  & 0.00 & 0.01 & 0.01 & 0.01\\
$\rho = 1.1$ & 1 & 1.47 & 1.23 & 1.18 &  & -0.03 & -0.02 & -0.02 & -0.02\\
$\rho = -0.5$ & 1 & 1.04 & 2.31 & 2.20 &  & 0.00 & 0.00 & 0.00 & 0.00\\
$\rho = -1$ & 1 & 1.58 & 3.22 & 3.37 &  & 0.00 & 0.00 & 0.00 & 0.00\\
$\rho = -1.1$ & 1 & 2.23 & 3.37 & 3.77 &  & 0.00 & 0.00 & 0.00 & -0.01\\
\bottomrule
\end{tabular}
\begin{tablenotes}[para]
\item \vspace{-5ex} \singlespacing \footnotesize{Notes: Simulation results for Compustat data based on 1,000 Monte Carlo experiments.
                               The DGP is summarized in Section \ref{sec:Simulations}. In all setups, we randomly draw 400 firms and follow them for 11 periods.
                               The baseline case has error terms that are randomly drawn from the residual distribution as described in Section  \ref{sec:Simulations}. 
                               In subsequent settings, we consider error terms with different degrees of serial correlation. The parameter of interest is $ES_\text{avg}$. Relative RMSE is reported using our efficient DiD estimator as a benchmark.
                             All results for bias are multiplied by 10 for readability.}
\end{tablenotes}
\end{threeparttable}}
\end{table}

Table \ref{tab:sim_Baker_etal} presents the bias and relative RMSE of the four different estimators for $ES_\text{avg}$. As expected, all estimators are unbiased. However, in terms of RMSE, there is a good amount of variation across estimators depending on the residual serial correlation. In all scenarios, our proposed efficient DiD estimator yields the lowest RMSE, which supports our theoretical results. When the errors are spherical ($\rho = 0)$, our efficient estimator performs very similarly to BJS-G-W. This is aligned with the efficiency results derived under these additional assumptions on the error term in \citet{Borusyak2023} and \citet{Wooldridge2021a}. Outside this specific setup, though, the BJS-G-W estimator has no efficiency guarantees. Outside of our efficient estimators, our simulations also suggest that it is generally not possible to rank other estimators in terms of efficiency, as their performance depends on the specific DGP. 

\begin{table}[!h]
\centering
\caption{\label{tab:sim_Baker_etal_inference}Simulation results for Compustat data: Coverage and relative length of confidence interval}
\centering
\resizebox{\ifdim\width>\linewidth\linewidth\else\width\fi}{!}{
\begin{threeparttable}
\begin{tabular}[t]{lrrrrrrrrr}
\toprule
\phantom{abc} & \multicolumn{4}{c}{Coverage} & \phantom{abc} & \multicolumn{4}{c}{Rel Length of CI} \\
\cmidrule(l{3pt}r{3pt}){2-5} \cmidrule(l{3pt}r{3pt}){7-10}
\multicolumn{1}{c}{} & \multicolumn{1}{c}{EDiD} & \multicolumn{1}{c}{BJS-G-W} & \multicolumn{1}{c}{CS-SA} & \multicolumn{1}{c}{CS-dCDH} & \phantom{abc} & \multicolumn{1}{c}{EDiD } & \multicolumn{1}{c}{BJS-G-W } & \multicolumn{1}{c}{CS-SA } & \multicolumn{1}{c}{CS-dCDH }\\
\midrule
$\rho = 0$ & 0.93 & 0.93 & 0.94 & 0.94 &  & 1 & 1.01 & 1.62 & 1.55\\
$\rho = 0.5$ & 0.95 & 0.95 & 0.95 & 0.95 &  & 1 & 1.05 & 1.27 & 1.23\\
$\rho = 1$ & 0.94 & 0.96 & 0.95 & 0.95 &  & 1 & 1.24 & 1.11 & 1.07\\
$\rho = 1.1$ & 0.94 & 0.96 & 0.95 & 0.95 &  & 1 & 1.54 & 1.27 & 1.22\\
$\rho = -0.5$ & 0.93 & 0.93 & 0.94 & 0.95 &  & 1 & 1.04 & 2.35 & 2.25\\
$\rho = -1$ & 0.93 & 0.94 & 0.95 & 0.95 &  & 1 & 1.64 & 3.33 & 3.48\\
$\rho = -1.1$ & 0.95 & 0.95 & 0.94 & 0.94 &  & 1 & 2.22 & 3.38 & 3.78\\
\bottomrule
\end{tabular}
\begin{tablenotes}[para]
\item \vspace{-5ex} \singlespacing \footnotesize{Notes: Simulation results for Compustat data based on 1,000 Monte Carlo experiments.
                               The DGP is summarized in Section \ref{sec:Simulations}. In all setups, we randomly draw 400 firms and follow them for 11 periods.
                               The baseline case has error terms that are randomly drawn from the residual distribution as described in Section  \ref{sec:Simulations}. 
                               In subsequent settings, we consider error terms with different degrees of serial correlation.
                              The parameter of interest is $ES_\text{avg}$. Coverage is computed using empirical coverage associated with 95\% confidence intervals based on asymptotic approximations (analytical standard errors).
                               Relative length of confidence intervals are reported using our efficient DiD estimator as benchmark.}
\end{tablenotes}
\end{threeparttable}}
\end{table}

Table \ref{tab:sim_Baker_etal_inference} presents the empirical coverage of $95\%$ confidence intervals for $ES_\text{avg}$. All estimators display correct coverage across all considered DGPs. In terms of inferential precision measured by the length of the $95\%$ confidence intervals, the results once again highlight the practical gains of using our efficient DiD estimator. The relative size of these gains varies with the degree of serial correlation in the data, as expected and consistent with our theoretical results.

\begin{figure}[!ht]
\caption{Contribution of treatment and comparison groups, and pre-treatment periods for the Efficient DiD estimator for $ES_{\text{avg}}$} 
\label{fig:heatmap_stg}
\begin{subfigure}[t]{.9\textwidth}
         \centering
          \includegraphics[width=\linewidth, keepaspectratio]{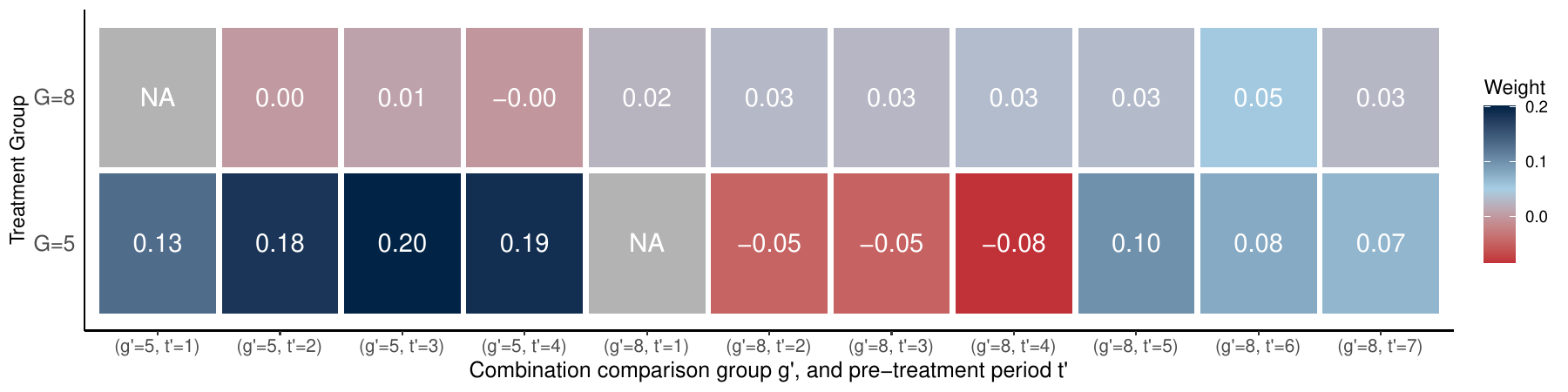}
         \caption{DGP with $\rho = 0$}
         \label{fig:heatmap_stg_avg0}
     \end{subfigure}
     
     \begin{subfigure}[t]{.9\textwidth}
         \centering
         \includegraphics[width=\linewidth, keepaspectratio]{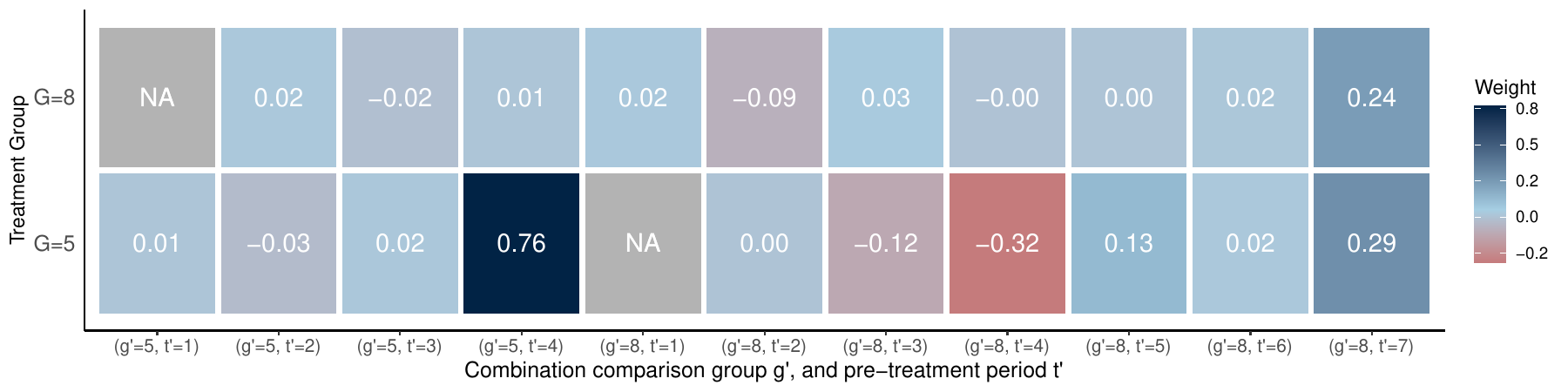}
         \caption{DGP with $\rho = 1$}
         \label{fig:heatmap_stg_avg1}
     \end{subfigure}

\begin{subfigure}[t]{.9\textwidth}
         \centering
         \includegraphics[width=\linewidth, keepaspectratio]{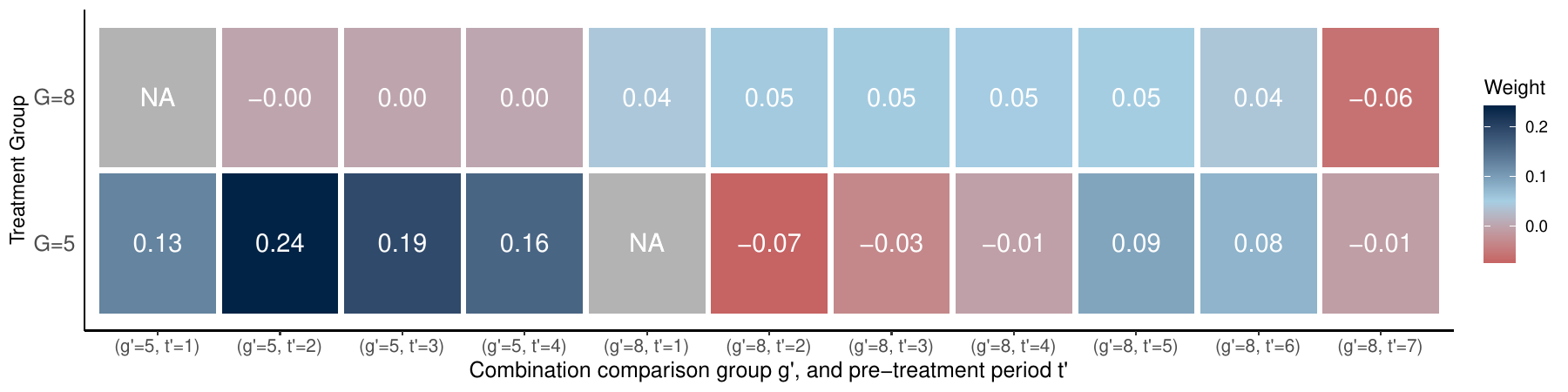}
         \caption{DGP with $\rho = -1$}
         \label{fig:heatmap_stg_avg-1}
     \end{subfigure}

\justifying
\noindent\scriptsize{\textit{Notes:} The figure displays the contribution of each pre-treatment period, comparison group, and treatment group when constructing the efficient DiD estimator for $ES_{\text{avg}}$ based on a single draw of the data-generating process. Sample size $n=400$ and $T=11$, and staggered treatment adoption. The weight color scale varies across DGPs. }
     
\end{figure}

As we discussed before, an appealing feature of our proposal is that we can visualize the contribution of the treatment group, the effective comparison group, and the pre-treatment period used as baseline for the construction of our efficient DiD estimator for $ES_{\text{avg}}$. Figure \ref{fig:heatmap_stg} plots this in a heatmap-style for three different DGPs, using a single simulation draw. Overall, we can see that the treatment group $G=5$ ``contributes more'' to the $\widehat{ES}_{\text{avg}}$, which is again expected, as this group has six available post-treatment periods; group $G=8$ has only three. Across DGPs, we can also see that the efficient weights vary across comparison groups and pre-treatment periods. When there is no serial correlation, most of the weights come from when $G=g'=5$, indicating that the efficient estimator is leveraging the ``never-treated'' units as a leading comparison group. In this case, we can also see that all pre-treatment periods have similar weighting schemes, again highlighting the appeal of exploring various pre-treatment periods as a baseline. When the error term follows a unit root process, we observe that the efficient estimator allocates almost all the weights to the last available pre-treatment periods, and earlier pre-treatment periods are less important for efficiency considerations. When the errors are negatively serially correlated with $\rho = -1$, we see that the weights are differently spread than in the previous case. Overall, these heatmap plots illustrate that the optimal way to aggregate information across treatment groups, comparison groups, and pre-treatment periods varies from DGP to DGP, and being able to (asymptotically) achieve the semiparametric efficiency bound without needing to know these additional features of the data can be very attractive. Similarly to the setup with a single treatment date, negative efficiency weights provide no reason to be concerned about bias in our context. These efficiency weights appear as a consequence of the semiparametric efficiency results, and leveraging their structure can lead to practical gains in precision.

\section{Empirical Illustration}\label{sec:application}

We illustrate how our efficient DiD and ES estimators can be used by estimating the effect of hospitalization on out-of-pocket medical spending using publicly available survey data from the Health and Retirement Study (HRS) from the replication package of \citet{Dobkin_etal_2018_AER}. Like  \citet{Dobkin_etal_2018_AER}, we explore the variation in the timing of hospitalization observed in the HRS to estimate the different treatment effects of interest. \citet{Dobkin_etal_2018_AER} analyzes many other outcomes of interest and also leverage (not publicly available) hospitalization data linked to credit reports.

To conduct our illustrative analysis, we follow the same sample selection steps used by \citet{Sun2021}, who also revisited \citet{Dobkin_etal_2018_AER}. The sample selection closely follows \citet{Dobkin_etal_2018_AER}, and restricts attention to non-pregnancy-related hospital admissions and to adults who are hospitalized at ages 50-59. Like \citet{Sun2021}, we restrict our analysis to a subsample of individuals who appear throughout waves 7-11 (approximately spanning 2004-2012), so we maintain a balanced panel with 652 individuals in 5 waves, $t=7, \dots, 11$. 
Given the data construction, all individuals in the sample are hospitalized by $t=11$, but were not hospitalized in wave $t=7$. Thus, as we do not have a valid comparison group for wave $t=11$, we also drop data from that period. 

Overall, we have 4 treatment groups: individuals (first) hospitalized in $t=8$ ($G_i = 8$, 252 individuals), $t=9$, ($G_i = 9$, 176 individuals),  $t=10$ ($G_i = 10$, 163 individuals), and those hospitalized in $t=11$ ($G_i=\infty$, 65 individuals). Note that as we are only using data up to wave 10, we can re-label the units treated in wave $t=11$ as the ``never-treated'' group, so we match our notation. 

As discussed in detail in \citet{Sun2021}, (unconditional) parallel trends (in all periods and groups) and no-anticipation assumptions are plausible in this setting, and one should a priori expect that the effect of hospitalization is heterogeneous across individuals and periods. Figure \ref{fig:trends_application} plots the evolution of the average out-of-pocket medical spending across all treatment groups, and also provides suggestive evidence in favor of the plausibility of parallel trends assumptions across all periods and groups. Thus, the ``heterogeneous robust'' DiD and ES estimators for staggered treatment adoption proposed by \citet{deChaisemartin2020_AER}, \citet{Callaway_Santanna_2021}, \citet{Sun2021}, \citet{Gardner2021}, \citet{Wooldridge2021a}, and \citet{Borusyak2023} are well-motivated. When parallel trends is plausible across all periods and groups, we should expect all modern DiD estimators to provide similar point estimates.

\begin{figure}[ht]
\caption{Evolution of average out-of-pocket medical spending across different treatment groups} 
\label{fig:trends_application}
         \centering
          \includegraphics[width=.75\textwidth, keepaspectratio]{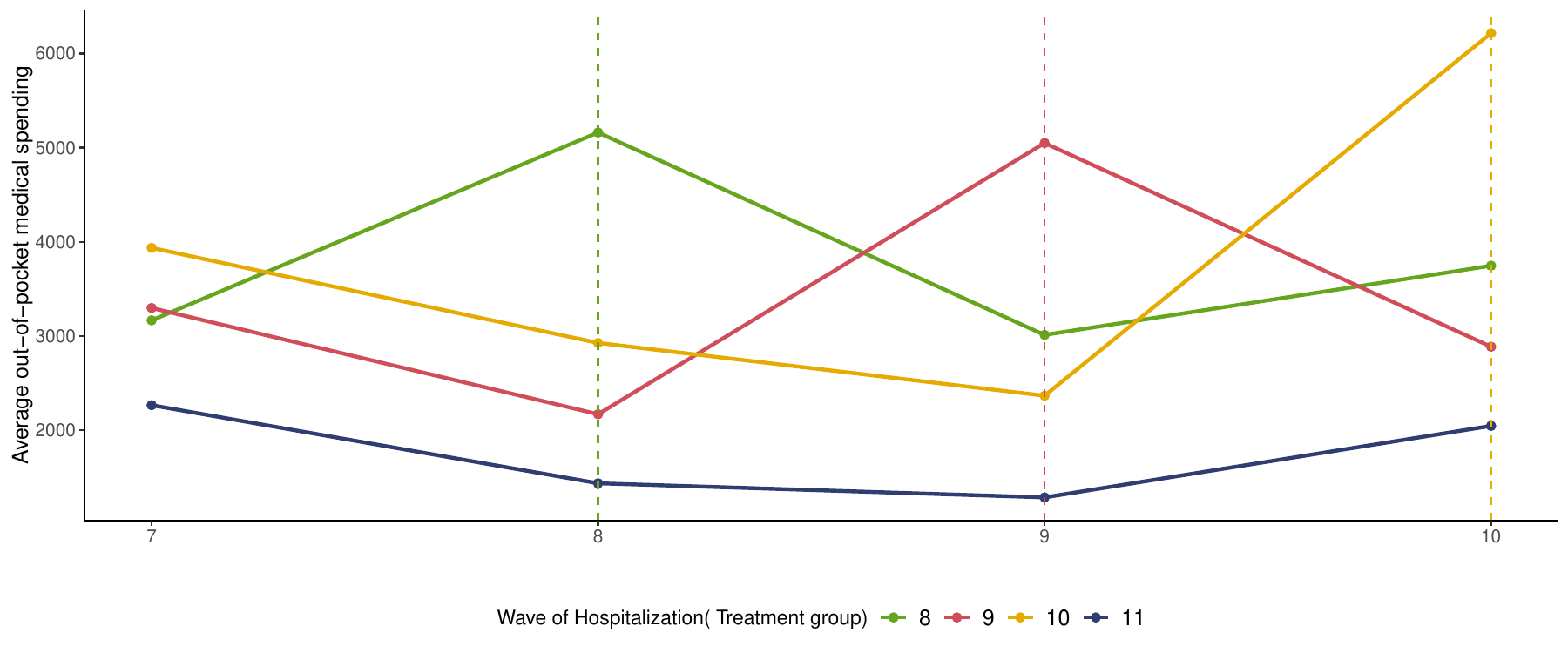}

\justifying
\noindent\scriptsize{\textit{Notes:} Plot of evolutions of average out-of-pocket medical spending across four cohorts defined by the time they were first hospitalized. The waves of the first hospitalizations are displayed in dashed vertical lines. For each cohort, points before the vertical lines are pre-treatment periods, and after the dashed vertical lines are post-treatment periods. We use the same balanced panel data as \citet{Sun2021}.  }
     
\end{figure}

In this context, we are interested in different treatment effect parameters. As we discussed before, a natural starting point is to analyze the six different post-treatment $ATT(g,t)$'s that measure the evolution of ATTs since hospitalization for each group. For instance, $ATT(8,8)$ would give the average treatment effect of hospitalization on out-of-pocket medical spending in time $t=8$, among individuals in group $g=8$, i.e., those who were hospitalized in time $8$; $ATT(8,9)$ and $ATT(8,10)$ would provide the average treatment effect for these same individuals but in periods $t=9$ (one year after hospitalization) and $t=10$ (two years after hospitalization). The interpretation of the other $ATT(g,t)$'s is analogous. We may also want to summarize all these different $ATT(g,t)$'s into more aggregated summary measures to ease interpretations and potentially gain precision. A natural way to summarize the effects is to look at the event-study parameters $ES(e)$ as in \eqref{eqn:ES}---which provides an average treatment effect by length of treatment exposure $e$---and their averages, $ES_{\text{avg}}$, as in \eqref{eqn:overall_ATT}.

\begin{table}[!h]
\centering
\caption{\label{tab:application_sa2021}Estimates for the effect of hospitalization on out-of-pocket medical spending}
\centering
\resizebox{\ifdim\width>\linewidth\linewidth\else\width\fi}{!}{
\begin{threeparttable}
\begin{tabular}[t]{lcccccccccc}
\toprule
\multicolumn{1}{c}{Estimator} & \multicolumn{1}{c}{$ATT(8,8)$} & \multicolumn{1}{c}{$ATT(8,9)$} & \multicolumn{1}{c}{$ATT(8,10)$} & \multicolumn{1}{c}{$ATT(9,9)$} & \multicolumn{1}{c}{$ATT(9,10)$} & \multicolumn{1}{c}{$ATT(10,10)$} & \multicolumn{1}{c}{$ES(0)$} & \multicolumn{1}{c}{$ES(1)$} & \multicolumn{1}{c}{$ES(2)$} & \multicolumn{1}{c}{$ES_\text{avg}$}\\
\midrule
EDiD & 3072 & 1112 & 1038 & 3063 & 90 & 2908 & 3024 & 692 & 1038 & 1585\\
 & (806) & (637) & (817) & (690) & (641) & (894) & (486) & (471) & (816) & (521)\\
\noalign{\vskip 2.5mm}
CS-SA & 2826 & 825 & 800 & 3031 & 107 & 3092 & 2960 & 530 & 800 & 1430\\
 & (1035) & (909) & (1008) & (702) & (651) & (995) & (539) & (585) & (1008) & (647)\\
\noalign{\vskip 2.5mm}
CS-dCDH & 3029 & 1248 & 800 & 3324 & 107 & 3092 & 3134 & 779 & 800 & 1571\\
 & (913) & (861) & (1008) & (959) & (651) & (995) & (536) & (570) & (1008) & (566)\\
\noalign{\vskip 2.5mm}
BJS-G-W & 3029 & 1285 & 1021 & 3239 & 77 & 2758 & 3017 & 788 & 1021 & 1609\\
 & (916) & (767) & (851) & (862) & (729) & (957) & (555) & (587) & (851) & (582)\\
\bottomrule
\end{tabular}
\begin{tablenotes}[para]
\item \vspace{-5ex} \singlespacing \footnotesize{Notes:  This table reports four different sets 
                             of estimates for the dynamic effects of hospitalization on out-of-pocket medical spending. Each column is a different
                             target parameter related to $ATT(g,t)$, $ES(e)$ or $ES_{avg}$. Standard errors clustered at the individual level are reported in parentheses.
                             The sample includes observations from wave $t=7$ to $t=10$, for 652 individuals. We leverage data from \citet{Dobkin_etal_2018_AER} and 
                             follow the same sample construction steps as \citet{Sun2021}.
}
\end{tablenotes}
\end{threeparttable}}
\end{table}

Table \ref{tab:application_sa2021} reports the estimates of all these target parameters and displays their standard errors clustered at individual level in parenthesis. We report estimates using different methods: our efficient DiD procedure (EDiD), the DiD estimator of \citet{Callaway_Santanna_2021} and \citet{Sun2021} that uses the never-treated units as the comparison group (CS-SA), the one of \citet{Callaway_Santanna_2021} and \citet{deChaisemartin2020_AER} that uses the not-yet-treated units as the comparison group (CS-dCDH), and the imputation DiD estimator of \citet{Borusyak2023}, \citet{Gardner2021}, and \citet{Wooldridge2021a} (BJS-G-W). Several overall patterns arise. First, note that the point estimates for each target parameter are similar across all the different DiD estimators. This is expected, as parallel trends across all periods and groups and no-anticipation are plausible in this application. We also note that instantaneous treatment effects tend to be larger than later effects, across all cohorts; for instance, estimates of $ATT(8,8)$ tend to be larger than $ATT(8,9)$ and $ATT(8,10)$, and $ATT(9,9)$ larger than $ATT(9,10)$. This is well captured in the event-study estimates, $ES(0)$, $ES(1)$, and $ES(2)$, indicating that the effect of hospitalization on out-of-pocket medical spending is mostly transitory. This is in line with the findings in \citet{Dobkin_etal_2018_AER}.

Another overall pattern that emerges from the results in Table \ref{tab:application_sa2021} relates to the inference precision of these estimators, as measured by their standard deviations. In line with our theoretical results, our efficient DiD estimators have the smallest standard errors across all target parameters. Outside our efficient DiD procedure, though, no other estimator uniformly dominates the remaining ones in terms of length of confidence intervals. In fact, CS-SA, CS-dCDH, and BJS-G-W take turns into leading to the second most precise estimator across target parameters, with CS-SA being ``second best'' for $ATT(9,9)$ and $ATT(9,10)$, CS-dCDH coming second for $ATT(8,8)$, $ATT(9,10)$, $ES(0)$, $ES(1)$, and $ES_{\text{avg}}$, and BJS-G-W for $ATT(8,9)$, $ATT(8,10)$, $ATT(10,10)$, and $ES(2)$. This highlights that relying on additional assumptions, such as homoskedasticity and restrictions on serial correlation, to justify efficiency gains may be practically hard. Our proposed efficient DiD estimator does not rely on these conditions and, in fact, delivers more precise inference procedures.

\begin{table}[!h]
\centering
\caption{\label{tab:application_sa2021_ARE}Relative efficiency of estimators for the effect of hospitalization on medical spending}
\centering
\resizebox{\ifdim\width>\linewidth\linewidth\else\width\fi}{!}{
\begin{threeparttable}
\begin{tabular}[t]{lcccccccccc}
\toprule
\multicolumn{1}{c}{Estimator} & \multicolumn{1}{c}{$ATT(8,8)$} & \multicolumn{1}{c}{$ATT(8,9)$} & \multicolumn{1}{c}{$ATT(8,10)$} & \multicolumn{1}{c}{$ATT(9,9)$} & \multicolumn{1}{c}{$ATT(9,10)$} & \multicolumn{1}{c}{$ATT(10,10)$} & \multicolumn{1}{c}{$ES(0)$} & \multicolumn{1}{c}{$ES(1)$} & \multicolumn{1}{c}{$ES(2)$} & \multicolumn{1}{c}{$ES_\text{avg}$}\\
\midrule
EDiD & 1.00 & 1.00 & 1.00 & 1.00 & 1.00 & 1.00 & 1.00 & 1.00 & 1.00 & 1.00\\
\noalign{\vskip -1.5mm}\\
CS-SA & 1.65 & 2.04 & 1.52 & 1.04 & 1.03 & 1.24 & 1.23 & 1.54 & 1.52 & 1.54\\
\noalign{\vskip -1.5mm}\\
CS-dCDH & 1.28 & 1.83 & 1.52 & 1.93 & 1.03 & 1.24 & 1.21 & 1.46 & 1.52 & 1.18\\
\noalign{\vskip -1.5mm}\\
BJS-G-W & 1.29 & 1.45 & 1.09 & 1.56 & 1.29 & 1.15 & 1.30 & 1.55 & 1.09 & 1.25\\
\bottomrule
\end{tabular}
\begin{tablenotes}[para]
\item \vspace{-5ex} \singlespacing \footnotesize{Notes:  This table reports the
                             estimated asymptotic relative efficiency (ARE) of four different sets 
                             of estimates for the dynamic effects of hospitalization on out-of-pocket medical spending, with
                             our efficient DiD estimator as the benchmark.
                             Each column provides the ARE for a different target parameter related to $ATT(g,t)$, $ES(e)$ or $ES_{avg}$.
                             The ARE is computed as the square of the ratio of standard errors of the DiD estimator relatitve to
                             the standard error of our efficient DiD estimator from Table \ref{tab:application_sa2021}.
                            
}
\end{tablenotes}
\end{threeparttable}}
\end{table}

But one may wonder: do these efficiency gains matter in practice? To answer this practically relevant question, we display in Table \ref{tab:application_sa2021_ARE} estimates of the asymptotic relative efficiency (ARE) of our proposed efficient DiD estimator with respect to the other available DiD estimators.\footnote{For any parameter $\eta$ of a distribution $F$, and for estimators $\widehat{\eta}_{1}$ and $\widehat{\eta }_{2}$ approximately $N\left(  \eta,V_{1}/n\right)  $ and $N\left(  \eta, V_{2}/n\right)  $, respectively, the asymptotic relative efficiency of $\widehat{\eta}_{2}$ with respect to $\widehat{\eta}_{1}$ is given by $V_{1}/V_{2}$; see, e.g., Section 8.2 in \cite{VanderVaart1998}.} Heuristically speaking, the ARE provides a relative measure of sample size needed for other DiD estimators to achieve the same precision as our efficient DiD estimator. The results in Table \ref{tab:application_sa2021_ARE} highlight that the gains in ARE of our efficient DiD estimator tend to be large. For example, for $ATT(8,8)$, CS-SA, CS-dCDH, and BJS-G-W DiD estimators would respectively need a sample size with $65\%$, $28\%$, and $29\%$ more hospitalized individuals than our EDiD estimator to achieve the same precision as EDiD. For $ATT(8,9)$, the gains are even larger, with CS-SA, CS-dCDH, and BJS-G-W DiD requiring $104\%$, $83\%$, and $45\%$ more hospitalized individuals than our EDiD estimator to equalize precision. These ARE gains remain large for the event-study coefficients, and their average. 

\begin{figure}[!ht]
\caption{Contribution of comparison groups and pre-treatment periods for the Efficient DiD estimator for $ATT(g,t)$'s} 
\label{fig:heatmap_stg_application}
         \centering
          \includegraphics[width=1\textwidth, keepaspectratio]{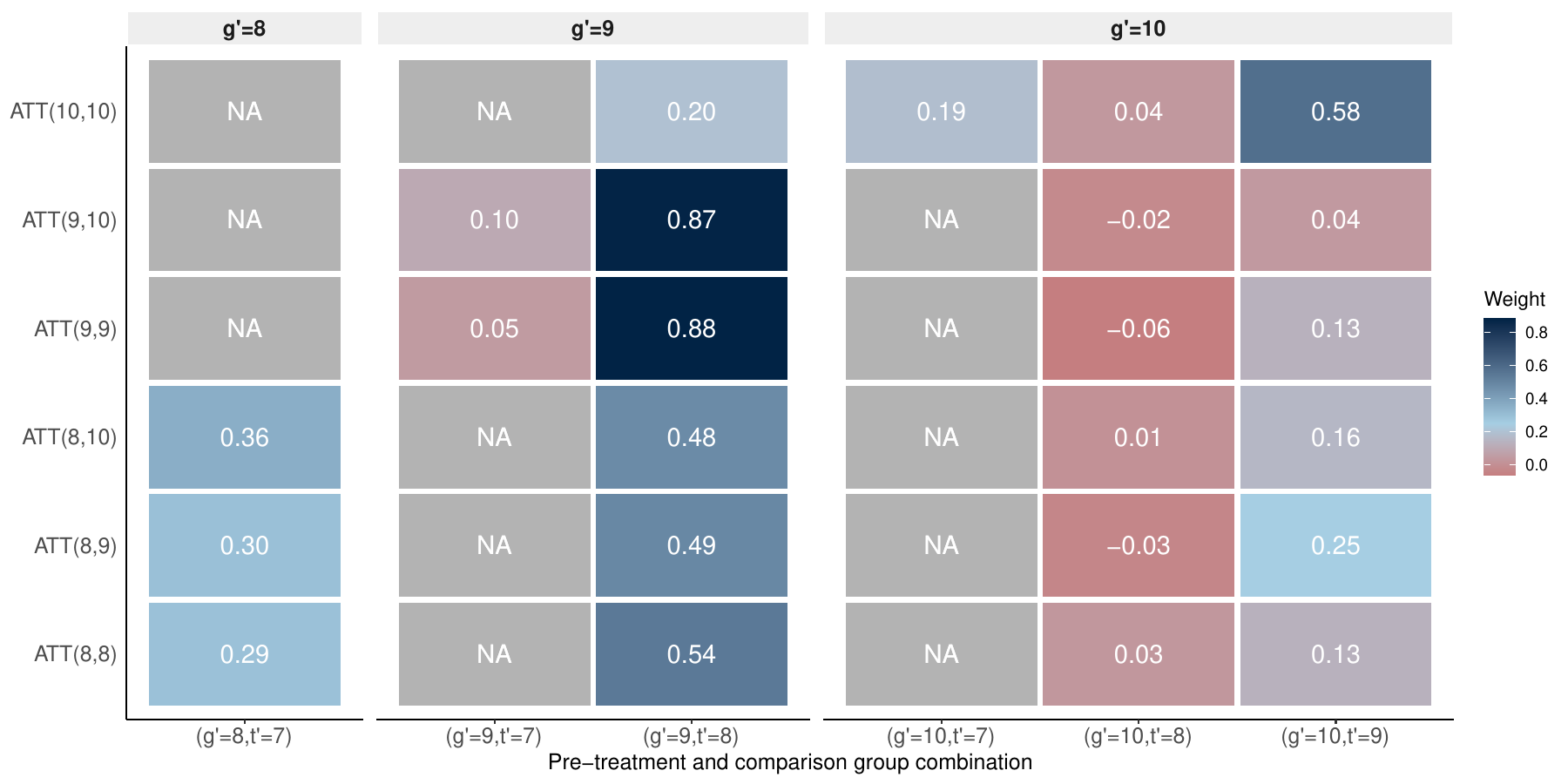}

\justifying
\noindent\scriptsize{\textit{Notes:} Weights attached to different $ATT(g,t)$ estimators that leverage different comparison groups and baseline periods in our application. The weights sum up to one in each row and arise as a consequence of our semiparametric efficient results in Theorem \ref{thm:efficiency}.  }
     
\end{figure} 

To gain further insights into how our efficient DiD and ES estimators achieve efficiency gains, in Figure \ref{fig:heatmap_stg_application} we report the weights attached to the different (effective) comparison groups and pre-treatment periods for each $ATT(g,t)$ of interest. Overall, there are four different (non-collinear) $ \widehat{\widetilde{Y}}^{\text{att(g,t)}}_{g',\tprime}$ as defined in \eqref{eqn:generic_att_gt_stg_estimated} that one can use to estimate each $ATT(g,t)$, and the weights attached to them sum up to one, i.e., the weights in Figure \ref{fig:heatmap_stg_application} sum up to one in each row. If we were to weigh each of these uniformly, each of them would have a 0.25 weight. In practice, we can see that our efficient DiD estimators never use these uniform weights, and, in fact, these weights vary according to the parameter of interest. This, by itself, highlights the appeal of our procedure to ``adapt'' to these different scenarios with the researcher needing to take a stand on additional hard-to-justify assumptions. 

It is also interesting to analyze these weights separately. We start with $ATT(8,8)$. For this parameter, the efficient estimator puts weights of 0.29 on the DiD estimator that uses only the never-treated (in our case, cohort treated in wave 11) as a comparison group ($G=g'=5)$,  0.54 on the DiD estimator that only uses cohort $G=9$ as a comparison group ($g'=9,t'=8)$, and 0.03 on the DiD estimator that uses cohort $G=10$ as the comparison group. Both of these use period $t=7$ as the baseline. Interestingly, the efficient DiD estimator for $ATT(8,8)$ puts 0.13 weight on the DiD estimator that uses ``never-treated'' and $G=10$ cohort as comparison group and leverages the first period as pre-treatment periods as well as $t_{\text{pre}}=9$, illustrating the ``bridging'' phenomenon discussed in the stylized example in Figure \ref{fig:illustration}. Here, note that $t_{\text{pre}}=9$ is a post-period for $G=5$ but a pre-treatment period for group $G=10$. The flexibility of our efficient DiD estimator to leverage this additional information from the data leads to the efficiency gains highlighted in Table \ref{tab:application_sa2021_ARE}.

For cohort $G=8$ in period $t=9$, we note that the efficient DiD estimator weights the DiD estimator based on never-treated units (i.e., the CS-SA estimator) by 0.30, the DiD estimator that leverages never-treated and $G=9$ as comparison groups by 0.49 (0.48), and the DiD estimator that leverages only $G=10$ as a comparison group by 0.25. An interpretation for the weights for $ATT(8,10)$ is analogous. For $ATT(9,9)$ and $ATT(9,10)$, note that the efficient DiD estimator puts nearly 0.90 weight on the CS-SA DiD estimator, explaining why its performance is similar to theirs (but better than the other alternative DiD estimators). For $ATT(10,10)$, our efficient estimator put 0.58 weight on the CS-SA estimator, 0.04, and 0.19 weight on the DiD estimator that uses never-treated as comparison and period eight and seven as baseline period, respectively, and the remaining 0.20 weight on the DiD estimator that combines never-treated and the cohort treated in wave $9$ as comparison groups and use first two periods as baseline periods. Together, these flexible weighting schemes demonstrate that our efficient DiD estimators are flexible and capable of extracting information from the data to provide the gains in precision that our theoretical results indicate.

\begin{figure}[!ht]
\caption{Assessing the stability of our event-study estimators} 
\label{fig:stability_estimators}
         \centering
          \includegraphics[width=1\textwidth, keepaspectratio]{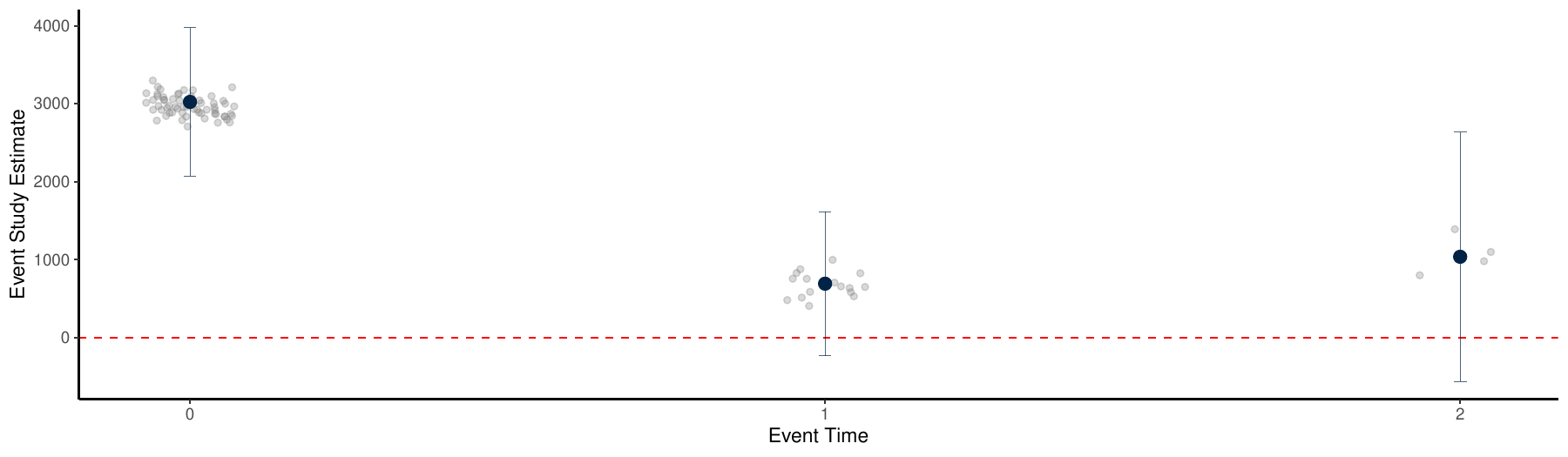}

\justifying
\noindent\scriptsize{\textit{Notes:} The plot displays our efficient event study estimators for event time zero, one, and two, and their 95\% confidence intervals computed using analytical standard errors in blue. In gray, we display all possible event-study estimators that attach unit weight to the different $ATT(g,t)$ estimators in Figure \ref{fig:heatmap_stg_application} that uses a particular $(g', \tprime)$ pair. For $ES(0)$, that are 64 possible estimators ($4\times4\times4$) that leverage different combinations of $ATT(8,8)$, $ATT(9,9)$, and $ATT(10,10)$. For $ES(1)$, we have 16 possible estimators, and for $ES(2)$, we have four possible estimators. }
     
\end{figure} 

We conclude the session by discussing how one can use our results to assess the stability of our results. In Figure \ref{fig:stability_estimators}, we plot our efficient ES estimators with their 95\% confidence intervals in blue, and in gray, we plot all possible event-study estimators that attach unit weight to the different $ATT(g,t)$ estimators in Figure \ref{fig:heatmap_stg_application} that uses a particular $(g', \tprime)$ pair. For $ES(0)$, that are 64 possible estimators ($4\times4\times4$) that leverage different combinations of $ATT(8,8)$, $ATT(9,9)$, and $ATT(10,10)$. For $ES(1)$, we have 16 possible estimators, and for $ES(2)$, we have four possible estimators. As you can see in Figure \ref{fig:stability_estimators}, all these different combinations lead to very stable event-study estimators, indicating the plausibility of Assumption \ref{asm:pt-n-all} in our context. Figure \ref{fig:stability_estimators} can be interpreted as an adaptation of the specification curve analysis \citep{Simonsohn_Simmons_Nelson_2020_spec_curve} to our DiD context.

\section{Conclusion}\label{sec:conclusion}

Empirical researchers routinely face a variety of alternative DiD and ES estimators to choose from. Existing approaches often overlook the fact that different pre-treatment periods and comparison groups may have varying identification power, leading them to treat all pre-treatment periods as equally informative or to use only the last pre-treatment period as a baseline. The semiparametric efficiency framework developed in this paper provides a principled approach to leveraging information from the pre-treatment period to form DiD and ES estimators with minimum variance under the maintained identification assumptions, while remaining agnostic to parametric functional form or parametric error structure. As our simulations and empirical illustration demonstrate, the gains in precision can be substantial, addressing the concern that heterogeneous robust DiD estimators may not be informative compared to more traditional estimators that assume treatment effect heterogeneity away.

In practice, practitioners may want to assess if there is any statistical evidence against using any pre-treatment period as the baseline in their DiD/ES procedure, or against using any specific available comparison group (when treatment is staggered). We present in Appendix \ref{sec:hausman-tests} how to construct Hausman-type tests for overidentification specification test of DiD models, or how to construct alternative estimators that aim to trade off bias and variance \citep{Armstrong_Kline_Sun_2024_adaptive}. We also discuss visualization schemes to informally quantify the sensitivity of the results concerning the usage of different pre-treatment baseline periods, and different comparison groups, or to ``eyeball'' the magnitude of potential model violations.

Our framework opens avenues for several extensions. It would be useful to study semiparametric efficient estimation for nonlinear DiD-type setups such as the Changes-in-Changes model \citep{Athey2006}, setups where treatment can turn on and off over multiple periods \citep{deChaisemartin2020_AER, deChaisemartin2024_intertemporal}, and setups with unbalanced panel or repeated cross-sections \citep{Callaway_Santanna_2021}---in this latter case, we expect new sources of overidentification related to event-study weights. Finally, it is important to stress that all results derived in this paper leveraged the large-$n$, fixed-$T$ panel data regime, as parallel trends are more likely to be satisfied when the number of periods is small. In setups where parallel trends are not plausible, we recommend that researchers use alternative estimators that do not rely on parallel trends assumptions, e.g., \citet{Arkhangelsky2021_SDiD}, \citet{Viviano2021}, and \citet{Imbens_Viviano_2023}. Developing semiparametric efficient estimators in such frameworks is left for future research.

\appendix

\numberwithin{assumption}{section}


\section{Assessing the plausibility of PT assumptions}\label{sec:hausman-tests}

Another consequence of Lemmas \ref{lm:1} and \ref{lm:2} is that Assumption \ref{asm:pt-n-all} can be \emph{directly} tested, as our DiD model is overidentified. In this section, we describe how to construct a Hausman-type test based on event-study parameters in the context of staggered treatment adoption. We focus on event-study parameters as they are often the main parameter of interest in empirical research, and are also often used to assess the plausibility of the identification assumptions; see, e.g., \citet{Roth2022_AERI} and \citet{Borusyak2023} for a discussion.

The main idea of our Hausman-type test is to compare $\widehat{ES} =(\widehat{ES}(e), e \in \mathcal{E})$ as defined in \eqref{eqn:Event_study_efficient}---which is consistent and semiparametrically efficient under Assumption \ref{asm:pt-n-all}---with an event-study estimator that is consistent under Assumption \ref{asm:pt-post} but does not require Assumption \ref{asm:pt-n-all}. Toward this end, let $\widecheck{ES} = (\widecheck{ES}(e),e \in \mathcal{E})$, where each event-study $\widecheck{ES}(e)$ is given by
\begin{align*}
    \widecheck{ES}(e) =\sum_{g \in \mathcal{G}_{\text{trt}}} \frac{\widehat{\pi}_g}{\sum_{g' \in \mathcal{G}_{\text{trt}}} \widehat{\pi}_{g'}} \widecheck{ATT}_{\text{stg}}(g,g+e),
\end{align*}
with 
\begin{align}
\widecheck{ATT}_{\text{stg}}(g,t)= \E_n\left[ \widehat{\widetilde{Y}}^{\text{att(g,t)}}_{g,g-1} \right], \label{eqn:ATT-hat-just-id}
\end{align}
and $\widehat{\widetilde{Y}}^{\text{att(g,t)}}_{g',\tprime}$ as in \eqref{eqn:generic_att_gt_stg_estimated}. Under this construction, $\widecheck{ES}$ is consistent for $ES$ under Assumption \ref{asm:pt-post} but less efficient than $\widehat{ES} =(\widehat{ES}(e), e \in \mathcal{E})$.

Based on these two estimators, we can construct a Hausman-type test statistic as
\begin{align} \label{eqn:Hausman_statistic}
    \widehat{H} =n \left( \widehat{ES}-\widecheck{ES} \right)' \left( \widehat{\aCov}(\widecheck{ES}) - \widehat{\aCov}(\widehat{ES}) \right)^{-1}\left( \widehat{ES}-\widecheck{ES} \right),
\end{align}
where $\widehat{\aCov} (\cdot)$ denotes the corresponding asymptotic covariance estimator for the asymptotic covariance matrices of $\widehat{ES}$ and $\widecheck{ES}$ given by Theorem \ref{thm:efficiency} and Corollary \ref{cor:2-period}.\footnote{It is also straightforward to construct an alternative estimator to $\widehat{\aCov}(\widecheck{ES}) - \widehat{\aCov}(\widehat{ES})$ that leverages the difference of the (efficient) influence function of these two estimators. An advantage of this alternative estimator is that it is positive definite in finite samples. We omit the details as the notation is heavier.} The test rejects the parallel trends assumption for all periods and all groups if $\widehat{H}$ exceeds the corresponding critical value of a $\chi^2(|\mathcal{E}|)$ distribution.\footnote{Note that our Hausman-type test is effectively testing if the event-study aggregation of the parallel trends is the same under Assumptions \ref{asm:pt-n-all} and \ref{asm:pt-post}. It is straightforward to construct a Hausman-type test for Assumptions \ref{asm:pt-n-all} and \ref{asm:pt-post} based on all the estimators of the $ATT(g,t)$'s. However, we anticipate that this would not be as empirically attractive as the one in \eqref{eqn:Hausman_statistic}, as event-study parameters often play a more prominent role than the $ATT(g,t)$'s in setups with staggered treatment adoption.}

\begin{theorem} \label{thm:test}
    Suppose that the estimator $\widehat{ES}$ is constructed under the conditions of Theorem \ref{thm:asy_properties_estimator}, and $\widecheck{ES}$ also satisfies the corresponding conditions. Assume that the covariance matrices estimators $\widehat{\aCov}(\widecheck{ES})$ and $\widehat{\aCov}(\widehat{ES})$ are consistent. Then the test statistic $\widehat{H}$ converges in distribution to a $\chi^2(|\mathcal{E}|)$ distributions, where $|\mathcal{E}|$ denotes the number of elements in $\mathcal{E}$. Also, this Hausman test has nontrivial power against all local alternatives.
\end{theorem}

    Although assessing the plausibility of Assumption \ref{asm:pt-n-all} via the Hausman-type test in Theorem \ref{thm:test} is attractive, it is also worth noting that Assumptions \ref{asm:pt-n-all} and \ref{asm:pt-post} respectively define the largest and smallest sets of conditional moment restrictions for estimating treatment effects. Thus, in practice, it may be the case that the ``plausible'' parallel trends condition lies between these two extremes, and one may be interested in approximating this set of conditional moment restrictions. It is possible to do so by coupling our Hausman-type test with a Holm-Bonferroni \citep{Holm1979} sequential procedure to select conditional moment restrictions that align with Assumption \ref{asm:pt-post}. The main idea is to contrast $\widecheck{ES}$ with event-study estimators that (sequentially) include additional overidentifying restrictions \eqref{eqn:over-id-moments}, and select the event-study estimator with the largest set of restrictions that is not statistically different from $\widecheck{ES}$ (after adjusting for multiple testing). This essentially entails an incremental Sargan Test, as discussed in \citet{Chen_Santos_2018_ECMA} in different contexts.

    The specific procedure is described as follows. Let \(\mathcal{M}\) denote the set of conditional moment restrictions specified by (\ref{eqn:CATT-id}) for all \(t \geq g\), with \(g' = g\) and \(\tprime = g - 1\), corresponding to the restrictions implied by Assumption \ref{asm:pt-post}. For each \(g' > \tprime\) (with \(\tdoubleprime\) fixed at \(1\)), let \(\mathcal{M}_{g',\tprime}\) represent \(\mathcal{M}\) extended by the additional conditional moment restriction given by (\ref{eqn:CATT-id}) for \(g'\) and \(\tprime\). Let \(L\) denote the total number of such models \(\mathcal{M}_{g',\tprime}\). For each \(\mathcal{M}_{g',\tprime}\), we perform a Hausman-type test against the baseline model \(\mathcal{M}\) similar to the approach in (\ref{eqn:Hausman_statistic}) and Theorem \ref{thm:test}. Specifically, we construct event-study estimators \(\widehat{ES}\) using the moment restrictions in \(\mathcal{M}_{g',\tprime}\) and compute the corresponding p-value, \(\mathrm{p}_{g',\tprime}\), of the Hausman statistic in (\ref{eqn:Hausman_statistic}). These p-values are then ordered from smallest to largest as \(\mathrm{p}_{(1)}, \ldots, \mathrm{p}_{(L)}\). Let \(\alpha\) denote the family-wise error rate (e.g., \(\alpha = 0.05\)). The procedure starts with \(\ell = 1\) and compares \(\mathrm{p}_{(1)}\) with \(\frac{\alpha}{L}\). If \(\mathrm{p}_{(1)} \geq \frac{\alpha}{L}\), the procedure terminates without rejecting any models. Otherwise, the conditional moment corresponding to \(\mathrm{p}_{(1)}\) is rejected, and the procedure proceeds to the next step. For each subsequent \(\ell = 2, \ldots, L\), \(\mathrm{p}_{(\ell)}\) is compared with \(\frac{\alpha}{L+1-\ell}\). If \(\mathrm{p}_{(\ell)} \geq \frac{\alpha}{L+1-\ell}\), the procedure terminates; otherwise, the conditional moment corresponding to \(\mathrm{p}_{(\ell)}\) is rejected, and the process continues until either a p-value falls below the threshold or all models are tested.

\begin{remark}\label{rm:adapting_to_misspecification}
    We caveat that using our Hausman-type test $\widehat{H}$ in Theorem \ref{thm:test} to decide whether to report $\widehat{ES}$ or $\widecheck{ES}$ can lead to unnecessarily high mean squared errors compared to an oracle selection procedure \citep{Armstrong_Kline_Sun_2024_adaptive}. Instead of following this approach, one can adopt the adaptive estimation procedure of  \citet{Armstrong_Kline_Sun_2024_adaptive} that can be understood as weighting $\widehat{ES}$ and $\widecheck{ES}$, with weights that are designed to take into account a bias-variance tradeoff. In their context, $\widehat{ES}$ would play the role of the ``restricted'' estimator while $\widecheck{ES}$ would play the role of the ``unrestricted'' estimator. 
\end{remark}

\begin{remark}\label{rm:visualizations}
    Another popular technique commonly used to assess the plausibility of parallel trends assumptions is to examine whether event-study coefficients in pre-treatment periods are close to zero, i.e., whether there is some evidence of parallel pre-treatment trends. Although our proposed efficient event-study estimators leverage all the available pre-treatment information to estimate post-treatment average treatment effects, we can also construct ``placebo'' pre-treatment effects by fixing the estimation procedure in Section \ref{sec:estimation_inference} but allowing for $t<g$. This shares the same spirit as the pre-treatment event study analysis of \citet{Borusyak2023}. Alternatively, one can fix any comparison group and pre-treatment period, and report a pre-treatment event study analysis for these. 
\end{remark}

\section{Extension: Instrumented DiD}\label{sec:Ext1}

We extend the instrumented DiD (DiD-IV) setup in \cite{Miyaji_DiD_IV_2024} to incorporate covariates. There are $T$ time periods: $t=1,2,\cdots,T$. Let $Z_{t}$ denote the instrument in time period $t$ and collect them into the path $Z \coloneqq (Z_1,\cdots,Z_T)$. The instrument is irreversible: $Z_{t} \geq Z_{t-1}$ for all $t$. Therefore, the instrument path is uniquely characterized by the initial date of exposure $G^{\text{IV}} \coloneqq \min \{t:Z_t = 1\}$. The units are grouped based on $G^{\text{IV}}$ instead of on the actual treatment. Denote $G^{\text{IV}}_g \coloneqq \mathbf{1}\{G^{\text{IV}}=g\}$. Let $\mathcal{G}^{\text{IV}}_{\text{trt}}$ be the support of $G^{\text{IV}}$ among the units who are eventually exposed to the instrument. 
Denote $D_t(g)$ as the potential treatment if the unit is first exposed to the instrument in period $g$. Denote $Y_t(d_t)$ as the potential outcome if the treatment in period $t$ is $d_t$. This definition already imposes the no carryover assumption that the potential outcomes depend only on the current treatment
status \citep{deChaisemartin2020_AER,Miyaji_DiD_IV_2024} and the exclusion of the instrument. Let $X$ be a set of pretreatment covariates.

The target parameter is the local average treatment effect for the treated (LATT) defined as
\begin{align*}
    \textit{LATT}(g,t) \coloneqq \mathbb{E}[Y_t(1) - Y_t(0)|G^{\text{IV}}=g,D_t(g) > D_t(\infty)].
\end{align*}

The following assumptions are imposed.

\begin{namedassumption}{DiD-IV} \label{asm:didiv}
\phantom{abcgsfafafaagagagagahahah}
   \begin{enumerate}[(1)]
   \item (Random Sampling) $\{(Y_{i,t=1},\dots,Y_{i,t=T},D_{i,t=1},\dots,D_{i,t=T}, X_i', G^{\text{IV}}_i)'\}_{i=1}^n$ is a random sample from $(Y_{t=1},\dots,Y_{t=T},D_{t=1},\dots,D_{t=T}, X', G^{\text{IV}})'$.
    \item (Overlap) For each $g$, $\mathbb{E}[G^{\text{IV}}_g|X] \in (0,1)$ a.s.
        \item (Monotonicity) $\mathbb{P}(D_t(g) \geq D_t(\infty)|X)=1$ a.s., for $t \geq g$.
        \item (No-anticipation in the first stage) $\mathbb{E}[D_t(g)|G^{\text{IV}}=g,X] = \mathbb{E}[D_t(\infty)|G^{\text{IV}}=g,X], t < g$.
        \item (Parallel trends in the treatment) $\mathbb{E}[D_t(\infty) - D_{t-1}(\infty)|G^{\text{IV}}=g,X] = \mathbb{E}[D_t(\infty) - D_{t-1}(\infty)|G^{\text{IV}}=\infty,X]$, for all $g,t$.
        \item (Parallel trends in the outcome) $\mathbb{E}[Y_{t}(D_t(\infty)) - Y_{t-1}(D_{t-1}(\infty))|G^{\text{IV}}=g,X] = \mathbb{E}[Y_{t}(D_t(\infty)) - Y_{t-1}(D_{t-1}(\infty))|G^{\text{IV}}=\infty,X]$, for all $g,t$.
    \end{enumerate}
\end{namedassumption}

\begin{lemma}\label{lm:LATT-id}
    Under Assumption \ref{asm:didiv}, $LATT(g,t)$ is identified as
    \begin{align*}
       \textit{LATT}(g,t) = \frac{\mathbb{E}\big[G^{\text{IV}}_g (\mathbb{E}[Y_t - Y_{g-1} | G^{\text{IV}}=g, X] - \mathbb{E}[Y_t - Y_{g-1} | G^{\text{IV}}=\infty, X]) \big]}{\mathbb{E}\big[G^{\text{IV}}_g(\mathbb{E}[D_t - D_{g-1} | G^{\text{IV}}=g, X] - \mathbb{E}[D_t - D_{g-1} | G^{\text{IV}}=\infty, X]) \big]}.
    \end{align*}
\end{lemma}

The lemma shows that the LATT parameter is a ratio between two ATT-type parameters. The following moment restrictions define our DiD-IV model under Assumption \ref{asm:didiv}. For simplicity in exposition, we consider the case with a single date of exposure to the instrument $g$. The more general staggered case can be derived similarly but with a more complicated expression of the efficient influence function.

\begin{lemma}[Moment-restrictions for overidentified DiD-IV with a single instrument exposure time] \label{lm:3}
The family of probability distributions of $(Y_{t=1},....Y_{t=T}, D_{t=1},....D_{t=T},X', G^{IV})$ satisfying Assumption \ref{asm:didiv} are observationally equivalent to the family of probability distributions of $(Y_{t=1},....Y_{t=T}, D_{t=1},....D_{t=T} ,X', G^{IV})$ satisfying Assumption \ref{asm:didiv}(1)-(3) and the following set of moment restrictions: for all post-treatment periods $t\in \{g,\dots, T\}$, with probability one, 
\begin{align*}
    \mathbb{E}[G^{\text{IV}}_g (LATT(g,t)_{num} - h_1(g,t,X))] & = 0, \\
    \mathbb{E}[G^{\text{IV}}_g (LATT(g,t)_{den} - h_2(g,t,X))] & = 0, \\
    \mathbb{E}\left[ h_1(g,t,X) - \frac{G^{\text{IV}}_g(Y_{t} - Y_{g-1})}{p^{\text{IV}}_g(X)} + \frac{G^{\text{IV}}_{\infty}(Y_{t} - Y_{g-1})}{p^{\text{IV}}_{\infty}(X)} \Big| X \right] & = 0, \\
    \mathbb{E}\left[ h_2(g,t,X) - \frac{G^{\text{IV}}_g(D_{t} - D_{g-1})}{p^{\text{IV}}_g(X)} + \frac{G^{\text{IV}}_{\infty}(D_{t} - D_{g-1})}{p^{\text{IV}}_{\infty}(X)} \Big| X \right] & = 0, \\
    \mathbb{E}\left[\frac{G^{\text{IV}}_{g}(Y_{\tprime} - Y_1)}{p^{\text{IV}}_{g}(X)} - \frac{G^{\text{IV}}_\infty(Y_{\tprime} - Y_1)}{p^{\text{IV}}_\infty(X)}\Big| X \right] & = 0, \text{for all } 2 \leq \tprime \leq g-1, \\
    \mathbb{E}\left[\frac{G^{\text{IV}}_{g}(D_{\tprime} - D_1)}{p^{\text{IV}}_{g}(X)} - \frac{G^{\text{IV}}_\infty(D_{\tprime} - D_1)}{p^{\text{IV}}_\infty(X)}\Big| X \right] & = 0, \text{for all } 2 \leq \tprime \leq g-1, \\
        \mathbb{E}[G^{\text{IV}}_g - p^{\text{IV}}_g(X)|X]&=0.
\end{align*}
\end{lemma}

Using the moment restrictions, the LATT parameter can be written as $\textit{LATT}(g,t) = \frac{LATT(g,t)_{num}}{LATT(g,t)_{den}}$. The nuisance parameters \(h_1\) and \(h_2\) correspond to the numerator and denominator of the conditional LATT parameter, respectively, while \(LATT(g,t)_{num}\) and \(LATT(g,t)_{den}\) represent their unconditional counterparts. Compared to the moment restrictions in Lemma \ref{lm:1}, the ones in Lemma \ref{lm:3} incorporate the instrument as the treatment and the treatment as the outcome in the restrictions for \(LATT(g,t)_{den}\). Nonetheless, this minor distinction does not affect the efficiency calculations.
To formally define the efficient influence function, let
\begin{align*} 
    \mathbb{IF}_{\tprime}^{latt(g,t),num,Y} & =\Big( G_g(m_{g,t,\tprime}(X) - m_{\infty,t,\tprime}(X) - LATT(g,t)_{num}) \Big) \\
    & + p_g(X) \Big( \frac{G_g}{p_g(X)}(Y_t - Y_{1} - m_{g,t,\tprime}(X)) \\
    & - \frac{G_\infty}{p_\infty(X)}(Y_t - Y_{1} - m_{\infty,t,\tprime}(X)) \Big), 1 \leq \tprime \leq g-1,
\end{align*}
and
\begin{align*} 
    \mathbb{IF}_{\tprime}^{latt(g,t),num,D} & =\Big( G_g(\mu_{g,t,\tprime}(X) - \mu_{\infty,t,\tprime}(X) - LATT(g,t)_{num}) \Big) \\
    & + p_g(X) \Big( \frac{G_g}{p_g(X)}(Y_t - Y_{1} + D_{\tprime} - D_1 - \mu_{g,t,\tprime}(X)) \\
    & - \frac{G_\infty}{p_\infty(X)}(Y_t - Y_{1} + D_{\tprime} - D_1 - \mu_{\infty,t,\tprime}(X)) \Big), 2 \leq \tprime \leq g-1,
\end{align*}
where $\mu_{g,t,\tprime}(X) \coloneqq \mathbb{E}[Y_t - Y_{1} + D_{\tprime} - D_1|G^{\text{IV}}=g,X]$. Denote the column vector that stacks these $2(g-1)-1$ functions by 
\begin{align*}
    \mathbb{IF}^{latt(g,t),num} \coloneqq (\mathbb{IF}_{1}^{latt(g,t),num,Y},\cdots,\mathbb{IF}_{g-1}^{latt(g,t),num,Y},\mathbb{IF}_{2}^{latt(g,t),num,D},\cdots,\mathbb{IF}^{latt(g,t),num,D}_{g-1})',
\end{align*}
and the conditional covariance matrix of $\mathbb{IF}^{latt(g,t),num}$ as $V^{latt(g,t),num}(X)$. Similarly, let
\begin{align*} 
    \mathbb{IF}_{\tprime}^{latt(g,t),den,D} & =\Big( G_g(\tilde{m}_{g,t,\tprime}(X) - \tilde{m}_{\infty,t,\tprime}(X) - LATT(g,t)_{den}) \Big) \\
    & + p_g(X) \Big( \frac{G_g}{p_g(X)}(D_t - D_{1} - \tilde{m}_{g,t,\tprime}(X)) \\
    & - \frac{G_\infty}{p_\infty(X)}(D_t - D_{1} - \tilde{m}_{\infty,t,\tprime}(X)) \Big), 1 \leq \tprime \leq g-1,
\end{align*}
and
\begin{align*} 
    \mathbb{IF}_{\tprime}^{latt(g,t),den,Y} & =\Big( G_g(\tilde{\mu}_{g,t,\tprime}(X) - \tilde{\mu}_{\infty,t,\tprime}(X) - LATT(g,t)_{num}) \Big) \\
    & + p_g(X) \Big( \frac{G_g}{p_g(X)}(D_t - D_{1} + Y_{\tprime} - Y_1 - \tilde{\mu}_{g,t,\tprime}(X)) \\
    & - \frac{G_\infty}{p_\infty(X)}(D_t - D_{1} + Y_{\tprime} - Y_1 - \tilde{\mu}_{\infty,t,\tprime}(X)) \Big), 2 \leq \tprime \leq g-1,
\end{align*}
where $\tilde{m}_{g,t,\tprime}(X) \coloneqq \mathbb{E}[D_t - D_{1}|G^{\text{IV}}=g,X]$, and $\tilde{\mu}_{g,t,\tprime}(X)\coloneqq \mathbb{E}[D_t - D_{1} + Y_{\tprime} - Y_1|G^{\text{IV}}=g,X]$. The vector $\mathbb{IF}^{latt(g,t),den}$ of influence functions and its conditional covariance matrix $V^{latt(g,t),den}(X)$ are defined analogously. 

Compared to the previous efficiency results, the difference in this IV setting is the presence of two parallel trends -- one in the treatment and one in the outcome. To fully utilize all available information, when estimating the numerator of the LATT, we must account for the overidentifying information from both the parallel trend in the outcome and treatment, and vice versa for the denominator. Consequently, the estimation of either the numerator or the denominator involves twice as many influence functions due to the existence of two parallel trends. As before, the efficient influence function is then obtained by optimally weighting these functions. This result is summarized in the following corollary.

\begin{corollary} \label{cor:iv}
    Under Assumption \ref{asm:didiv}, the efficient influence function for $\textit{LATT}(g,t)$, $t \geq g$, is given by
    \begin{align*}
        \frac{\mathbb{EIF}^{latt(g,t),num} - \textit{LATT}(g,t)\mathbb{EIF}^{latt(g,t),den}}{\mathbb{E}\big[G^{\text{IV}}_g(\mathbb{E}[D_t - D_{g-1} | G^{\text{IV}}=g, X] - \mathbb{E}[D_t - D_{g-1} | G^{\text{IV}}=\infty, X]) \big]},
    \end{align*}
    where
    \begin{align*}
        \mathbb{EIF}^{latt(g,t),j} = \frac{\mathbf{1}'V^{latt(g,t),j}(X)^{-1}}{\mathbf{1}'V^{latt(g,t),j}(X)^{-1} \mathbf{1} } \mathbb{IF}^{latt(g,t),j}, j=num,den.
    \end{align*}
    Assuming the second moment of the efficient influence function is finite, the semiparametric efficiency bound for $\textit{LATT}(g,t)$ is its second moment.
\end{corollary}

The efficiency results for staggered exposure to the instrument can be derived analogously to those in Section \ref{sec:staggered}, with the incorporation of additional overidentifying restrictions from multiple comparison groups. Building on the derived efficient influence function, semiparametrically efficient estimators for the LATT parameters can be constructed following the approach outlined in Section \ref{sec:estimation_inference}.

\section{Proofs for theoretical results}  \label{sec:prof.ident}

\textbf{Notation:} For simplicity, we use $t'$ and $t''$ to denote pre-treatment periods instead of $t_{\text{pre}}$ and $t_{\text{pre}}'$. We write $p_{\text{ratio}}$ simply as $p$. We denote the ``generated outcome'' $\tilde{Y}^{att(g,t)}_{g',t_{\text{pre}}}$ as $\theta_{g',t''}(W;p,m)$, making explicit its dependence on the nuisance parameters $p$ and $m$.

\begin{proof}[Proof of Lemmas \ref{lm:1} and \ref{lm:2}]
    We prove only the staggered case for Lemma \ref{lm:2} as Lemma \ref{lm:1} is a more straightforward case.
    It is easy to see that the moment restrictions are implied by the identification assumptions. Therefore, we only prove the converse implication. It suffices to construct a joint distribution of the potential outcomes that is consistent with the observed outcome and satisfies the identification assumptions. Without loss of generality, we can suppress the covariates, since the analysis can be conducted conditional on each value of the covariates. Denote \(Y \equiv (Y_{t=1}, \cdots, Y_{t=T})\) as the vector of observed outcomes and \(\Delta Y \equiv (\Delta Y_2, \cdots, \Delta Y_T)\) as the vector of differenced outcomes, where \(\Delta Y_t \equiv Y_t - Y_{t-1}\). The notations \(Y(g)\), \(\Delta Y_t(g)\), and \(\Delta Y(g)\) are defined analogously for potential outcomes. Let $(Y,G)$ denote the observed variables that already satisfy random sampling and overlap. Once we have the observed distribution of $G$, we can define the parameters $\pi_g$ and $p_g$ accordingly. For each given \(g \in \mathcal{G}\), we construct the potential outcomes as follows:
    \begin{align*}
        & \text{the vectors } Y(g'), g' \in \mathcal{G} \text{ are jointly independent conditional on $G=g$,} \\
        & \text{for $g'=g$: } Y(g) | \{G=g\} \overset{d}{=} Y | \{G=g\} \text{, and} \\
        & \text{for $g'\ne g$: } Y_1(g') | \{G=g\} \overset{d}{=} Y_1(g)| \{G=g\}, \Delta Y(g') | \{G=g\} \overset{d}{=} \Delta Y | \{G=\infty\}, Y_1(g') \perp \Delta Y(g') | \{G=g\},
    \end{align*}
    where $\overset{d}{=}$ denotes equal in distribution. Notice that this construction already ensures that the potential outcome induces the observed outcome because $Y(g) | \{G=g\} \overset{d}{=} Y | \{G=g\}$. The parallel trends condition also holds because $\Delta Y(\infty) | \{G=g\} \overset{d}{=} \Delta Y | \{G=\infty\}$. For the no-anticipation assumption, we have for any $2 \leq t < g$,
    \begin{align*}
        \mathbb{E}[\Delta Y_t(g)|G=g] & = \mathbb{E}[\Delta Y_t|G=g] \\
        & = \mathbb{E}[\Delta Y_t|G = \infty] \text{ (by the moment condition $\mathbb{E}[Y_t-Y_{t-1}|G=\infty] = \mathbb{E}[Y_t-Y_{t-1}|G=g]$)} \\
        & = \mathbb{E}[\Delta Y_t(\infty)|G=g] \text{ (by $\Delta Y(\infty) | \{G=g\} \overset{d}{=} \Delta Y | \{G=\infty\}$)}.
    \end{align*}
    Combining the above equality with the fact that $\mathbb{E}[Y_{1}(\infty)|G=g] = \mathbb{E}[Y_{1}(g)|G=g]$ (by construction $Y_1(\infty) | \{G=g\} \overset{d}{=} Y_1(g)| \{G=g\}$), we obtain the no-anticipation condition. This completes the proof.
\end{proof}

\begin{proof}[Proof of Lemma \ref{lem:id-attgt-over}]
    Since both $\tprime$ and $\tdoubleprime$ are pre-treatment for group $g'$, by Assumption \ref{asm:no anticipation}, the right-hand side of (\ref{eqn:CATT-id}) is equal to
    \begin{align*}
        & \mathbb{E}[Y_t(g) - Y_{\tdoubleprime}(\infty) |G=g,X] - (\mathbb{E}[Y_t(\infty) - Y_{\tprime}(\infty) |G=\infty,X] + \mathbb{E}[Y_{\tprime}(\infty) - Y_{\tdoubleprime}(\infty) |G=g',X]) \\
        = & \mathbb{E}[Y_t(g) - Y_{\tdoubleprime}(\infty) |G=g,X] - \mathbb{E}[Y_t(\infty) - Y_{\tdoubleprime}(\infty) |G=\infty,X] \\
        = & CATT(g,t,X),
    \end{align*}
    where the first equality follows from Assumption \ref{asm:pt-n-all}, and the second inequality follows from standard DiD calculation. The identification of ATT given CATT follows from taking the expectation conditional on $G=g$.
\end{proof}

\begin{proof}[Proof of Theorem \ref{thm:efficiency-fixed-g}]
    We divide the proof into two parts. In the first part, we focus on a submodel for $ATT(g,t)$ at a single time period $t$:
    \begin{align} \label{eqn:model1-tt}
    \begin{split}
            \mathbb{E}[G_g(ATT(g,t) - CATT(g,t,X))] & =0, \\
    \mathbb{E} \left[ CATT(g,t,X) - \frac{G_g(Y_{t} - Y_{1})}{p_g(X)} + \frac{G_\infty(Y_{t} - Y_{1})}{p_\infty(X)} \Big|X \right] & = 0, \\
    \mathbb{E} \left[ \frac{G_g(Y_{t'} - Y_{1})}{p_g(X)} - \frac{G_\infty(Y_{t'} - Y_{1})}{p_\infty(X)} \Big|X \right] & = 0, 2 \leq t' \leq g-1, \\
    \mathbb{E}[G_g - p_g(X) | X] & = 0,
    \end{split}
\end{align}
which is an equivalent representation of the restrictions in Lemma \ref{lm:1}, obtained by replacing \(Y_{g-1}\) in the second line with \(Y_1\) for convenience in the proof.
Then in the second part of the proof, we show that the EIFs remain the same for the entire set of moment restrictions in Lemma \ref{lm:1}.

\paragraph{Part 1}
We follow the orthogonalization method in \cite{Ai_Chen_2012} to derive the efficient influence function and the semiparametric efficiency bound. We adopt the notations in that paper. Denote the available random variables as $W \equiv (Y_{t=1},\cdots,Y_{t=T},X',G)'$ and the parameters $\alpha \equiv (\theta,h)$, where $\theta \equiv ATT(g,t)$ is the finite dimensional parameter, and $h \equiv (CATT(g,t,\cdot),p_g)$ contains the first-stage nuisance parameters. Below, we slightly modify the moment conditions to make the derivation easier while preserving the model and the efficiency bound. Let the unconditional and conditional moments be denoted by $\rho_1$ and $\rho_2$, respectively:
\begin{align*}
    \rho_1(W,\alpha) & \equiv p_g(X) ATT(g,t) - G_g CATT(g,t,X), \\
    \rho_2(W,h) & \equiv 
    \begin{pmatrix}
        CATT(g,t,X) - \frac{G_g(Y_t - Y_1)}{p_g(X)} + \frac{G_\infty(Y_t - Y_1)}{p_\infty(X)} \\
        \cdots \\
        CATT(g,t,X) - \frac{G_g(Y_t - Y_{g-1})}{p_g(X)} + \frac{G_\infty(Y_t - Y_{g-1})}{p_\infty(X)} \\
        G_g - p_g(X) \\
        G_\infty - p_\infty(X)
    \end{pmatrix}.
\end{align*}
In the unconditional moment, we replace $G_g$ by $p_g(X)$, which makes $\rho_1$ and $\rho_2$ orthogonal. This is the orthogonalized unconditional moment one would obtain following the procedure in \cite{Ai_Chen_2012}. To obtain $\rho_2$, we simply rotate the original conditional moment restrictions by using the following invertible matrix 
\begin{align*}
    \begin{pmatrix}
        1 & 0 & 0 & \cdots & 0 & 0 & 0 \\
        1 & 1 & 0 & \cdots & 0 & 0 & 0 \\
        1 & 0 & 1 & \cdots & 0 & 0 & 0 \\
        \cdots \\
        1 & 0 & 0 & \cdots & 1 & 0 & 0 \\
        0 & 0 & 0 & \cdots & 0 & 1 & 0 \\
        0 & 0 & 0 & \cdots & 0 & 0 & 1 
    \end{pmatrix}.
\end{align*}
The subsequent analysis shows that the semiparametric efficiency bound remains invariant to such rotations.
Whenever necessary, we will use $\alpha^*, \theta^*$, and $h^*$ to denote the true values for the respective parameters, and $\alpha,\theta$, and $h$ to denote generic values. Let
\begin{align*}
    \Sigma_1 & \equiv \mathbb{E}[\rho_1 (W,\alpha^*) \rho_1 (W,\alpha^*)'] = \mathbb{E}[p_g(X)^2 (CATT(g,t,X) - ATT(g,t))^2] , \\
    \Sigma_2(X) & \equiv \mathbb{E}[\rho_2 (W,h^*) \rho_1 (W,h^*)'|X] , \\
    m_1(\alpha) & \equiv \mathbb{E}[\rho_1 (W,\alpha)], \\
    m_2(X,\alpha) & \equiv \mathbb{E}[\rho_2 (W,h)|X]. 
\end{align*}
The derivatives of $m_1$ and $m_2$ with respect to $\theta$ are 
\begin{align*}
    \frac{dm_1(\alpha^*)}{d ATT(g,t)} & = \mathbb{E}\left[\frac{d\rho_1(W,\alpha^*)}{d ATT(g,t)} \right] = \mathbb{E}[p_g(X)] = \pi_g, \\
    \frac{dm_2(X,\alpha^*)}{d ATT(g,t)} & = \mathbb{E}\left[\frac{d\rho_2(W,h^*)}{d ATT(g,t)} \right] = 0.
\end{align*}
Let $h_\tau = h^* + \tau r$ be a smooth path in $\tau \in [0,1]$, where $r \equiv (r_1,r_2,r_3)'$, and $h^*+r$ lies in the nuisance parameter space. The derivative of $m_1$ with respect to $h$ (in the direction of $r$) is
\begin{align*}
    \frac{dm_1(\alpha^*)}{d h}[r] & = \mathbb{E}\left[ \frac{d\rho_1(W,\alpha^*)}{d h}[r] \right] = \mathbb{E}[(-p_g(X),ATT(g,t)-CATT(g,t,X),0)r(X)|X] \equiv \mathbb{E}[A_{t}(X)r(X)],
\end{align*}
where $A_{t}(X) \equiv (-p_g(X),ATT(g,t)-CATT(g,t,X),0)$. Similarly, the derivative of $m_2$ with respect to $h$ (in the direction of $r$) is
\begin{align*}
    \frac{dm_2(X,\alpha^*)}{d h}[r] = 
    \begin{pmatrix}
        r_1(X) + \frac{\mathbb{E}[G_g(Y_t - Y_1)|X]}{p_g(X)^2} r_2(X)  -\frac{\mathbb{E}[G_\infty(Y_t - Y_1)|X]}{p_\infty(X)^2} r_3(X) \\
        \cdots \\
        r_1(X) + \frac{\mathbb{E}[G_g(Y_t - Y_{g-1})|X]}{p_g(X)^2} r_2(X) - \frac{\mathbb{E}[G_\infty(Y_t - Y_{g-1})|X]}{p_\infty(X)^2} r_3(X) \\
        -r_2(X) \\
        -r_3(X)
    \end{pmatrix}
    \equiv L(X) r(X),
\end{align*}
where the matrix $L(X)$ is defined as 
\begin{align*}
    L(X) \equiv 
    \begin{pmatrix}
        1 &  \frac{\mathbb{E}[Y_t - Y_1|G=g,X]}{p^*_g(X)} &  -\frac{\mathbb{E}[Y_t - Y_1|G=\infty,X]}{p^*_\infty(X)} \\
        \cdots \\
        1 &  \frac{\mathbb{E}[Y_t - Y_{g-1}|G=g,X]}{p^*_g(X)} &  -\frac{\mathbb{E}[Y_t - Y_{g-1}|G=\infty,X]}{p^*_\infty(X)}  \\
        0 & -1 & 0 \\
        0 & 0 & -1
    \end{pmatrix}.
\end{align*}
By Theorem 2.1 of \cite{Ai_Chen_2012}, the semiparametric efficiency bound for $ATT(g,t)$ is obtained by solving the following minimization problem:
\begin{align*}
    \inf_{r} \quad &  (\pi_g - \mathbb{E}[A_{t}(X)r(X)]) \Sigma_1^{-1}(\pi_g - \mathbb{E}[A_{t}(X)r(X)]) + \mathbb{E}[ (L(X)r(X))'\Sigma_2(X)^{-1}L(X)r(X) ].
\end{align*}
By the calculus of variations, the optimizer $r^*$ satisfies the following first-order condition:
\begin{align*}
&  (\pi_g - \mathbb{E}[A_{t}(X)r^*(X)]) \Sigma_1^{-1}\mathbb{E}[A_{t}(X)r(X)] - \mathbb{E}[ (L(X)r^*(X))'\Sigma_2(X)^{-1}L(X)r(X) ] = 0,
\end{align*}
for all $r$. Denote $B(X) \equiv L(X)'\Sigma_2(X)^{-1}L(X)$. The left-hand side of the above first-order condition can be written as
\begin{align*}
    \mathbb{E}\big[\big((\pi_g - \mathbb{E}[A_{t}(X)r^*(X)])\Sigma_1^{-1} A_{t}(X) - r^*(X)'B(X) \big) r(X)\big].
\end{align*}
In order for this expectation to be zero for all $r$, it must hold that
\begin{align*}
    (\pi_g - \mathbb{E}[A_{t}(X)r^*(X)])\Sigma_1^{-1} A_{t}(X) - r^*(X)'B(X) = 0.
\end{align*}
We can verify that the solution $r^*(X)$ is given by
\begin{align*}
    r^*(X) = \frac{B(X)^{-1}A_{t}(X)' \pi_g}{\mathbb{E}[A_{t}(X)B(X)^{-1}A_{t}(X)'] + \Sigma_1}.
\end{align*}
To see this, notice that we have
\begin{align*}
    \mathbb{E}[A_{t}(X)r^*(X)] =  \frac{\mathbb{E}[A_{t}(X)B(X)^{-1}A_{t}(X)'] \pi_g}{\mathbb{E}[A_{t}(X)B(X)^{-1}A_{t}(X)'] + \Sigma_1}.
\end{align*}
Then we have
\begin{align*}
    r^*(X)'B(X) & = \frac{\pi_g A_{t}(X) B(X)^{-1}B(X)}{\mathbb{E}[A_{t}(X)B(X)^{-1}A_{t}(X)'] + \Sigma_1} = \frac{\pi_g A_{t}(X)}{\mathbb{E}[A_{t}(X)B(X)^{-1}A_{t}(X)'] + \Sigma_1} \\
    & = (\pi_g - \mathbb{E}[A_{t}(X)r^*(X)])\Sigma_1^{-1} A_{t}(X).
\end{align*}
Substituting this expression for $r^*$ into the minimization problem, we obtain the Fisher information:
\begin{align*}
    & (\pi_g - \mathbb{E}[A_{t}(X)r^*(X)]) \Sigma_1^{-1} (\pi_g - \mathbb{E}[A_{t}(X)r^*(X)]) + \mathbb{E}[r^*(X)'B(X)r^*(X)] \\
    = & \frac{\Sigma_1 \pi_g^2}{(\mathbb{E}[A_{t}(X)B(X)^{-1}A_{t}(X)'] + \Sigma_1)^2} + \frac{\mathbb{E}[A_{t}(X)B(X)^{-1}A_{t}(X)'] \pi_g^2}{(\mathbb{E}[A_{t}(X)B(X)^{-1}A_{t}(X)'] + \Sigma_1)^2} \\
    = & \frac{\pi_g^2}{\mathbb{E}[A_{t}(X)B(X)^{-1}A_{t}(X)'] + \Sigma_1}.
\end{align*}
The semiparametric efficiency bound is the inverse of the Fisher information:
\begin{align*}
    \frac{\mathbb{E}[A_{t}(X)B(X)^{-1}A_{t}(X)'] + \Sigma_1}{\pi_g^2}.
\end{align*}
To further simplify the inverse matrix $B(X)^{-1} = (L(X)'\Sigma_2(X)^{-1}L(X))^{-1}$, we aim to show that this matrix has the following block-diagonal form:
\begin{align*}
    (L(X)'\Sigma_2(X)^{-1}L(X))^{-1} = 
    \begin{pmatrix}
        1/(\sum_{g',t'} \sum_{j'=1}^{g-1} s_{j,j'}) & 0 & 0 \\
        0 & p_g(X)(1-p_g(X)) & -p_g(X)p_\infty(X) \\
        0 & -p_g(X)p_\infty(X) & p_\infty(X)(1-p_\infty(X))
    \end{pmatrix},
\end{align*}
where $s_{j,j'}$ denotes the $(j,j')$-th element of the inverse matrix $S(X) \equiv \Sigma_2(X)^{-1}$. We denote the entries of $\Sigma_2(X)$ as
\begin{align*}
    \Sigma_2(X) = 
    \begin{pmatrix}
        \sigma_{1,1} & \cdots & \sigma_{g-1,1} & \sigma_{g,1} & \sigma_{g+1,1} \\
        \vdots & & \vdots & \vdots & \vdots \\
        \sigma_{g-1,1} & \cdots & \sigma_{g-1,g-1} & \sigma_{g,g-1} & \sigma_{g+1,g-1} \\
        \sigma_{g,1} & \cdots & \sigma_{g,g-1} & \sigma_{g,g} & \sigma_{g+1,g} \\
        \sigma_{g+1,1} & \cdots & \sigma_{g+1,g-1} & \sigma_{g+1,g} & \sigma_{g+1,g+1} \\
    \end{pmatrix}.
\end{align*}
Let the bottom-right $2\times 2$ submatrix be denoted by
\begin{align*}
\Pi \equiv
\begin{pmatrix}
    \sigma_{g,g} & \sigma_{g+1,g} \\
    \sigma_{g+1,g} & \sigma_{g+1,g+1}
\end{pmatrix} =
    \begin{pmatrix}
        p_g(X)(1-p_g(X)) & -p_g(X)p_\infty(X) \\
        -p_g(X)p_\infty(X) & p_\infty(X)(1-p_\infty(X))
    \end{pmatrix}.
\end{align*}
For $1 \leq t' \leq g-1$, the term $\sigma_{g,t'}$ is
\begin{align*}
    \sigma_{g,t'} & = \mathbb{E}\left[\left( CATT(g,t,X) - \frac{G_g(Y_t-Y_{t'})}{p_g(X)} + \frac{G_\infty(Y_t-Y_{t'})}{p_\infty(X)} \right) (G_g-p_g(X))|X\right] \\
    & = p_g(X) CATT(g,t,X) - m_{g,t,t'}(X) \\
    & = -(1-p_g(X))m_{g,t,t'}(X) - p_g(X)m_{\infty,t,t'}(X).
\end{align*}
The term $\sigma_{g+1,t'}$ is
\begin{align*}
    \sigma_{g+1,t'} & = \mathbb{E}\left[\left( CATT(g,t,X) - \frac{G_g(Y_t-Y_{t'})}{p_g(X)} + \frac{G_\infty(Y_t-Y_{t'})}{p_\infty(X)} \right) (G_\infty-p_\infty(X))|X\right] \\
    & = p_\infty(X) CATT(g,t,X) + m_{\infty,t,t'}(X) \\
    & = p_\infty(X)m_{g,t,t'}(X) + (1-p_\infty(X))m_{\infty,t,t'}(X).
\end{align*}
Denote two vectors $\bm{a} = (a_1,\cdots,a_{g-1})$ and $\bm{b} = (b_1,\cdots,b_{g-1})$ respectively as
    $a_{t'} \equiv m_{g,t,t'}(X)/p_g(X)$ and
    $b_{t'} \equiv - m_{\infty,t,t'}(X)/p_\infty(X)$.
Then the matrix $L(X)$ can be written as
\begin{align*}
    L(X) = 
    \begin{pmatrix}
        \bm{1}_{g-1} & (\bm{a}',\bm{b}') \\
        \bm{0}_{2} & - I_2
    \end{pmatrix},
\end{align*}
where $\bm{1}_{g-1}$ is a column vector of ones of length $g-1$, $\bm{0}_{2}$ is a column vector of zeros of length 2, and $I_2$ is the 2-dimensional identity matrix. Notice that
\begin{align*}
    a_{t'} \sigma_{g,g} + b_{t'} \sigma_{g+1,g} & = (1-p_g(X))m_{g,t,t'}(X) + p_g(X)m_{\infty,t,t'}(X) = -\sigma_{g,t'}, \\
    a_{t'} \sigma_{g+1,g} + b_{t'} \sigma_{g+1,g+1} & = -p_\infty(X)m_{g,t,t'}(X) - (1-p_\infty(X))m_{\infty,t,t'}(X) = -\sigma_{g+1,t'}.
\end{align*}
In matrix notation, this means that
\begin{align*}
    \Pi 
    \begin{pmatrix}
        \bm{a} \\ 
        \bm{b}
    \end{pmatrix}
    = - 
    \begin{pmatrix}
        \sigma_{g,1} & \cdots & \sigma_{g,g-1} \\
        \sigma_{g+1,1} & \cdots & \sigma_{g+1,g-1} 
    \end{pmatrix} \equiv -
    \begin{pmatrix}
        \tilde{\bm{\sigma}}_{g} \\
        \tilde{\bm{\sigma}}_{g+1}
    \end{pmatrix}
    \implies
    \begin{pmatrix}
        \bm{a} \\ 
        \bm{b}
    \end{pmatrix}
    = -\Pi^{-1}
    \begin{pmatrix}
        \tilde{\bm{\sigma}}_{g} \\
        \tilde{\bm{\sigma}}_{g+1}
    \end{pmatrix},
\end{align*}
where we denote $\tilde{\bm{\sigma}}_{g} \equiv (\sigma_{g,1} \cdots \sigma_{g,g-1})$ and $\tilde{\bm{\sigma}}_{g+1} \equiv (\sigma_{g+1,1} \cdots \sigma_{g+1,g-1})$, which correspond to the second-to-last row and last row of $\Sigma_2(X)$, respectively, with the last two entries removed.
Let $S = \Sigma_2(X)^{-1}$ be denoted by $S = (\bm{s}_1,\cdots,\bm{s}_{g+1})$, where $\bm{s}_j$ are column vectors. Let $\tilde{\bm{s}}_j$ denote the vector $\bm{s}_j$ with its last two entries removed.
We now compute $L(X)'\Sigma_2(X)^{-1}$:
\begin{align*}
    L(X)'S = 
    \begin{pmatrix}
        \bm{1}_{g-1}' & 0 & 0 \\
        \bm{a} & -1 & 0 \\
        \bm{b} & 0 & -1 
    \end{pmatrix}
    S = 
    \begin{pmatrix}
        \bm{1}_{g-1}' \tilde{\bm{s}}_{1} & \cdots & \bm{1}_{g-1}' \tilde{\bm{s}}_{g-1} & \bm{1}_{g-1}' \tilde{\bm{s}}_{g} &  \bm{1}_{g-1}' \tilde{\bm{s}}_{g+1} \\
        0 & \cdots & 0 & \bm{a} \tilde{\bm{s}}_{g} - s_{g,g} & \bm{a} \tilde{\bm{s}}_{g+1} - s_{g+1,g} \\
        0 & \cdots & 0 & \bm{b} \tilde{\bm{s}}_{g} - s_{g+1,g} & \bm{b} \tilde{\bm{s}}_{g+1} - s_{g+1,g+1} \\
    \end{pmatrix}.
\end{align*}
The bottom-left entries are zero because because this submatrix can be written as
\begin{align*}
    \begin{pmatrix}
        \bm{a} \\ 
        \bm{b}
    \end{pmatrix}
    (\tilde{\bm{s}}_1,\cdots,\tilde{\bm{s}}_{g-1}) - 
    \begin{pmatrix}
        \bm{s}_{g,1} & \cdots & \bm{s}_{g,g-1} \\
        \bm{s}_{g+1,1} & \cdots & \bm{s}_{g+1,g-1}
    \end{pmatrix}.
\end{align*}
Factoring out $\Pi^{-1}$, we see that this matrix is equal to $\Pi^{-1}$ multiplied by the following matrix:
\begin{align*}
    & \begin{pmatrix}
        \tilde{\bm{\sigma}}_{g} \\
        \tilde{\bm{\sigma}}_{g+1}
    \end{pmatrix}
    (\tilde{\bm{s}}_1,\cdots,\tilde{\bm{s}}_{g-1})
    + 
    \Pi
    \begin{pmatrix}
        \bm{s}_{g,1} & \cdots & \bm{s}_{g,g-1} \\
        \bm{s}_{g+1,1} & \cdots & \bm{s}_{g+1,g-1}
    \end{pmatrix} \\
    = & 
    \begin{pmatrix}
        \tilde{\bm{\sigma}}_{g} \tilde{\bm{s}}_{t'} + \sigma_{g,g} s_{g,t'} + \sigma_{g+1,g} s_{g+1,t'} \\
        \tilde{\bm{\sigma}}_{g+1} \tilde{\bm{s}}_{t'} + \sigma_{g+1,g} s_{g,t'} + \sigma_{g+1,g+1} s_{g+1,t'}
    \end{pmatrix}_{1 \leq t' \leq g-1} \\
    = & 
    \begin{pmatrix}
        \bm{\sigma}_{g} \bm{s}_{1} & \cdots & \bm{\sigma}_{g} \bm{s}_{g-1}\\
        \bm{\sigma}_{g+1} \bm{s}_{1} & \cdots & \bm{\sigma}_{g+1} \bm{s}_{g-1}
    \end{pmatrix},
\end{align*}
where $\bm{\sigma}_{g}$ and $\bm{\sigma}_{g+1}$ denote the second-to-last and last rows, respectively, of $\Sigma_2(X)$. The above entries are all zero because, by definition, $S$ is the inverse of $\Sigma_2(X)$. Then $L(X)'SL(X)$ is equal to
\begin{align*}
    L(X)'SL(X) = 
    \begin{pmatrix}
        \sum_{t'=1}^{g-1}\bm{1}_{g-1}' \tilde{\bm{s}}_{t'} & 0 & 0 \\
        0 & -\bm{a} \tilde{\bm{s}}_{g} + s_{g,g} & -\bm{a} \tilde{\bm{s}}_{g+1} + s_{g+1,g} \\
        0 & -\bm{b} \tilde{\bm{s}}_{g} + s_{g+1,g} & -\bm{b} \tilde{\bm{s}}_{g+1} + s_{g+1,g+1}
    \end{pmatrix}.
\end{align*}
Observe that the upper-right block is also zero because $L(X)'SL(X)$ is symmetric by construction. The bottom-right $2\times 2$ matrix 
\begin{align*}
    \begin{pmatrix}
        \bm{a} \\
        \bm{b}
    \end{pmatrix} 
    \begin{pmatrix}
        \tilde{\bm{s}}_{g} & \tilde{\bm{s}}_{g+1}
    \end{pmatrix}
    +
    \begin{pmatrix}
         s_{g,g} & s_{g+1,g} \\
        s_{g+1,g} & s_{g+1,g+1}
    \end{pmatrix}
    =  \Pi^{-1}.
\end{align*}
This holds because, after factoring out $\Pi^{-1}$, this matrix equals $\Pi^{-1}$ multiplied by the following matrix
\begin{align*}
    & \begin{pmatrix}
        \tilde{\bm{\sigma}}_{g} \\
        \tilde{\bm{\sigma}}_{g+1}
    \end{pmatrix}
    \begin{pmatrix}
        \tilde{\bm{s}}_{g} & \tilde{\bm{s}}_{g+1}
    \end{pmatrix}
    + 
    \Pi
    \begin{pmatrix}
        s_{g,g} & s_{g+1,g} \\
        s_{g+1,g} & s_{g+1,g+1}
    \end{pmatrix} \\
    = & 
    \begin{pmatrix}
        \tilde{\bm{\sigma}}_{g}\tilde{\bm{s}}_{g} + \sigma_{g,g} s_{g,g} + \sigma_{g+1,g} s_{g+1,g} & \tilde{\bm{\sigma}}_{g}\tilde{\bm{s}}_{g+1} + \sigma_{g,g} s_{g+1,g} + \sigma_{g+1,g} s_{g+1,g+1} \\
        \tilde{\bm{\sigma}}_{g+1}\tilde{\bm{s}}_{g} + \sigma_{g+1,g} s_{g,g} + \sigma_{g+1,g+1} s_{g+1,g} & \tilde{\bm{\sigma}}_{g+1}\tilde{\bm{s}}_{g+1} + \sigma_{g+1,g} s_{g+1,g} + \sigma_{g+1,g+1} s_{g+1,g+1}
    \end{pmatrix} \\
    = & 
    \begin{pmatrix}
        \bm{\sigma}_{g} \bm{s}_{g}  & \bm{\sigma}_{g} \bm{s}_{g+1} \\
        \bm{\sigma}_{g+1} \bm{s}_{g}  & \bm{\sigma}_{g+1} \bm{s}_{g+1}
    \end{pmatrix} 
    = I_2,
\end{align*}
where the last equality follows because $S$ is the inverse of $\Sigma_2(X)$. Therefore, by the property of the block-diagonal matrix, the inverse of $L(X)'SL(X)$ is again a block-diagonal matrix:
\begin{align*}
    (L(X)'SL(X))^{-1} = \text{diag}\left(\left(\sum_{j,j'=1}^{g-1} s_{j,j'}\right)^{-1},\Pi\right).
\end{align*}
Applying this block-diagonal structure, we obtain that
\begin{align} \label{eqn:speb-initial-form}
    A_{t}(X)(L(X)'SL(X))^{-1}A_{t}(X)' = \frac{p_g(X)^2}{\sum_{j,j'=1}^{g-1} s_{j,j'}} + p_g(X)(1-p_g(X))(CATT(g,t,X)-ATT(g,t))^2.
\end{align}
The term $\frac{1}{\sum_{j,j'=1}^{g-1} s_{j,j'}}=|\Sigma_2(X)|/(\sum_{j,j'=1}^{g-1} (-1)^{j+j'} |\Sigma_{2,jj'}|)$, where $\Sigma_{2,jj'}$ denotes the submatrix of $S$ obtained by removing the $j$th row and $j'$th column.  We decompose $\Sigma_2$ into 4 blocks:
    \begin{align*}
        \sigma_2 = 
        \begin{pmatrix}
            V_1 & V_2' \\
            V_2 & \Pi 
        \end{pmatrix},
    \end{align*}
    where $V_1$ and $V_2$ are defined as
    \begin{align*}
        V_1 \equiv 
        \begin{pmatrix}
            \sigma_{11} & \cdots & \sigma_{g-1,1} \\
            \vdots & & \vdots \\
            \sigma_{g-1,1} & \cdots & \sigma_{g-1,g-1}
        \end{pmatrix}, 
        V_2 \equiv 
        \begin{pmatrix}
            \sigma_{g,1} & \cdots & \sigma_{g,g-1} \\
            \sigma_{g+1,1} & \cdots & \sigma_{g+1,g-1}
        \end{pmatrix}.
    \end{align*}
    The determinant of $\Sigma_2$ can be computed using the Schur complement as
    \begin{align*}
        |\Sigma_2| = |\Pi| |V_1 - V_2'\Pi^{-1} V_2|.
    \end{align*}
    The term $|\Pi|$ equals $\sigma_{g,g}\sigma_{g+1,g+1}-\sigma_{g+1,g}^2$.
    The Schur complement of $\Pi$ is the $(g-1) \times (g-1)$ symmetric matrix $V_1 - V_2'\Pi^{-1} V_2$, whose $(j,j')$-th elements is
    \begin{align*}
        \sigma_{j,j'} - \frac{\sigma_{g,j}\sigma_{g,j'}\sigma_{g+1,g+1} - \sigma_{g,j}\sigma_{g+1,j'}\sigma_{g+1,g} - \sigma_{g,j'}\sigma_{g+1,j}\sigma_{g+1,g} - \sigma_{g+1,j}\sigma_{g+1,j'}\sigma_{g,g}}{\sigma_{g,g}\sigma_{g+1,g+1}-\sigma_{g+1,g}^2}.
    \end{align*}
    A direct calculation verifies that the above term is equal to the $(j,j')$-th entry of $V^*_{gt}(X)$, yielding the identity $V_1 - V_2'\Pi^{-1} V_2 = V^*_{gt}(X)$, where recall that $V^*_{gt}(X)$ is the $(g-1)\times (g-1)$ matrix whose $(j,k)$-th element is
    \begin{align} 
        \frac{1}{p_g(X)} \Cov(Y_{t} - Y_{j},Y_{t} - Y_{k} |G=g,X) + \frac{1}{1 - p_g(X)} \Cov (Y_{t} - Y_{j} ,Y_{t} - Y_{k}  |G=\infty,X ).\label{eqn:V-star2}
    \end{align}
    The same procedure can be applied to each $\Sigma_{2,jj'}$ and show that the determinant of each block $\Sigma_{2,jj'}$ is equal to $|\Pi|$ multiplied by $|V^*_{gt,jj'}(X)|$, where $V^*_{gt,jj'}(X)$ is the minor of $V^*_{gt}(X)$ formed by deleting its $j$-th row and $j'$-th column. Substituting these into (\ref{eqn:speb-initial-form}), we find that the semiparametric efficiency bound equals:
    \begin{align*}
    & \frac{\mathbb{E}[A_{t}(X)B(X)^{-1}A_{t}(X)'] + \Sigma_1}{\pi_g^2} \\
    = & \frac{1}{\pi_g^2}\mathbb{E}\left[ \frac{p_g(X)^2}{\sum_{j,j'=1}^{g-1} s_{j,j'}} + p_g(X)(CATT(g,t,X)-ATT(g,t))^2 \right] \\
    = & \frac{1}{\pi_g^2}\mathbb{E}\left[ \frac{p_g(X)^2|V^*_{gt}(X)|}{\sum_{j,j'=1}^{g-1} (-1)^{j+j'}|V^*_{gt,jj'}(X)|} + p_g(X)(CATT(g,t,X)-ATT(g,t))^2 \right].
    \end{align*}
    The efficient influence function is the efficiency bound multiplied by the efficient score function. The efficient score function is given by the proof of Theorem 2.1 in \citet{Ai_Chen_2012}:
    \begin{align*}
        & - \left( \pi_g - \mathbb{E}[A_{t}(X)r^*(X)]  \right) \Sigma_1^{-1} (p_g(X)(ATT(g,t) - CATT(g,t,X))) + (L(X)r^*(X))' \Sigma_{2}(X)^{-1} \rho_2(W,h^*).
    \end{align*}
    The first part of the efficient score is equal to
    \begin{align*}
        \left( \pi_g - \mathbb{E}[A_{t}(X)r^*(X)]  \right) \Sigma_1^{-1} (p_g(X)(ATT(g,t) - CATT(g,t,X))) = \frac{p_g(X)(ATT(g,t) - CATT(g,t,X))}{\pi_g\Omega_g(ATT(g,t))}.
    \end{align*}
    In the second part, $L(X)r^*(X)$ is equal to
    \begin{align*}
        L(X)r^*(X) = \frac{L(X)B(X)^{-1}A_{t}(X)'}{\pi_g\Omega_g(ATT(g,t))}.
    \end{align*}
    We calculate the term $A_{t}(X)'$ in the above numerator. Denote $\Bar{s} \equiv \sum_{j,j'=1}^{g-1}s_{j,j'}$. The product $L(X)B(X)^{-1}$ is 
    \begin{align*}
        &L(X)B(X)^{-1}  \\
        = & 
        \begin{pmatrix}
        1 &  \frac{m_{g,t,1}(X)}{p_g(X)} &  -\frac{m_{\infty,t,1}(X)}{p_\infty(X)} \\
        \vdots \\
        1 &  \frac{m_{g,t,g-1}(X)}{p_g(X)} &  -\frac{m_{\infty,t,g-1}(X)}{p_\infty(X)}  \\
        0 & -1 & 0 \\
        0 & 0 & -1
    \end{pmatrix} 
    \begin{pmatrix}
        \left( \Bar{s} \right)^{-1} & 0 & 0 \\
        0 & p_g(X) (1-p_g(X)) & -p_g(X)p_\infty(X) \\
        0 & -p_g(X)p_\infty(X) & p_\infty(X) (1-p_\infty(X))
    \end{pmatrix} \\
    = &
    \begin{pmatrix}
        \left( \Bar{s} \right)^{-1} & (1-p_g(X))m_{g,t,1}(X) + p_g(X)m_{\infty,t,1}(X) & -p_\infty(X)m_{g,t,1}(X) - (1-p_\infty(X))m_{\infty,t,1}(X) \\
        \vdots & \vdots & \vdots  \\
        \left( \Bar{s} \right)^{-1} & (1-p_g(X))m_{g,t,g-1}(X) + p_g(X)m_{\infty,t,g-1}(X) & -p_\infty(X)m_{g,t,1}(X) - (1-p_\infty(X))m_{\infty,t,g-1}(X)  \\ 
         0 & -p_g(X) (1-p_g(X)) & p_g(X)p_\infty(X) \\
         0 & p_g(X)p_\infty(X) & -p_\infty(X) (1-p_\infty(X))
    \end{pmatrix} \\
    = &
    \begin{pmatrix}
        ((\mathbf{1}_{g-1}(\Bar{s})^{-1})',0,0)' & -\bm{\sigma}_g & -\bm{\sigma}_{g+1}
    \end{pmatrix}.
    \end{align*}
    Multiplying the above matrix with $A_{t}(X)'$, we obtain
    \begin{align*}
     L(X)B(X)^{-1} A_{t}(X)' & = \begin{pmatrix}
        ((\mathbf{1}_{g-1}(\Bar{s})^{-1})',0,0)' & -\bm{\sigma}_g & -\bm{\sigma}_{g+1}.
        \end{pmatrix}
        \begin{pmatrix}
        -p_g(X) \\
        ATT(g,t) - CATT(g,t,X) \\
        0
    \end{pmatrix} \\
    & = - (p_g(X)(\Bar{s})^{-1}(\mathbf{1}_{g-1}',0,0)' + (ATT(g,t) - CATT(g,t,X))\bm{\sigma}_g).
    \end{align*}
    Multiplying the transpose of this matrix with the inverse of $\Sigma_2(X)$, we obtain
    \begin{align*}
        & - (p_g(X)(\Bar{s})^{-1}(\mathbf{1}_{g-1}',0,0) + (ATT(g,t) - CATT(g,t,X))\bm{\sigma}_g') (\bm{s}_1,\cdots,\bm{s}_{g+1}) \\
        = & - p_g(X)(\Bar{s})^{-1}((\mathbf{1}_{g-1}',0,0)\bm{s}_1, \cdots, (\mathbf{1}_{g-1}',0,0)\bm{s}_{g+1}) - (ATT(g,t) - CATT(g,t,X)) (0,\cdots,0,1,0) \\
        = & -\frac{p_g(X)}{\Bar{s}} \left( \sum_{j'=1}^{g-1}s_{1j'},\cdots, \sum_{j'=1}^{g-1}s_{g+1,j'}\right) - (ATT(g,t) - CATT(g,t,X)) (0,\cdots,0,1,0)
    \end{align*}
    because $\bm{\sigma}_g'\bm{s}_j = \mathbf{1}\{g=j\}$. Lastly, we multiply the above matrix with the first stage moments $\rho_2$ and obtain that
    \begin{align*}
        & L(X)B(X)^{-1}A_{t}(X)' \Sigma_{2}(X)^{-1} \rho_2(W,h^*) \\
        = & -\frac{p_g(X)}{\Bar{s}} \sum_{j,j'=1}^{g-1} s_{j,j'} \left(CATT(g,t,X)- \frac{G_g(Y_t - Y_j)}{p_g(X)} + \frac{G_\infty(Y_t - Y_j)}{p_\infty(X)} \right) \\
        & -\frac{p_g(X)}{\Bar{s}} \sum_{j'=1}^{g-1} s_{g,j'} (G_g - p_g(X)) -\frac{p_g(X)}{\Bar{s}} \sum_{j'=1}^{g-1} s_{g+1,j'} (G_\infty - p_\infty(X)) \\
        & - (ATT(g,t) - CATT(g,t,X))(G_g - p_g(X)).
    \end{align*}
    This gives the expression for the efficient score. The efficient influence function is equal to the efficient score premultiplied by the efficiency bound:
    \begin{align*}
        \mathbb{EIF}_g(ATT(g,t)) = \frac{1}{\pi_g} \Bigg( & G_g(CATT(g,t,X) - ATT(g,t)) \\
        & + \frac{p_g(X)}{\Bar{s}} \sum_{j,j'=1}^{g-1} s_{j,j'} \left(\frac{G_g(Y_t - Y_j)}{p_g(X)} - \frac{G_\infty(Y_t - Y_j)}{p_\infty(X)} - CATT(g,t,X) \right) \\
        &-\frac{p_g(X)}{\Bar{s}} \sum_{j'=1}^{g-1} s_{g,j'} (G_g - p_g(X)) -\frac{p_g(X)}{\Bar{s}} \sum_{j'=1}^{g-1} s_{g+1,j'} (G_\infty - p_\infty(X)) \Bigg).
    \end{align*}
    By direct calculation, we can show that the EIF can be written as
    \begin{align}
        \mathbb{EIF}_g(ATT(g,t)) = \frac{1}{\pi_g} \Bigg( & G_g(CATT(g,t,X) - ATT(g,t)) + \sum_{j,j'=1}^{g-1} \frac{s_{j,j'}}{\Bar{s}} G_g(Y_t - Y_j - m_{g,t,j}(X)) \nonumber \\
         & - \sum_{j,j'=1}^{g-1} \frac{s_{j,j'}}{\Bar{s}} \frac{p_g(X)}{p_\infty(X)} G_\infty(Y_t - Y_j - m_{\infty,t,j}(X)) \Bigg) \nonumber \\
        & + \frac{p_g(X)}{\Bar{s}} (G_g - p_g(X)) \sum_{j'=1}^{g-1}\left( \sum_{g',t'}s_{j,j'} \frac{m_{g,t,j}}{p_g} - s_{g,j'} \right), \label{eqn:EIF-extra-term1} \\
        & + \frac{p_g(X)}{\Bar{s}} (G_\infty - p_\infty(X)) \sum_{j'=1}^{g-1}\left( \sum_{g',t'}s_{j,j'} \frac{m_{\infty,t,j}}{p_\infty} - s_{g+1,j'} \right). \label{eqn:EIF-extra-term2}
    \end{align}
    We want to show that the last two terms on the left-hand side of the above equation, (\ref{eqn:EIF-extra-term1}) and (\ref{eqn:EIF-extra-term2}), are zero. Notice that for each $1 \leq j' \leq g-1$, the term
    \begin{align*}
        \sum_{g',t'}s_{j,j'} \frac{m_{g,t,j}}{p_g} - s_{g,j'}
    \end{align*}
    is equal to the $j'$th row of $\Sigma_2(X)^{-1}$ multiplied by the second column of $L(X)$. In other words, it is the $(g,j')$-th element of the matrix $L(X)'\Sigma_2(X)^{-1}$. By the previous analysis, this term is zero for any $j \leq g-1$, which implies that (\ref{eqn:EIF-extra-term1}) evaluates to zero. Similarly, the term 
    \begin{align*}
        \sum_{g',t'}s_{j,j'} \frac{m_{\infty,t,j}}{p_\infty} - s_{g+1,j'}
    \end{align*}
    is the $(g+1,j')$-th element of the matrix $L(X)'\Sigma_2(X)^{-1}$, which is also equal to zero, implying that (\ref{eqn:EIF-extra-term2}) is zero. The weights $s_{j,j'}/\Bar{s}$ can be represented using $V^*_{gt}(X)$ as
    \begin{align*}
        \frac{s_{j,j'}}{\Bar{s}} = \frac{(-1)^{j+j'}|\Sigma_{2,jj'}|/|\Sigma_{2}|}{\sum_{j,j'=1}^{g-1}(-1)^{j+j'}|\Sigma_{2,jj'}|/|\Sigma_{2}|} =  \frac{(-1)^{j+j'}|V^*_{gt,jj'}(X)|/|V^*_{gt}(X)}{\sum_{j,j'=1}^{g-1}(-1)^{j+j'}|V^*_{gt,jj'}(X)|/|V^*_{gt}(X)|} = \frac{(j,j')\text{th entry of } V^*_{gt}(X)^{-1}}{\mathbf{1}'V_{gt}^*(X)^{-1} \mathbf{1}}.
    \end{align*}
    Notice that the first part of the EIF, $G_g(CATT(g,t,X) - ATT(g,t))$, does not depend on $j$ and can be put into the weighted average.
    The expression of the EIF becomes
    \begin{align*}
        \mathbb{EIF}_g(ATT(g,t)) = \frac{\mathbf{1}'V_{gt}^*(X)^{-1}}{\mathbf{1}'V_{gt}^*(X)^{-1} \mathbf{1} } \mathbb{IF}_g(ATT(g,t)).
    \end{align*}
    Lastly, we want to show that the weights can be represented by using $V_{gt}(X)$. Notice that $V_{gt}(X)$ is equal to
    \begin{align*}
        V_{gt}(X) = \frac{p_g(X)^2}{\pi_g^2} V^*_{gt}(X) + c(X) \mathbf{1} \mathbf{1}',
    \end{align*}
    where $c(X)$ is
    \begin{align*}
        c(X) \equiv \frac{1}{\pi_g^2}p_g(X)(1-p_g(X))(CATT(g,t,X) - ATT(g,t))^2.
    \end{align*}
    By the Sherman–Morrison formula, $V_{gt}(X)^{-1}$ is equal to
    \begin{align*}
        V_{gt}(X)^{-1} = \frac{\pi_g^2}{p_g(X)^2} V^*_{gt}(X)^{-1} - \frac{\frac{c(X)\pi_g^4}{p_g(X)^4}V^*_{gt}(X)^{-1}\mathbf{1}\mathbf{1}' V^*_{gt}(X)^{-1}}{1 + \frac{c(X)\pi_g^2}{p_g(X)^2} \mathbf{1}'V^*_{gt}(X)^{-1}\mathbf{1}}.
    \end{align*}
    Therefore, we have 
    \begin{align*}
        \mathbf{1}'V_{gt}(X)^{-1} & = \frac{\pi_g^2}{p_g(X)^2} \mathbf{1}'V^*_{gt}(X)^{-1} - \frac{\frac{c(X)\pi_g^4}{p_g(X)^4}\mathbf{1}'V^*_{gt}(X)^{-1}\mathbf{1}\mathbf{1}' V^*_{gt}(X)^{-1}}{1 + \frac{c(X)\pi_g^2}{p_g(X)^2} \mathbf{1}'V^*_{gt}(X)^{-1}\mathbf{1}} \\
        & = \mathbf{1}'V^*_{gt}(X)^{-1} \left( \frac{\pi_g^2}{p_g(X)^2} - \frac{\frac{c(X)\pi_g^4}{p_g(X)^4}\mathbf{1}'V^*_{gt}(X)^{-1}\mathbf{1}}{1 + \frac{c(X)\pi_g^2}{p_g(X)^2} \mathbf{1}'V^*_{gt}(X)^{-1}\mathbf{1}} \right),
    \end{align*}
    and
    \begin{align*}
        \mathbf{1}'V_{gt}(X)^{-1}\mathbf{1} = \mathbf{1}'V^*_{gt}(X)^{-1}\mathbf{1} \left( \frac{\pi_g^2}{p_g(X)^2} - \frac{\frac{c(X)\pi_g^4}{p_g(X)^4}\mathbf{1}'V^*_{gt}(X)^{-1}\mathbf{1}}{1 + \frac{c(X)\pi_g^2}{p_g(X)^2} \mathbf{1}'V^*_{gt}(X)^{-1}\mathbf{1}} \right).
    \end{align*}
    Their ratio is equal to
    \begin{align*}
        \frac{\mathbf{1}'V_{gt}(X)^{-1}}{\mathbf{1}'V_{gt}(X)^{-1}\mathbf{1}} = \frac{\mathbf{1}'V^*_{gt}(X)^{-1}}{\mathbf{1}'V^*_{gt}(X)^{-1}\mathbf{1}}.
    \end{align*}
    This completes the first part of the proof, which gives the EIF for $ATT(g,t)$ in the submodel constructed for a single time period $t$. 

    \paragraph{Part 2}
    For the second part of the proof, we examine the entire set of moments in Lemma \ref{lm:1} for all post-treatment periods: 
    \begin{align*}
    \mathbb{E}[G_g(ATT(g,t) - CATT(g,t,X))] & =0, \text{ for all $t \in [g,T]$,} \\
    \mathbb{E} \left[ CATT(g,t,X) - \frac{G_g(Y_{t} - Y_{g-1})}{p_g(X)} + \frac{G_\infty(Y_{t} - Y_{g-1})}{p_\infty(X)} \Big|X \right] & = 0, \text{ for all $t \in [g,T]$,} \\
    \mathbb{E} \left[ \frac{G_g(Y_{t'} - Y_{1})}{p_g(X)} - \frac{G_\infty(Y_{t'} - Y_{1})}{p_\infty(X)} \Big|X \right] & = 0, 2 \leq t' \leq g-1, \\
    \mathbb{E}[G_g - p_g(X) | X] & = 0.
    \end{align*}
    Intuitively, the additional moment conditions do not alter the efficiency bound of $ATT(g,t)$, and we aim to establish this formally.
    We first show that, to derive the EIF of a single $ATT(g,t)$, we may remove the unconditional moments (in the first line) corresponding time periods other than $t$. In other words, it suffices to examine the following model:
    \begin{align} \label{eqn:model1-t}
    \begin{split}
            \mathbb{E}[G_g(ATT(g,t) - CATT(g,t,X))] & =0, \\
    \mathbb{E} \left[ CATT(g,t,X) - \frac{G_g(Y_{t} - Y_{g-1})}{p_g(X)} + \frac{G_\infty(Y_{t} - Y_{g-1})}{p_\infty(X)} \Big|X \right] & = 0, \text{ for all $t \in [g,T]$,} \\
    \mathbb{E} \left[ \frac{G_g(Y_{t'} - Y_{1})}{p_g(X)} - \frac{G_\infty(Y_{t'} - Y_{1})}{p_\infty(X)} \Big|X \right] & = 0, 2 \leq t' \leq g-1, \\
    \mathbb{E}[G_g - p_g(X) | X] & = 0.
    \end{split}
    \end{align}
    To prove this point, we use the notations $\tilde{\rho}_1,\tilde{\rho}_2,\tilde{m}_1,\tilde{m}_2,\tilde{\Sigma}_1,\tilde{\Sigma}_2$ to denote the corresponding terms in the larger model. Denote $\theta = (ATT(g,t):t \in [g,T])$ as the vector of finite-dimensional parameter and $h = (CATT(g,g,\cdot),\cdots,CATT(g,T,\cdot),p_g)$ the nuisance parameters. The derivatives are 
    \begin{alignat*}{2}
    \frac{d\tilde{m}_1(\alpha^*)}{d \theta} & = \pi_g I , \quad &
    \frac{d\tilde{m}_2(X,\alpha^*)}{d \theta} & = 0 , \\
    \frac{d\tilde{m}_1(\alpha^*)}{d h}[r] & = A(X)r(X) , \quad &
    \frac{d\tilde{m}_2(X,\alpha^*)}{d h}[r] & = \tilde{L}(X)r(X),
\end{alignat*}
    where $A$ essentially stacks all the $A_t$'s, i.e., 
    \begin{align*}
        A(X) = (-p_g(X) I , (ATT(g,g) - CATT(g,T,X),\cdots,ATT(g,g) - CATT(g,T,X))).
    \end{align*}
    We omit the specific expression of $\tilde{L}(X)$ for now. Define $\tilde{B}(X) \equiv \tilde{L}(X)' \tilde{\Sigma}_2(X)^{-1} \tilde{L}(X)$. We want to solve the following optimization:
    \begin{align*}
    \inf_{r} \quad &  (\pi_g I - \mathbb{E}[A(X)r(X)]) \tilde{\Sigma}_1^{-1}(\pi_g - \mathbb{E}[A(X)r(X)]) + \mathbb{E}[ (\tilde{L}(X)r(X))'\tilde{\Sigma}_2(X)^{-1}\tilde{L}(X)r(X) ].
\end{align*}
By solving the first-order condition, we obtain the optimal $r^*$ as
\begin{align*}
    r^*(X)' = \pi_g (\tilde{\Sigma}_1 + \mathbb{E}[A(X)\tilde{B}(X)^{-1}A(X)'])^{-1} A(X) \tilde{B}(X)^{-1}.
\end{align*}
Therefore, the efficient score is
\begin{align*}
    & -\pi_g(\tilde{\Sigma}_1 + \mathbb{E}[A(X)\tilde{B}(X)^{-1}A(X)'])^{-1} \tilde{\rho}_1 \\ & + \pi_g (\tilde{\Sigma}_1 + \mathbb{E}[A(X)\tilde{B}(X)^{-1}A(X)'])^{-1} A(X) \tilde{B}(X)^{-1} \tilde{L}(X)' \tilde{\Sigma}_2(X)^{-1} \tilde{\rho}_2.
\end{align*}
The efficiency bound is the inverse of the expected outer product of the efficient score: 
\begin{align*}
    (\tilde{\Sigma}_1 + \mathbb{E}[A(X)\tilde{B}(X)^{-1}A(X)']) / \pi_g^2.
\end{align*}
The efficient influence function is equal to the efficiency score pre-multiplied by the efficiency bound:
\begin{align*}
    -\frac{1}{\pi_g} \tilde{\rho}_1 + \frac{1}{\pi_g}A(X) \tilde{B}(X)^{-1} \tilde{L}(X)' \tilde{\Sigma}_2(X)^{-1} \tilde{\rho}_2.
\end{align*}
For each $ATT(g,t)$ as an entry of $\theta$, its EIF is equal to the corresponding row in the above expression:
\begin{align} \label{eqn:EIF-model1-t}
    & -\frac{1}{\pi_g} p_g(X)(ATT(g,t) - CATT(g,t,X)) \nonumber \\
    & + \frac{1}{\pi_g}(0,\cdots,0,-p_g(X),0,\cdots,0,(ATT(g,t) - CATT(g,t,X))) \tilde{B}(X)^{-1} \tilde{L}(X)' \tilde{\Sigma}_2(X)^{-1} \tilde{\rho}_2.
\end{align}
Observe that this is exactly the expression for the EIF we would obtain for the model given by (\ref{eqn:model1-t}). This proves that to derive the EIF of a single $ATT(g,t)$, we can remove the unconditional moments (in the first line) for time periods other than $t$. 

The remaining task for the second part of the proof is to show that the conditional moments defining the irrelevant/redundant $CATT(g,t'',\cdot), t'' \ne t$ can also be removed. If this is true, then we can claim that the EIFs derived in the first part (based on the model given in (\ref{eqn:model1-tt})) coincide with the EIFs for the entire set of moment restrictions in Lemma \ref{lm:1}. This can be shown by induction. Assume first that there is only one additional conditional moment for an irrelevant $CATT(g,t'',X)$. For convenience, this moment is placed at the end of the model:
\begin{align*} 
            \mathbb{E}[G_g(ATT(g,t) - CATT(g,t,X))] & =0, \\
    \mathbb{E} \left[ CATT(g,t,X) - \frac{G_g(Y_{t} - Y_{g-1})}{p_g(X)} + \frac{G_\infty(Y_{t} - Y_{g-1})}{p_\infty(X)} \Big|X \right] & = 0, \\
    \mathbb{E} \left[ \frac{G_g(Y_{t'} - Y_{1})}{p_g(X)} - \frac{G_\infty(Y_{t'} - Y_{1})}{p_\infty(X)} \Big|X \right] & = 0, 2 \leq t' \leq g-1, \\
    \mathbb{E}[G_g - p_g(X) | X] & = 0, \\
    \mathbb{E} \left[ CATT(g,t'',X) - \frac{G_g(Y_{t''} - Y_{g-1})}{p_g(X)} + \frac{G_\infty(Y_{t''} - Y_{g-1})}{p_\infty(X)} \Big|X \right] & = 0.
    \end{align*}
    The parameters of this model are $\theta = ATT(g,t)$ and $h = (CATT(g,t,\cdot),p_g,CATT(g,t'',\cdot))$. The second term of the EIF given in (\ref{eqn:EIF-model1-t}) now becomes
    \begin{align} \label{eqn:EIF-model1-t-part2}
        \frac{1}{\pi_g}(A_t(X),0) \tilde{B}(X)^{-1} \tilde{L}(X)' \tilde{\Sigma}_2(X)^{-1} \tilde{\rho}_2.
    \end{align}
    The goal is to show that this is the same as the one for the model derived in the first part of the proof
    \begin{align*}
        \frac{1}{\pi_g}A_t(X) B(X)^{-1} L(X)' \Sigma_2(X)^{-1} \rho_2.
    \end{align*}
    Since the redundant $CATT(g,t'',\cdot)$ does not appear in any other moments except the last one, we can decompose $\tilde{L}$ as 
    \begin{align*}
    \tilde{L} = 
        \begin{pmatrix}
            L & \bm{0} \\
            \ell' & 1
        \end{pmatrix},
    \end{align*}
    where the specific expression for $\ell$ is not important here. We write the inverse matrix $\tilde{\Sigma}_2(X)^{-1} \equiv \Tilde{S}$ as the following $2 \times 2$ block matrix
    \begin{align*}
        \Tilde{S} = 
        \begin{pmatrix}
            \tilde{S}_{UL} & \tilde{S}_{LL}' \\
            \tilde{S}_{LL} & \tilde{S}_{LR}
        \end{pmatrix}.
    \end{align*}
    The matrix $\tilde{B}$ is equal to
    \begin{align*}
        \tilde{B} & = \tilde{L}' \tilde{S} \tilde{L} = \begin{pmatrix}
            L' & \ell \\
            \bm{0}' & 1
        \end{pmatrix}
        \begin{pmatrix}
            \tilde{S}_{UL} & \tilde{S}_{LL}' \\
            \tilde{S}_{LL} & \tilde{S}_{LR}
        \end{pmatrix}
        \begin{pmatrix}
            L & \bm{0} \\
            \ell' & 1
        \end{pmatrix} \\
        & = 
        \begin{pmatrix}
            L'\tilde{S}_{UL}L + \ell \tilde{S_{LL}} + L'\tilde{S}_{LL}' \ell' + \ell \tilde{S}_{LR} \ell' & L'\tilde{S}_{LL}' + \ell \tilde{S}_{LR} \\
            \tilde{S}_{LL}L + \tilde{S}_{LR}\ell' & \tilde{S}_{LR}
        \end{pmatrix}.
    \end{align*}
    Using Schur complement for block matrix inversion, we know that the upper-left block of $\tilde{B}^{-1}$ is equal to the inverse of the Schur complement of $\tilde{S}_{LR}$:
    \begin{align*}
        & \left( L'\tilde{S}_{UL}L + \ell \tilde{S}_{LL} + L'\tilde{S}_{LL}' \ell' + \ell \tilde{S}_{LR} \ell' - ( L'\tilde{S}_{LL}' + \ell \tilde{S}_{LR}) S_{LR}^{-1} (\tilde{S}_{LL}L + \tilde{S}_{LR}\ell') \right)^{-1} \\
        = & (L'(\tilde{S}_{UL} - \tilde{S}_{LL}'\tilde{S}_{LR}^{-1}\tilde{S}_{LL})L)^{-1} = (L'\Sigma_2^{-1}L)^{-1},
    \end{align*}
    where the second inequality follows from the fact that $\Sigma_2$ is the upper-left block of the inverse of $\tilde{S}$. Similarly, we can show that the upper-right block of $\tilde{B}^{-1}$ is equal to $-(L'\Sigma_2^{-1}L)^{-1}(L'\tilde{S}_{LL}' + \ell \tilde{S}_{LR})\tilde{S}_{LR}^{-1}$. Now the expression in (\ref{eqn:EIF-model1-t-part2}) becomes
    \begin{align*}
        & \frac{1}{\pi_g}A_t(X) 
        \begin{pmatrix}
            (L'\Sigma_2^{-1}L)^{-1} & -(L'\Sigma_2^{-1}L)^{-1}(L'\tilde{S}_{LL}' + \ell \tilde{S}_{LR})\tilde{S}_{LR}^{-1}
        \end{pmatrix}
        \begin{pmatrix}
            L'\tilde{S}_{UL}' + \ell \tilde{S}_{LL} & L'\tilde{S}_{LL}' + \ell \tilde{S}_{LR} \\
            \tilde{S}_{LL} & S_{LR}
        \end{pmatrix}
        \tilde{\rho}_2 \\
        = & \frac{1}{\pi_g}A_t(X)  (L'\Sigma_2^{-1}L)^{-1} ( L'\Sigma_2^{-1}L,0) \tilde{\rho}_2 \\
        = &  \frac{1}{\pi_g}A_t(X) B(X)^{-1} L(X)' \Sigma_2(X)^{-1} \rho_2,
    \end{align*}
    where the last equality follows from the fact that $\tilde{\rho}_2$ augments $\rho_2$ by including an additional moment for $CATT(g,t'',\cdot)$ at the end. This proves that the efficiency bound remains unchanged when we include an additional moment based on an irrelevant $CATT(g,t'',\cdot), t'' \ne t$. We can then sequentially include further moments for irrelevant CATTs, and, by induction, show that these extra moments leave the efficient influence function unchanged. This completes the second part of the proof.
\end{proof}

\begin{proof}[Proof of Corollary \ref{cor:2-period}]
    In the exactly identified model under Assumption \ref{asm:pt-post}, there is only one influence function. Therefore, the weight $\frac{\mathbf{1}'V_{gt}(X)^{-1}}{\mathbf{1}'V_{gt}(X)^{-1} \mathbf{1} }$ is equal to one. The efficient influence function coincides with the only influence function.
\end{proof}

\begin{proof}[Proof of Theorem \ref{thm:efficiency}]
    We first orthogonalize the unconditional moment conditions by replacing $G_g$ with $p_g(X)$ as in the proof of Theorem \ref{thm:efficiency-fixed-g}. The orthogonalized unconditional moments become the following: for each $g \in \mathcal{G}_{\text{trt}}$, 
    \begin{align*}
        \mathbb{E}[\pi_g - p_g(X)] & = 0, \\
        \mathbb{E}[p_g(X) (ATT(g,t) - CATT(g,t,X))] & = 0, g \leq t \leq T. 
    \end{align*}
    We reuse the notations $\rho_1,\rho_2,m_1,m_2,\Sigma_1,\Sigma_2$ in the proof of Theorem \ref{thm:efficiency-fixed-g} to denote the corresponding terms in the set of moment restrictions in Lemma \ref{lm:2}. Notice that the parameters $\pi_g$ and $ATT(g,t)$ are just identified by the unconditional moments in $\rho_1$. In particular, each parameter only appears in a single moment equation. Therefore, we can derive the efficient influence function for each parameter separately.\footnote{Such claims can be verified more formally using the induction approach in the second part of the proof for Theorem \ref{thm:efficiency-fixed-g}. These proofs are omitted to avoid repetition.} The following proof is divided into two parts: (i) the efficient influence function for the group probability $\pi_g$, and (ii) the efficient influence function for the treatment effect $\mathrm{ATT}(g,t)$.

    \paragraph{EIF for $\pi_g$} For a single $\pi_g$, its efficient influence function can be derived equivalently using the following model
    \begin{align*}
    \mathbb{E}[\pi_g - p_g(X)] & = 0, \\
        \mathbb{E}\left[ CATT(g,t,X) - \frac{G_{g}(Y_{t} - Y_1)}{p_{g}(X)} + \frac{G_\infty(Y_{t} - Y_1)}{p_\infty(X)} \Big| X \right] & = 0, {g} \leq t \leq T, g \in \mathcal{G}_{\text{trt}}, \\
        \mathbb{E}\left[\frac{G_{g'}(Y_{t'} - Y_1)}{p_{g'}(X)} - \frac{G_\infty(Y_{t'} - Y_1)}{p_\infty(X)}\Big| X \right] & = 0, 2 \leq t' \leq {g'}-1, {g'} \in \mathcal{G}_{\text{trt}}, \\
        \mathbb{E}[G_{g'} - p_{g'}(X)|X]&=0, {g'} \in \mathcal{G}_{\text{trt}}, \\
        \mathbb{E}[G_\infty - p_\infty(X)|X]&=0.
    \end{align*}
    Notice that the second line contains only definitions of CATT instead of restrictions. Since the CATTs are not involved in the definition of $\pi_g$, we can remove them from the model without affecting the calculation of the efficient influence function. Then the model becomes
    \begin{align*}
    \mathbb{E}[\pi_g - p_g(X)] & = 0, \\
        \mathbb{E}\left[\frac{G_{g'}(Y_{t'} - Y_1)}{p_{g'}(X)} - \frac{G_\infty(Y_{t'} - Y_1)}{p_\infty(X)}\Big| X \right] & = 0, 2 \leq t' \leq {g'}-1, {g'} \in \mathcal{G}_{\text{trt}}, \\
        \mathbb{E}[G_{g'} - p_{g'}(X)|X]&=0, {g'} \in \mathcal{G}_{\text{trt}}, \\
        \mathbb{E}[G_\infty - p_\infty(X)|X]&=0.
    \end{align*}
    In this model, the parameters are $\theta = \pi_g$ and $h=(p_g, g \in \mathcal{G})$.
    The derivatives of $m_1$ and $m_2$ with respect to $\pi_g$ are
    \begin{align*}
    \frac{dm_1(\alpha^*)}{d \pi_g} = 1,~ \frac{dm_2(X,\alpha^*)}{d \pi_g} = 0.
\end{align*}
    The derivatives of $m_1$ and $m_2$ with respect to $h$ are 
    \begin{align*}
        \frac{dm_1(\alpha^*)}{dh}[r] & = \mathbb{E}[-\bm{e}_g' r(X)], \\
        \frac{dm_2(X,\alpha^*)}{dh}[r] & = L(X) r(X),
    \end{align*}
    where $\bm{e}_g$ is the one-hot vector with the $g$th entry being 1 and the remaining entries being zero, and $L(X)$ is defined by
    \begin{align*}
        L(X) \equiv 
        \begin{pmatrix}
        \begin{pmatrix}
        \begin{pmatrix}
            -\frac{m_{g',2,1}(X)}{p_{g'}(X)} \\
            \vdots \\
            -\frac{m_{g',g'-1,1}(X)}{p_{g'}(X)}
        \end{pmatrix} \bm{e}_{g'}',
        \begin{pmatrix}
            \frac{m_{\infty,2,1}(X)}{p_\infty(X)} \\
            \vdots \\
            \frac{m_{\infty,g'-1,1}(X)}{p_\infty(X)}
        \end{pmatrix}
        \end{pmatrix}_{g' \in \mathcal{G}_{\text{trt}}} \\
            -I
        \end{pmatrix},
    \end{align*}
    where $I$ denotes the identity matrix (of dimension $|\mathcal{G}|$). Similar to the proof of Theorem \ref{thm:efficiency-fixed-g}, we follow Theorem 2.1 of \cite{Ai_Chen_2012} and solve the following minimization problem:
    \begin{align*}
    \inf_{r} \quad &  (1 + \mathbb{E}[\bm{e}_{g}'r(X)]) \Sigma_1^{-1}(1 + \mathbb{E}[\bm{e}_g'r(X)]) + \mathbb{E}[ (L(X)r(X))'\Sigma_2(X)^{-1}L(X)r(X) ].
    \end{align*}
    The corresponding solution is given by
    \begin{align*}
    r^*(X) = \frac{-B(X)^{-1}\bm{e}_g}{\mathbb{E}[\bm{e}_g'B(X)^{-1}\bm{e}_g] + \Sigma_1},
    \end{align*}
    with $B(X) = L(X)'\Sigma_2(X)^{-1}L(X)$. The efficient score of $\pi_g$ is given by
    \begin{align*}
        & - \left( \frac{dm_1(\alpha^*)}{d \pi_g} - \frac{dm_1(\alpha^*)}{d h}[r^*] \right)' \Sigma_1^{-1} \rho_1 - \left( \frac{dm_2(X,\alpha^*)}{d \pi_g} - \frac{dm_2(X,\alpha^*)}{d h}[r^*] \right)' \Sigma_2(X)^{-1} \rho_2 \\
        = & - (\mathbb{E}[\bm{e}_g'B(X)^{-1}\bm{e}_g] + \Sigma_1)^{-1} (\rho_1 + \bm{e}_g'B(X)^{-1}L(X)'\Sigma_2(X)^{-1}\rho_2 ).
    \end{align*}
    The semiparametric efficiency bound is hence equal to $(\mathbb{E}[\bm{e}_g'B(X)^{-1}\bm{e}_g] + \Sigma_1)^{-1}$. Therefore, the efficient influence function for $\pi_g$ is equal to
    \begin{align*}
        - (\rho_1 + \bm{e}_g'B(X)^{-1}L(X)'\Sigma_2(X)^{-1}\rho_2 ).
    \end{align*}
    The analysis of $L(X)'\Sigma_2(X)^{-1}$ and $B(X)$ is similar to the proof of Theorem \ref{thm:efficiency-fixed-g}. Define the bottom right $|\mathcal{G}| \times |\mathcal{G}|$ submatrix of $\Sigma_2(X)$ as
    \begin{align*}
    \Pi \equiv 
    \begin{pmatrix}
        p_2 (1-p_2) & - p_3 p_2 & -p_4 p_2 & \cdots & - p_\infty p_2 \\
        - p_3 p_2 & p_3 (1-p_3) & - p_4 p_3  & \cdots & - p_\infty p_3 \\ 
        \vdots & \vdots & \vdots & & \vdots \\
        -p_{\infty} p_2 & -p_{\infty} p_3 & -p_{\infty} p_4 & \cdots & p_\infty (1-p_\infty) \\
    \end{pmatrix}.
    \end{align*}
    We can decompose the covariance matrix $\Sigma_2(X)$ into the following block matrix:
\begin{align*}
    \Sigma_2 = 
    \begin{pmatrix}
        \Sigma_{2,UL} & \Sigma_{2,LL}' \\
        \Sigma_{2,LL} & \Pi 
    \end{pmatrix}.
\end{align*}
Here $\Sigma_{2,LL}$ is a matrix with $|\mathcal{G}|$ rows. Each row of $\Sigma_{2,LL}$ is $(\sigma_{g,g',t'}(X), 2 \leq t' \leq g'-1, g' \in \mathcal{G}_{\text{trt}})$, where $\sigma_{g,g',t'}$ is defined as
with 
\begin{align*}
    \sigma_{g,g',t'}(X) & \equiv \mathbb{E}\left[ (G_{g} - p_{g}(X))  \left(\frac{G_{g'}(Y_{t'} - Y_{1})}{p_{g'}(X)} - \frac{G_\infty(Y_{t'} - Y_{1})}{p_\infty(X)} \right) | X \right] \\
    & = 
    \begin{cases}
        -p_{g}(X) (m_{g',t',1}(X) - m_{\infty,t',1}(X)), & \text{ if } g \notin \{g', \infty\}, \\
        (1-p_g(X)) m_{g,t',1}(X) + p_g(X) m_{\infty,t',1}(X), & \text{ if } g = g', \\
        - p_\infty(X) m_{g,t',1}(X) - (1-p_\infty(X)) m_{\infty,t',1}(X) , & \text{ if } g = \infty.
    \end{cases}
    \end{align*}
    Notice that $L(X)$ is related to $\Sigma_2(X)$ as $L(X) \Pi = -(\Sigma_{2,LL},\Pi)'$, which implies that $L(X)' = -(\Pi^{-1}\Sigma_{2,LL}, I)$. Therefore, we have
    \begin{align*}
        L(X)' \Sigma_2(X)^{-1} & = - (\bm{0},\Pi^{-1}), \\
        L(X)' \Sigma_2(X)^{-1} L(X) & = \Pi^{-1}, \\
        B(X)^{-1} & = \Pi, \\
        B(X)^{-1}L(X)' \Sigma_2(X)^{-1} & = -(\bm{0},I).
    \end{align*}
    The efficient influence function of $\pi_g$ is hence equal to
    \begin{align*}
        - (\rho_1 + \bm{e}_g'B(X)^{-1}L(X)'\Sigma_2(X)^{-1}\rho_2 ) = -(\rho_1 - \bm{e}_g'(\bm{0},I)\rho_2) = G_g - \pi_g,
    \end{align*}
    which is the same influence function in a model without any restrictions.

    \paragraph{EIF for $ATT(g,t)$} The efficient influence function of $ATT(g,t)$ can be derived similarly. The corresponding model can be reduced to the following:
    \begin{align*}
    \mathbb{E}[p_g(X) (ATT(g,t) - CATT(g,t,X))] & = 0, \\
        \mathbb{E}\left[ CATT(g,t,X) - \frac{G_{g}(Y_{t} - Y_1)}{p_{g}(X)} + \frac{G_\infty(Y_{t} - Y_1)}{p_\infty(X)} \Big| X \right] & = 0, \\
        \mathbb{E}\left[\frac{G_{g'}(Y_{t'} - Y_1)}{p_{g'}(X)} - \frac{G_\infty(Y_{t'} - Y_1)}{p_\infty(X)}\Big| X \right] & = 0, 2 \leq t' \leq {g'}-1, {g'} \in \mathcal{G}_{\text{trt}}, \\
        \mathbb{E}[G_{g'} - p_{g'}(X)|X]&=0, {g'} \in \mathcal{G}_{\text{trt}}, \\
        \mathbb{E}[G_\infty - p_\infty(X)|X]&=0.
    \end{align*}
    After rotation, the model is equivalently represented as
    \begin{align*}
    \mathbb{E}[p_g(X) (ATT(g,t) - CATT(g,t,X))] & = 0, \\
        \mathbb{E}\left[ CATT(g,t,X) - \frac{G_{g}(Y_{t} - Y_1)}{p_{g}(X)} + \frac{G_\infty(Y_{t} - Y_1)}{p_\infty(X)} \Big| X \right] & = 0, \\
        \mathbb{E}\left[CATT(g,t,X) - \frac{G_{g}(Y_{t} - Y_1)}{p_{g}(X)} + \frac{G_\infty(Y_{t} - Y_{t'})}{p_\infty(X)} + \frac{G_{g'}(Y_{t'} - Y_1)}{p_{g'}(X)} \Big| X \right] & = 0, 2 \leq t' \leq {g'}-1, {g'} \in \mathcal{G}_{\text{trt}}, \\
        \mathbb{E}[G_{g'} - p_{g'}(X)|X]&=0, {g'} \in \mathcal{G}_{\text{trt}}, \\
        \mathbb{E}[G_\infty - p_\infty(X)|X]&=0.
    \end{align*}
    Then this model is essentially the same as the one studied in the proof of Theorem \ref{thm:efficiency-fixed-g}. Following the same steps, we can show that the efficient influence function is obtained by optimally weighting the influence functions in $\mathbb{IF}(ATT(g,t))$.
\end{proof}

We introduce some definitions. Recall that $r_{g,g'} \coloneqq p_g/p_{g'}$ for any $g,g' \in \mathcal{G}$, and denote $r \equiv (r_{g,g'},g,g' \in \mathcal{G})$. Denote $\mathcal{H}_w$, $\mathcal{H}_m$, and $\mathcal{H}_p$ as the nuisance parameter spaces containing respectively the true values of $w$, $m$, and $r$ and their estimates. For a generic $\mathcal{H}$ and norm $\lVert \cdot \rVert$, the covering number $N(\epsilon,\mathcal{H},\lVert \cdot \rVert)$ is the minimal number of $N$ for which there exist $\epsilon$-balls $\{  \{f: \lVert f-h_j \rVert \leq \epsilon \}, \lVert h_j \rVert < \infty, j=1,\cdots,N \}$ to cover $\mathcal{H}$.

\begin{assumption} \label{asm:regularity} The following regularity assumptions are imposed for Theorem \ref{thm:asy_properties_estimator}.
    \begin{enumerate}[(1)]
        \item Second moment: Each outcome $Y_t$ has finite second moment.
        \item Proper weighting: The estimated weights $\hat{w}^{att(g,t)}$ sum to one and is bounded in probability, i.e., $\lVert \hat{w}^{att(g,t)}\rVert_\infty = O_p(1)$.
        \item Donsker property: For each $j=w,m,r$, we have 
        \begin{align*}
            \int_0^\infty \sup_{Q} \sqrt{\log N(\epsilon \lVert H \rVert_{L_2(Q)} ,\mathcal{H}_j,\lVert\cdot\rVert_{L_2(Q)})} d\epsilon < \infty,
        \end{align*}
        where the supremum is taken over all finitely discrete probability measures $Q$ on the support of $X$, $H$ denotes an envelope of $\mathcal{H}$, and $L_2(Q)$ denotes the $L_2$ measure under $Q$.
        \item Overlap: For each $g,g' \in \mathcal{G}_{\text{trt}} \times \mathcal{G}$, $\mathbb{E}[r_{g,g'}(X)^2] < \infty$.
        \item Uniform consistency: $\lVert \hat{w}^{att(g,t)} - w^{att(g,t)} \rVert_\infty =o_p(1)$, $\lVert \hat{m} -m \rVert_\infty =o_p(1)$, and $\lVert \hat{r} -r \rVert_\infty =o_p(1)$, where $\lVert \cdot \rVert_\infty$ denotes the sup norm.
        \item Rate requirement:
        \begin{align*}
    \left\lVert \hat{m}_{\infty,t,t''} - m_{\infty,t,t''} \right\rVert_{L_2(X)} \left\lVert \hat{r}_{g,\infty} - r_{g,\infty}  \right\rVert_{L_2(X)} & = o_p(n^{-1/2}), t'' < t, \\
    \left\lVert \hat{m}_{g',t'',1} - m_{g',t'',1} \right\rVert_{L_2(X)} \left\lVert \hat{r}_{g,g'} - r_{g,g'} \right\rVert_{L_2(X)} & = o_p(n^{-1/2}), t'' < g', g' \in \mathcal{G}_{\text{trt}},
    \end{align*}
    where $\lVert \cdot \rVert_{L_2(X)}$ represents the $L_2$ norm under the marginal distribution of $X$.
    \end{enumerate}
\end{assumption}
A popular nonparametric class that may satisfy the Donsker condition is the smoothness class defined in Theorem 2.7.1 in \cite{VanderVaart1996}.

\begin{proof}[Proof of Theorem \ref{thm:asy_properties_estimator}]
To simplify the notation in the proof, we focus on a single ATT's estimation and drop the superscript $\text{att(g,t)}$ and subscript $\text{stg}$. We make explicit the dependence of $\theta$ on the nuisance parameters by writing it as $\theta(W;p,m;\pi)$. The ATT estimator is now written as $\widehat{ATT} = \mathbb{E}_n[\hat{w}(X)' \theta(W;\hat{p},\hat{m};\hat{\pi})]$. Define the infeasible estimator constructed using the true nuisance parameters (instead of their estimators) as
\begin{align*}
    \widetilde{ATT} \equiv \mathbb{E}_n[w(X)'\theta(W;\hat{p},\hat{m};\hat{\pi})].
\end{align*}
Our goal is to show that $\widehat{ATT}$ and $\widetilde{ATT}$ are first-order equivalent, i.e., $\sqrt{n}(\widehat{ATT} - \widetilde{ATT}) = o_p(1)$. Once this is established, the influence function of $\widehat{ATT}$ will be the same as that of $\widetilde{ATT}$. Since $\widetilde{ATT}$ is a ratio between two sample averages, its influence function is straightforwardly obtained by using the delta method, which is equal to the efficient influence function specified in Theorem \ref{thm:efficiency}. Then the asymptotic distribution is obtained by using the central limit theorem under the assumption that the second moment of the efficient influence functions exists. The efficiency of $\widehat{ES}(e)$ follows from another use of the delta method. In the remaining part of the proof, we focus on establishing the first-order equivalence between $\widehat{ATT}$ and $\widetilde{ATT}$.

The term $\hat{\pi}_g$ in the denominator has no impact on the asymptotic convergence of $\sqrt{n}(\widehat{ATT} - \widetilde{ATT})$ given that $\pi_g>0$. Therefore, we treat $\hat{\pi}_g$ as one and simply write $\theta(W_i;\hat{p},\hat{m})$ instead of $\theta(W_i;\hat{p},\hat{m};1)$. The difference $\widehat{ATT} - \widetilde{ATT}$ can be decomposed as
    \begin{align*}
        \widehat{ATT} - \widetilde{ATT} & = \frac{1}{n} \sum_{i=1}^n \hat{w}(X_i)' \theta(W_i;\hat{p},\hat{m}) - w(X_i)'\theta(W_i;p,m) \\
        & = \underbrace{\frac{1}{n} \sum_{i=1}^n (\hat{w}(X_i) - w(X_i))' \theta(W_i;p,m)}_{\equiv E_1} - \underbrace{\frac{1}{n} \sum_{i=1}^n \hat{w}(X_i)' (\theta(W_i;\hat{p},\hat{m}) - \theta(W_i;p,m))}_{\equiv E_2}.
    \end{align*}
    Denote the generic entries of $w$, $\hat{w}$, and $\theta$ by using the subscript $g',t''$.
    Recall that $\theta$ is written as
    \begin{align*}
        \theta_{g',t''}(W;\hat{p},\hat{m}) 
        & = G_g(Y_t - Y_{1} - m_{g,t,1}(X)) - \frac{p_g(X)}{p_\infty(X)}G_{g}(Y_t - Y_{t'} - m_{\infty,t,t''}(X)) \\
        & - \frac{p_g(X)}{p_{g'}(X)}G_{g'}(Y_{t''} - Y_{1} - m_{g',t'',1}(X)) + G_g\underbrace{(m_{g,t,1}(X) - m_{\infty,t,t''}(X) - m_{g',t'',1}(X))}_{ = CATT(g,t,X)}.
    \end{align*}
    Notice that the first and last terms on the right-hand side, $G_g(Y_t - Y_{1} - m_{g,t,1}(X))$ and $G_g CATT(g,t,X)$, do not depend on $g'$ or $t''$. These two terms will not contribute to $E_1$ since both $\hat{w}$ and $w$ sum to one, and hence
    \begin{align*}
        (\hat{w}(X_i) - w(X_i))' \mathbf{1} (G_{g,i}CATT(g,t,X_i)) & = 0, \\
        (\hat{w}(X_i) - w(X_i))' \mathbf{1} (G_{g,i}(Y_{i,t} - Y_{i,1} - m_{g,t,1}(X_i))) & = 0.
    \end{align*}
    Therefore, the term $E_1$ is equal to the sum of the following two terms
    \begin{align*}
        & \frac{1}{n} \sum_{g',t''} \sum_{i=1}^n (\hat{w}_{g',t''}(X_i) - w_{g',t''}(X_i))'\frac{p_g(X_i)}{p_\infty(X_i)}G_{g,i}(Y_{i,t} - Y_{i,t''} - m_{\infty,t,t''}(X_i)), \\
        & \frac{1}{n} \sum_{g',t''} \sum_{i=1}^n (\hat{w}_{g',t''}(X_i) - w_{g',t''}(X_i))'\frac{p_g(X_i)}{p_{g'}(X_i)}G_{g',i}(Y_{i,t''} - Y_{i,1} - m_{g',t'',1}(X_i)).
    \end{align*}
    Given that the convergence rate of the two terms can be derived using the same method, we choose to illustrate the convergence of the second term. Additionally, since the summation over $g'$ and $t''$ is finite, it suffices to analyze the convergence of a single term within the summation:
    \begin{align*}
        \frac{1}{n} \sum_{i=1}^n (\hat{w}_{g',t''}(X_i) - w_{g',t''}(X_i))'\frac{p_g(X_i)}{p_{g'}(X_i)}G_{g',i}(Y_{i,t''} - Y_{i,1} - m_{g',t'',1}(X_i)).
    \end{align*}
    By the uniform consistency of $\hat{w}$, the above term multiplied by $\sqrt{n}$ is bounded by
    \begin{align*}
        & \left| \frac{1}{\sqrt{n}} \sum_{i=1}^n (\hat{w}_{g',t''}(X_i) - w_{g',t''}(X_i))'\frac{p_g(X_i)}{p_{g'}(X_i)}G_{g',i}(Y_{i,t''} - Y_{i,1} - m_{g',t'',1}(X_i)) \right| \\
        \leq & \sup_{\substack{\tilde{w}_{g',t''} \in \mathcal{C}^{\alpha}(\mathcal{X})_M:\\ \lVert \tilde{w}_{g',t''} - w_{g',t''} \rVert_\infty < \delta_n}}
 \left| \frac{1}{\sqrt{n}} \sum_{i=1}^n (\tilde{w}_{g',t''}(X_i) - w_{g',t''}(X_i))'\frac{p_g(X_i)}{p_{g'}(X_i)}G_{g',i}(Y_{i,t''} - Y_{i,1} - m_{g',t'',1}(X_i)) \right| + o_p(1),
    \end{align*}
    where $\delta_n \downarrow 0$ denotes a sequence that converges to zero slower than the uniform convergence rate of $\hat{w}$.
    The first term on the right-hand side is the standard stochastic equicontinuity term, which is of order $o_p(1)$ because of the Donsker condition, Theorem 2.5.2 in \cite{VanderVaart1996}, and that $\frac{p_g(X)}{p_{g'}(X)}G_{g'}(Y_{t''} - Y_{1} - m_{g',t'',1}(X))$ is a fixed function with finite second moment by assumption. Following the same procedure, we can show that 
\begin{align*}
    \frac{1}{n} \sum_{g',t'} \sum_{i=1}^n (\hat{w}_{g',t''}(X_i) - w_{g',t''}(X_i))'\frac{p_g(X_i)}{p_\infty(X_i)}G_{g,i}(Y_{i,t} - Y_{i,t''} - m_{\infty,t,t''}(X_i)) = o_p(n^{-1/2}),
\end{align*}
Therefore, we have shown that $E_1 = o_p(n^{-1/2})$. For the term $E_2$, define $\theta^1$ and $\theta^2$ as vectors respectively collecting the following term:
\begin{align*}
        \theta_{g',t''}^1(W;p,m) & \equiv \frac{G_g p_\infty(X) - G_\infty p_g(X)}{\pi_g p_\infty(X)} (Y_t - Y_{t'} - m_{\infty,t,t''}(X)), \\
        \theta_{g',t''}^2(W;p,m) & \equiv \frac{G_g p_{g'}(X) - G_{g'} p_g(X)}{\pi_g p_{g'}(X)} (Y_{t'} - Y_{1} -m_{g',t'',1}(X)).
    \end{align*}
We can decompose $E_2$ as 
\begin{align*}
    E_2 = \underbrace{\frac{1}{n} \sum_{i=1}^n \hat{w}(X_i)' (\theta^1(W_i;\hat{p},\hat{m}) - \theta^1(W_i;p,m))}_{\equiv E_{2.1}} + \underbrace{\frac{1}{n} \sum_{i=1}^n \hat{w}(X_i)' (\theta^2(W_i;\hat{p},\hat{m}) - \theta^2(W_i;p,m))}_{\equiv E_{2.2}}.
\end{align*}
The two terms can be analyzed analogously. We examine $E_{2.1}$ first. This term can be decomposed into three terms 
\begin{align*}
    E_{2.1.1} & \equiv \frac{1}{n} \sum_{i=1}^n \sum_{g',t''} \hat{w}_{g',t''}(X_i) \left( r_{g,\infty}(X_i) - \hat{r}_{g,\infty}(X_i) \right) G_{\infty,i}(Y_{i,t} - Y_{i,t''} - m_{\infty,t,t''}(X_i)), \\
    E_{2.1.2} & \equiv \frac{1}{n} \sum_{i=1}^n \sum_{g',t''} \hat{w}_{g',t''}(X_i) \left( m_{\infty,t,t''}(X_i) - \hat{m}_{\infty,t,t''}(X_i) \right)\left(G_{g,i} - \frac{p_g(X_i)}{p_\infty(X_i)}G_{\infty,i}\right), \\
    E_{2.1.3} & \equiv \frac{1}{n} \sum_{i=1}^n \sum_{g',t'} \hat{w}_{g',t'}(X_i) G_{\infty,i}\left( m_{\infty,t,t'}(X_i) - \hat{m}_{\infty,t,t'}(X_i) \right) \left( r_{g,\infty}(X_i) - \hat{r}_{g,\infty}(X_i) \right).
\end{align*}
The convergence rate of $E_{2.1.1}$ and $E_{2.1.2}$ can be derived in the same way as that of $E_{1}$ under the smoothness and uniform consistency conditions. 
For the last term $E_{2.1.3}$, examine a single term in the summation over $j$:
\begin{align*}
    & \left| \frac{1}{n} \sum_{i=1}^n \hat{w}_{g',t''}(X_i) G_{\infty,i}\left( m_{\infty,t,t''}(X_i) - \hat{m}_{\infty,t,t''}(X_i) \right) \left( r_{g,\infty}(X_i) - \hat{r}_{g,\infty}(X_i) \right) \right| \\
    \leq & O_p(1) \times \frac{1}{n} \sum_{i=1}^n | m_{\infty,t,t''}(X_i) - \hat{m}_{\infty,t,t''}(X_i) | \left| r_{g,\infty}(X_i) - \hat{r}_{g,\infty}(X_i) \right|
\end{align*}
because the estimated weights are bounded in probability. The second factor on the right-hand side is bounded as
\begin{align*}
    & \frac{1}{n} \sum_{i=1}^n | m_{\infty,t,t''}(X_i) - \hat{m}_{\infty,t,t''}(X_i) | \left| r_{g,\infty}(X_i) - \hat{r}_{g,\infty}(X_i) \right| \\
    \leq & \sup_{(\tilde{m}_{\infty,t,t''},\tilde{r}_{g,\infty}) \in \mathcal{F}_n}  \frac{1}{n} \sum_{i=1}^n \Bigg( | m_{\infty,t,t''}(X_i) - \tilde{m}_{\infty,t,t''}(X_i) | \left| r_{g,\infty}(X_i) - \tilde{r}_{g,\infty}(X_i) \right| \\
    & \quad - \mathbb{E}|m_{\infty,t,t''}(X) - \tilde{m}_{\infty,t,t''}(X)| \left| r_{g,\infty}(X) - \tilde{r}_{g,\infty}(X) \right| \Bigg) \\
    & + \sup_{(\tilde{m}_{\infty,t,t''},\tilde{r}_{g,\infty}) \in \mathcal{F}_n} \mathbb{E} \left[|m_{\infty,t,t''}(X) - \tilde{m}_{\infty,t,t''}(X)| \left| r_{g,\infty}(X) - \tilde{r}_{g,\infty}(X) \right| \right],
\end{align*}
where $\mathcal{F}_n \coloneqq \big\{ \tilde{m}_{\infty,t,t''},\tilde{r}_{g,\infty} \in \mathcal{C}^{\alpha}(\mathcal{X})_M:\lVert m_{\infty,t,t''} - \tilde{m}_{\infty,t,t''} \rVert_\infty \leq \delta_n, \lVert m_{\infty,t,t''} - \tilde{m}_{\infty,t,t''} \rVert_{L_2(X)} \leq \lVert m_{\infty,t,t''} - \hat{m}_{\infty,t,t''} \rVert_{L_2(X)}, \lVert r_{g,\infty} - \tilde{r}_{g,\infty} \rVert_\infty \leq \delta_n, \lVert r_{g,\infty} - \tilde{r}_{g,\infty} \rVert_{L_2(X)} \leq \lVert r_{g,\infty} - \hat{r}_{g,\infty} \rVert_{L_2(X)} \big\}$, with $\delta_n \downarrow 0$ being a sequence that converges to zero slower than the uniform convergence rates of $\hat{m}_{\infty,t,t''}$ and $\hat{r}_{g,\infty}$.
By the fact that the finiteness of entropy integral is preserved under element-wise multiplication of function classes, we can use the previous stochastic equicontinuity argument to show that the first term on the right-hand side is of order $o_p(n^{-1/2})$. The second term on the right-hand side is also of order $o_p(n^{-1/2})$ by using Cauchy-Schwarz inequality together with the rate requirement on the nuisance estimators:
\begin{align*}
    & \sup_{(\tilde{m}_{\infty,t,t''},\tilde{r}_{g,\infty}) \in \mathcal{F}_n} \mathbb{E}\left[|m_{\infty,t,t''}(X) - \tilde{m}_{\infty,t,t''}(X)| \left| r_{g,\infty(X)} - \tilde{r}_{g,\infty}(X) \right|\right] \\
    \leq & \sup_{(\tilde{m}_{\infty,t,t''},\tilde{r}_{g,\infty}) \in \mathcal{F}_n} \lVert m_{\infty,t,t''} - \tilde{m}_{\infty,t,t''}\rVert_{L_2(X)} \left\lVert r_{g,\infty} - \tilde{r}_{g,\infty} \right\rVert_{L_2(X)} = o_p(n^{-1/2}).
\end{align*}
This proves that $E_{2.1} = o_p(n^{-1/2})$. The term $E_{2.2}$ can also be decomposed in a similar way as $E_{2.1}$:
\begin{align*}
    E_{2.2.1} & \equiv \frac{1}{n} \sum_{i=1}^n \sum_{g',t''} \hat{w}_{g',t''}(X_i) \left( r_{g,g'}(X_i) - \hat{r}_{g,g'}(X_i) \right) G_{{g'},i}(Y_{i,t''} - Y_{i,1} - m_{{g'},t'',1}(X_i)), \\
    E_{2.2.2} & \equiv \frac{1}{n} \sum_{i=1}^n \sum_{g',t''} \hat{w}_{g',t''}(X_i) \left( m_{{g'},t'',1}(X_i) - \hat{m}_{{g'},t'',1}(X_i) \right)\left(G_{g,i} - \frac{p_g(X_i)}{p_{g'}(X_i)}G_{{g'},i}\right), \\
    E_{2.2.3} & \equiv \frac{1}{n} \sum_{i=1}^n \sum_{g',t''} \hat{w}_{g',t''}(X_i) G_{{g'},i}\left( m_{{g'},t'',1}(X_i) - \hat{m}_{{g'},t'',1}(X_i) \right) \left( r_{g,g'}(X_i) - \hat{r}_{g,g'}(X_i) \right).
\end{align*}
The convergence rate of $E_{2.2.1}$ and $E_{2.2.2}$ can be derived in the same way as for $E_{2.1.1}$, $E_{2.1.2}$, and $E_{1}$. The term $E_{2.2.3}$ is similar to $E_{2.1.3}$.
This proves that $E_{2.1} = o_p(n^{-1/2})$. In the end, we have shown that $\widehat{ATT} - \widetilde{ATT}= o_p(n^{-1/2})$, and therefore, $\widehat{ATT}(g,t)$ and $\widehat{ES}(e)$ have the desired asymptotic distributions.
    
\end{proof}

The following lemma shows the consistency of the estimator based on the $\hat{K}$ in (\ref{eqn:series_ratio_pscore}), i.e., asymptotically all basis functions are selected. Let $\ell(W,r) = r(X)^2 G_{g'} - 2r(X)G_g$ and $\Bar{\ell}(r) = \mathbb{E}[\ell(W,r)]$. Denote the sieve space of $r$ as $\mathcal{R}_K$. For example, in (\ref{eqn:series_ratio_pscore}), $\mathcal{R}_K$ is the space spanned by $\Psi^K(X)$.

\begin{lemma} \label{lm:AIC}
    Let the following conditions hold.
    \begin{enumerate}[(1)]
    \item There is a sequence $c(K)$ such that $\inf_{\tilde{r} \in \mathcal{R}_K} [\Bar{\ell}(\tilde{r}) - \Bar{\ell}(r)] > c(K) >0$ for all $K$.
    \item $\sup_{\tilde{r} \in \mathcal{R}_K}|\frac{1}{n} \sum_{i=1}^n (\ell(W_i,\tilde{r})-\Bar{\ell}(\tilde{r}))| = o_p(1)$.
    \item There is a sequence $K_n^* \rightarrow \infty$ such that $\lVert \hat{r}_{K_n^*} - r\rVert_\infty = o_p(1)$.
    \item For any sequence $\eta_n = o(1)$, $\sup_{\lVert \tilde{r}  - r\rVert_\infty \leq \eta_n}|\frac{1}{n} \sum_{i=1}^n (\ell(W_i,\tilde{r})-\Bar{\ell}(r))| = o_p(1)$.
    \item $C_n K_n^*/n = o(1)$.
    \end{enumerate}
    Then for any fixed $K^*$, we have $\hat{K} > K^*$ with probability approaching one.
\end{lemma}

\begin{proof}[Proof of Lemma \ref{lm:AIC}]
For simplicity, we suppress the subscripts $g$ and $g'$ and denote the estimand and its estimator as $r$ and $\hat{r}_K$, where the index $K$ signifies the sieve dimension. Define the objective function in (\ref{eqn:series_ratio_pscore}) as $I_{n}(K)$. For any fixed $K^*$, using conditions (2) and (5), we have
    \begin{align*}
        I_n(K^*) & = 2 \Bar{\ell}(\hat{r}_{K^*}) + \frac{2}{n} \sum_{i=1}^n (\ell(W_i,\hat{r}_{K^*})-\Bar{\ell}(\hat{r}_{K^*})) + C_nK^*/n \\
        & = 2 \Bar{\ell}(\hat{r}_{K^*}) + o_p(1).
    \end{align*}
    On the other hand, using conditions (2)-(4) for the sequence $K_n^*$, we have 
    \begin{align*}
        I_n(K_n^*) & = 2 \Bar{\ell}(r) + \frac{2}{n} \sum_{i=1}^n (\ell(W_i,\hat{r}_{K_n^*})-\Bar{\ell}(r)) + C_nK^*/n \\
        & = 2 \Bar{\ell}(r) + o_p(1).
    \end{align*}
    Combining the above two results, we have
    \begin{align*}
        I_n(K^*) - I_n(K_n^*) = 2 (\Bar{\ell}(\hat{r}_{K^*}) - \Bar{\ell}(r)) + o_p(1),
    \end{align*}
    which, together with condition (1), implies that $I_n(K^*) - I_n(K_n^*) >0$ with probability approaching one. This completes the proof.
\end{proof}

\begin{proof}[Proof of Lemma \ref{lm:LATT-id}]
    Define the conditional LATT parameter as
    \begin{align*}
        \textit{CLATT}(g,t,X) \equiv \mathbb{E}[Y_t(1) - Y_t(0)|G^{\text{IV}}=g,D_t(g) > D_t(\infty),X].
    \end{align*}
    We first want to establish that 
    \begin{align*}
        \textit{CLATT}(g,t,X) = \frac{\mathbb{E}[Y_t - Y_{g-1} | G^{\text{IV}}=g, X] - \mathbb{E}[Y_t - Y_{g-1} | G^{\text{IV}}=\infty, X]}{\mathbb{E}[D_t - D_{g-1} | G^{\text{IV}}=g, X] - \mathbb{E}[D_t - D_{g-1} | G^{\text{IV}}=\infty, X]}.
    \end{align*}
    For the numerator, notice that 
    \begin{align*}
        & \mathbb{E}[Y_t - Y_{g-1} | G^{\text{IV}}=g, X] - \mathbb{E}[Y_t - Y_{g-1} | G^{\text{IV}}=\infty, X] \\
        = & \mathbb{E}[Y_t(D_t(g)) - Y_{g-1}(D_t(g)) | G^{\text{IV}}=g, X] - \mathbb{E}[Y_t(D_t(\infty)) - Y_{g-1}(D_t(\infty)) | G^{\text{IV}}=\infty, X] \\
        = & \mathbb{E}[Y_t(D_t(g)) - Y_{t}(D_t(\infty)) | G^{\text{IV}}=g, X] \\
        & + \mathbb{E}[Y_t(D_t(\infty)) - Y_{g-1}(D_t(g))] - \mathbb{E}[Y_t(D_t(\infty)) - Y_{g-1}(D_t(\infty)) | G^{\text{IV}}=\infty, X] \\
        = & \mathbb{E}[Y_t(D_t(g)) - Y_{t}(D_t(\infty)) | G^{\text{IV}}=g, X],
    \end{align*}
    where the last equality follows from no anticipation and parallel trends in the outcome. For the right-hand side, notice that
    \begin{align*}
        \mathbb{E}[Y_t(D_t(g))| G^{\text{IV}}=g, X] & = \mathbb{E}[D_t(g)(Y_t(1) - Y_t(0))| G^{\text{IV}}=g, X] + \mathbb{E}[Y_t(0)| G^{\text{IV}}=g, X], \\
        \mathbb{E}[Y_{t}(D_t(\infty))| G^{\text{IV}}=g, X] & = \mathbb{E}[D_t(\infty)(Y_{t}(1) - Y_{t}(0))| G^{\text{IV}}=g, X] + \mathbb{E}[Y_{t}(0)| G^{\text{IV}}=g, X],
    \end{align*}
    and therefore,
    \begin{align*}
        \mathbb{E}[Y_t(D_t(g)) - Y_{t}(D_t(\infty)) | G^{\text{IV}}=g, X] & = \mathbb{E}[(D_t(g)-D_t(\infty))(Y_{t}(1) - Y_{t}(0))| G^{\text{IV}}=g, X]\\
        & = \mathbb{E}[Y_{t}(1) - Y_{t}(0)| G^{\text{IV}}=g,D_t(g)-D_t(\infty)=1, X] \\
        & \times \mathbb{P}(D_t(g)-D_t(\infty)=1|G^{\text{IV}}=g,X) \\
        & = \textit{CLATT}(g,t,X) \mathbb{P}(D_t(g)-D_t(\infty)=1|G^{\text{IV}}=g,X),
    \end{align*}
    where the second line follows from monotonicity.
    Similarly, we can show that 
    \begin{align*}
        \mathbb{E}[D_t - D_{g-1} | G^{\text{IV}}=g, X] - \mathbb{E}[D_t - D_{g-1} | G^{\text{IV}}=\infty, X] & = \mathbb{E}[D_t(g)-D_t(\infty)| G^{\text{IV}}=g, X] \\
        & = \mathbb{P}(D_t(g)-D_t(\infty)=1|G^{\text{IV}}=g,X).
    \end{align*}
    This gives the identification of CLATT. To identify the unconditional LATT, we need to invoke the Bayes rule.
    \begin{align*}
        &\textit{LATT}(g,t) \\
        = & \mathbb{E}[\textit{CLATT}(g,t,X)|G^{\text{IV}}=g,D_t(g) > D_t(\infty)] \\
        = & \int \frac{\mathbb{E}[Y_t - Y_{g-1} | G^{\text{IV}}=g, x] - \mathbb{E}[Y_t - Y_{g-1} | G^{\text{IV}}=\infty, x]}{\mathbb{E}[D_t - D_{g-1} | G^{\text{IV}}=g, x] - \mathbb{E}[D_t - D_{g-1} | G^{\text{IV}}=\infty, x]} d F_{X|G^{\text{IV}}=g,D_t(g) > D_t(\infty)}(x) \\
        = & \int \frac{\mathbb{E}[Y_t - Y_{g-1} | G^{\text{IV}}=g, x] - \mathbb{E}[Y_t - Y_{g-1} | G^{\text{IV}}=\infty, x]}{\mathbb{E}[D_t - D_{g-1} | G^{\text{IV}}=g, x] - \mathbb{E}[D_t - D_{g-1} | G^{\text{IV}}=\infty, x]} \\
        & \times \frac{\mathbb{P}(D_t(g) > D_t(\infty)|G^{\text{IV}}=g,X=x) \mathbb{P}(G^{\text{IV}}=g|X=x)dF_X(x)}{\mathbb{P}(G^{\text{IV}}=g,D_t(g) > D_t(\infty))} \\
        = & \frac{\mathbb{E}\big[ (\mathbb{E}[Y_t - Y_{g-1} | G^{\text{IV}}=g, X] - \mathbb{E}[Y_t - Y_{g-1} | G^{\text{IV}}=\infty, X]) G^{\text{IV}}_g \big]}{\mathbb{P}(G^{\text{IV}}=g,D_t(g) > D_t(\infty))}.
    \end{align*}
    Lastly, the denominator on the right-hand side can be written as
    \begin{align*}
        \mathbb{P}(G^{\text{IV}}=g,D_t(g) > D_t(\infty)) & = \mathbb{E}[\mathbb{E}[D_t(g) - D_t(\infty)|G^{\text{IV}}=g,X] \mathbb{P}(G^{\text{IV}}=g|X)] \\
        & = \mathbb{E}\big[(\mathbb{E}[D_t - D_{g-1} | G^{\text{IV}}=g, X] - \mathbb{E}[D_t - D_{g-1} | G^{\text{IV}}=\infty, X]) G^{\text{IV}}_g\big].
    \end{align*}
    This completes the proof.
\end{proof}

\begin{proof}[Proof of Lemma \ref{lm:3}]
    Similar to the proof of Lemmas \ref{lm:1} and \ref{lm:2}, we want to construct a joint distribution of the potential outcomes and treatments that is consistent with the observed outcomes and treatments and satisfies the identification assumptions. Recall the notations $Y$ and $\Delta Y$ for the observed outcomes as defined in the proof of Lemmas \ref{lm:1} and \ref{lm:2}. Similarly, define $D$ and $\Delta D$ for the treatment. In the same way, define the vectors of potential variables $Y(g)$, $\Delta Y(g)$, $D(g)$, and $\Delta D(g)$, where the potential outcome $Y_t(g)$ is now shorthand for $Y_t(D_t(g))$. Let $(D,Y,G^{\text{IV}})$ denote the observed variables that already satisfy random sampling, overlap, and monotonicity. For each given $g$, we construct the potential outcomes and treatments as follows:
    \begin{align*}
        & \text{the vectors } (D(g'),Y(g')), g' \in \mathcal{G}^{\text{IV}} \text{ are jointly independent conditional on $G^{\text{IV}}=g$,} \\
        & \text{for $g'=g$: } (D(g),Y(g)) | \{G=g\} \overset{d}{=} (D,Y) | \{G=g\} \text{, and} \\
        & \text{for $g'\ne g$: } (D_1(g'),Y_1(g')) | \{G=g\} \overset{d}{=} (D_1(g),Y_1(g))| \{G=g\}, \\
        & \quad (\Delta D(g'),\Delta Y(g')) | \{G=g\} \overset{d}{=} (\Delta D,\Delta Y) | \{G=\infty\}, (D_1(g'),Y_1(g')) \perp (\Delta D(g'),\Delta Y(g')) | \{G=g\}.
    \end{align*}
    This induces the observed outcomes and treatments through constructions. The no-anticipation and parallel trends conditions can be verified similarly to the proof of Lemma \ref{lm:2}.

\end{proof}

\begin{proof}[Proof of Corollary \ref{cor:iv}]
    It is well-known that shape restrictions such as monotonicity do not affect the calculation of the semiparametric efficiency bound. Thus, we can focus on the moment equalities to derive the bound.
    Similar to the proof of Theorems \ref{thm:efficiency} and \ref{thm:efficiency-fixed-g}, we can separately derive the efficient influence function for $LATT(g,t)_{num}$ and $LATT(g,t)_{den}$ and then combine them to obtain the efficient influence function for $LATT(g,t)_{num}/LATT(g,t)_{den}$. For a single $LATT(g,t)_{num}$, the model becomes the following
    \begin{align*}
        \mathbb{E}[G^{\text{IV}}_g (LATT(g,t)_{num} - h_1(g,t,X))] & = 0, \\
    \mathbb{E}\left[ h_1(g,t,X) - \frac{G^{\text{IV}}_g(Y_{t} - Y_{g-1})}{p^{\text{IV}}_g(X)} + \frac{G^{\text{IV}}_{\infty}(Y_{t} - Y_{g-1})}{p^{\text{IV}}_{\infty}(X)} \Big| X \right] & = 0, \\
    \mathbb{E}\left[\frac{G^{\text{IV}}_{g}(Y_{t'} - Y_1)}{p^{\text{IV}}_{g}(X)} - \frac{G^{\text{IV}}_\infty(Y_{t'} - Y_1)}{p^{\text{IV}}_\infty(X)}\Big| X \right] & = 0, \text{for all } 2 \leq t' \leq g-1, \\
    \mathbb{E}\left[\frac{G^{\text{IV}}_{g}(D_{t'} - D_1)}{p^{\text{IV}}_{g}(X)} - \frac{G^{\text{IV}}_\infty(D_{t'} - D_1)}{p^{\text{IV}}_\infty(X)}\Big| X \right] & = 0, \text{for all } 2 \leq t' \leq g-1, \\
        \mathbb{E}[G^{\text{IV}}_g - p^{\text{IV}}_g(X)|X]&=0.
    \end{align*}
    We can rotate the conditional moment restrictions as in the proof of Theorem \ref{thm:efficiency-fixed-g} and obtain that
    \begin{align*}
        \mathbb{E}\left[ h_1(g,t,X) - \frac{G^{\text{IV}}_g(Y_{t} - Y_{1})}{p^{\text{IV}}_g(X)} + \frac{G^{\text{IV}}_{\infty}(Y_{t} - Y_{1})}{p^{\text{IV}}_{\infty}(X)} \Big| X \right] & = 0, \\
        \cdots \\
        \mathbb{E}\left[ h_1(g,t,X) - \frac{G^{\text{IV}}_g(Y_{t} - Y_{g-1})}{p^{\text{IV}}_g(X)} + \frac{G^{\text{IV}}_{\infty}(Y_{t} - Y_{g-1})}{p^{\text{IV}}_{\infty}(X)} \Big| X \right] & = 0, \\
        \mathbb{E}\left[ h_1(g,t,X) - \frac{G^{\text{IV}}_g(Y_{t} - Y_{1} + D_{2} - D_1)}{p^{\text{IV}}_g(X)} + \frac{G^{\text{IV}}_{\infty}(Y_{t} - Y_{1} + D_{2} - D_1)}{p^{\text{IV}}_{\infty}(X)} \Big| X \right] & = 0, \\
        \cdots \\
        \mathbb{E}\left[ h_1(g,t,X) - \frac{G^{\text{IV}}_g(Y_{t} - Y_{1} + D_{g-1} - D_1)}{p^{\text{IV}}_g(X)} + \frac{G^{\text{IV}}_{\infty}(Y_{t} - Y_{1} + D_{g-1} - D_1)}{p^{\text{IV}}_{\infty}(X)} \Big| X \right] & = 0, \\
        \mathbb{E}[G^{\text{IV}}_g - p^{\text{IV}}_g(X)|X]&=0.
    \end{align*}
    The structure of this set of moment restrictions aligns with that in Theorem \ref{thm:efficiency-fixed-g}. The term \(Y_1 - (D_{t'} - D_1)\) can be treated as an outcome in a baseline period, enabling the derivation of the efficient influence function for \(LATT(g,t)_{num}\) as \(\mathbb{EIF}^{latt(g,t),num}/\pi^{\text{IV}}_g\), where \(\pi^{\text{IV}}_g \equiv \mathbb{P}(G^{\text{IV}}=g)\). Similarly, the efficient influence function for \(LATT(g,t)_{den}\) is obtained as \(\mathbb{EIF}^{latt(g,t),den}/\pi^{\text{IV}}_g\). The efficient influence function for the LATT parameter then follows from the application of the Delta method.
\end{proof}

\begin{proof}[Proof of Theorem \ref{thm:test}]
    This is the well-known Hausman test, in which $\aCov(\widehat{ES}-\widecheck{ES})$ can be consistently estimated by $\widehat{\aCov}(\widecheck{ES}) - \widehat{\aCov}(\widehat{ES})$ because $\widehat{ES}$ achieves the semiparametric efficiency bound. The fact that this Hausman test has nontrivial power against all local alternatives is based on Theorem 3.3 and Remark 3.4 in \citet{Chen_Santos_2018_ECMA}.
\end{proof}

\bibliography{Bibliography.bib}

\end{document}